\documentclass{article}
\usepackage[utf8]{inputenc}
\usepackage{amsmath}
\usepackage{amssymb}
\usepackage{amsthm}
\usepackage{dsfont}
\usepackage{bm}
\usepackage{appendix}
\usepackage{comment}
\usepackage{tablefootnote}
\usepackage{empheq}
\newcommand*\widefbox[1]{\fbox{\hspace{2em}#1\hspace{2em}}}

\title{Analysis of Density Matrix Embedding Theory around the non-interacting limit}
\author{Eric Canc\`es\footnotemark[1], Fabian Faulstich\footnotemark[2], Alfred Kirsch\footnotemark[1], Elo\"ise Letournel\footnotemark[1] and Antoine Levitt\footnotemark[3]}

\def\R{{\mathbb R}}

\def\cD{{\mathcal D}}

\def\cF{{\mathcal F}}
\def\cH{{\mathcal H}}

\def\cY{{\mathcal Y}}

\def\1{{\mathds{1}}}

\newcommand \dps{\displaystyle }
\newcommand{\mat}{\textrm{mat}} % for matrices
\newcommand{\Ran}{\mathrm{Ran}} % for images
\newcommand{\Ker}{\mathrm{Ker}} % for kernels
\newcommand{\Trace}{\mathrm{Tr}} % for traces
\newcommand{\Span}{\mathrm{Span}} % for span of family of vectors
\newcommand{\llbracket}{[\![} %for integers intervals
\newcommand{\rrbracket}{]\!]} %for integers intervals
\newcommand{\HSpace}{\mathcal{H}}
\newcommand{\HSpaceN}[1][n]{\mathcal{H}_{#1}} %N-particle Hilbert space (n by default)
\newcommand{\DimH}{L} %Dimension of \HSpace
\newcommand{\NElec}{N} % Number of electrons
\newcommand{\Fock}{{\rm Fock}} %Fock space
\newcommand{\LinearMap}{\mathcal L}%Linear map of a space
\newcommand{\AtomicVector}[1][\kappa]{e_{{#1}}} %Vector of the atomic basis
\newcommand{\AtomicBasis}{\mathcal{B}_{\mathrm{at}}} %Atomic basis
\newcommand{\Ann}[1][\kappa]{\widehat a_{#1}} %Annihilation generator
\newcommand{\NumberOpe}{\widehat N} %Number operator
\newcommand{\Hamiltonian}{\widehat H} %Hamiltonian
\newcommand{\Partition}{{\rm Part}} % Partition of \llbracket 1, \DimH \rrbracket
\newcommand{\FragIndic}[1][x]{\mathcal{I}_{#1}} % Indices of the x-th set of \Partition
\newcommand{\FragSize}[1][x]{L_{#1}} % Size of the x-th fragment
\newcommand{\NFrag}{N_f} % Number of fragments
\newcommand{\DGS}{D_{0}}%Density matrix of the Ground State (DGS)
\newcommand{\Grass}[1][\NElec]{\mathcal{D}} %Grassmannian 
\newcommand{\CHGrass}[1][\NElec]{\mathrm{CH}({\mathcal{D})}} %Convex hull of the above Grassmannian
\newcommand{\Core}[2][x]{\HSpace_{{#1},{#2}}^\mathrm{core}} %x-th core subspace associated to D
\newcommand{\Virt}[2][x]{\HSpace_{{#1},{#2}}^\mathrm{virt}} % x-th virtual subspace associated to D
\newcommand{\Env}[2][x]{\HSpace_{{#1},{#2}}^\mathrm{env}} % x-th environnment subspace associated to D
\newcommand{\WF}{\Psi} % Symbol for a wave-function
\newcommand{\WFSlater}[1][D]{\WF^{0}_{\NElec,{#1}}} % Wave-function of Slater determinant form, which $D$ is the density matrix of it
\newcommand{\WFimp}{\WF^{0,\mathrm{imp}}_{x, D}} % Wave-function (impurity part) that helps to define $D$ thanks to the x-th decomposition
\newcommand{\WFcore}{\WF^{0, \mathrm{core}}_{x, D}} % Wave-function (core part) that helps to define $D$ thanks to the x-th decomposition
\newcommand{\WFimpTrial}{\WF^{\mathrm{imp}}_{x,D}} % Minimizer variable of the x-th impurity problem
\newcommand{\ImpurityHamiltonian}[2][x]{\Hamiltonian^{\mathrm{imp}}_{{#1},{#2}}} % The x-th impurity Hamiltonian
\newcommand{\ChemicalPotential}{\mu} %The global chemical potential, used for \HLMap
\newcommand{\GSDensity}[1][x]{P_{\ChemicalPotential,{#1}}} % 1-RDM of the regular impurity problem ground-state
\newcommand{\HLMapPartial}[1][x]{F^{\mathrm{HL}}_{\ChemicalPotential,{#1}}} % The x-th part of the high-level solver, with regular impurity
\newcommand{\HLMap}{F^{\mathrm{HL}}} % The high-level solver with regular impurity
\newcommand{\MatFrag}[1][x]{E_{#1}}%The matrix of the atomic orbitals of the fragments in the atomic basis $\AtomicBasis$.
\newcommand{\DiagonalBlockCHGrass}[1][\Partition]{\mathcal{P}} % The set of diagonal blocks (according to $\Partition$, all the other entries are 0) of $\CHGrass$. %NOTE 
\newcommand{\EMF}{{\mathcal E}^{\mathrm{HF}}}%Mean-field Hamiltonian defined for $\LLMap$
\newcommand{\LLMap}{F^{\mathrm{LL}}} %The low-level map.
\newcommand{\Frag}[1][x]{X_{#1}} %The x-th fragment 
\newcommand{\Impurity}[2][x]{W_{{#1},{#2}}}% The x-th impurity associated to $D$
\newcommand{\Projector}[1][E]{\Pi_{#1}}% The orthogonal projecor on $E$, a subspace of $\HSpace$.
\newcommand{\PartitionProj}{\mathrm{Bd}}% The projector defined in \eqref{eq:PartitionProjDefinition}
\newcommand{\DMETMap}{F^{\mathrm{DMET}}}% The DMET map with regular impurity problem
\newcommand{\GrassCompatiblePartition}{\PartitionProj^{-1}\mathring{\DiagonalBlockCHGrass}} %The part of $\Grass$ that is compatible with $\Partition$

\newcommand{\PGS}{P_0}
\usepackage[dvipsnames]{xcolor}
\usepackage{graphicx}
\usepackage{caption}
\usepackage{subcaption}
\definecolor{airforceblue}{rgb}{0.36, 0.54, 0.66}

\newtheorem{theorem}{Theorem}
\newtheorem{remark}[theorem]{Remark}
\newtheorem{lemma}[theorem]{Lemma}
\newtheorem{proposition}[theorem]{Proposition}
\newtheorem{definition}[theorem]{Definition}

\usepackage[symbol]{footmisc}

\begin{document}

\maketitle

\begin{abstract}
  This article provides the first mathematical analysis of the Density Matrix Embedding Theory (DMET) method. 
  We prove that, under certain assumptions, (i) the exact ground-state density matrix is a fixed-point of the DMET map for non-interacting systems, (ii) there exists a unique physical solution in the weakly-interacting regime, and (iii) DMET is exact at first order in the coupling parameter. 
  We provide numerical simulations to support our results and comment on the physical meaning of the assumptions under which they hold true.  
  We show that the violation of these assumptions may yield multiple solutions of the DMET equations.
  We moreover introduce and discuss a specific $\NElec$-representability problem inherent to DMET.
\end{abstract}

\footnotetext[1]{CERMICS, Ecole des Ponts and Inria Paris, 6 \& 8 avenue Blaise Pascal,
77455 Marne-la-Vall\'ee, France}
\footnotetext[2]{Department of Mathematics, Rensselaer Polytechnic Institute, Troy, 12180 NY}
\footnotetext[3]{Laboratoire de Math\'ematiques d'Orsay, Universit\'e Paris-Saclay, Orsay, 91405, France}
        
\setcounter{tocdepth}{2}        

\tableofcontents

\section{Introduction}

Electronic structure theory is a powerful quantum mechanical framework for investigating the intricate behavior of electrons within molecules and crystals. 
At the core lies the interaction between particles, specifically the electron-electron and electron-nuclei interactions.
Embracing the essential quantum physical effects, this theory is the foundation for {\it ab initio} quantum chemistry and materials science calculations performed by many researchers in chemistry and related fields, complementing and supplementing painstaking laboratory work.
With its diverse applications in chemistry and materials science, electronic structure theory holds vast implications for the mathematical sciences. 
Integrating mathematical doctrines into this field leads to the development of precise and scalable numerical methods, enabling extensive {\it in silico} studies of chemistry for e.g. sustainable energy, green catalysis, and nanomaterials. 
The synergy between mathematics and electronic structure theory offers the potential for groundbreaking advancements in addressing these global challenges.

Within the realm of electronic structure theory, the treatment of {\it strongly correlated quantum systems} is a particularly challenging and long-standing challenge.
Here, the application of high-accuracy quantum chemical methods that are able to capture the electronic correlation effects at chemical accuracy is inevitable. 
Unfortunately, the application of such high-accuracy methods is commonly stymied by a steep computational scaling with respect to the system's size.
A potential remedy is provided by quantum embedding theories, i.e., a paradigm for bootstrapping the success of highly accurate solvers at small scales up to significantly larger scales by decomposing the original system into smaller fragments, where each fragment is then solved individually and from which, a solution to the whole system is then obtained~\cite{gordon2012fragmentation,jones2020embedding,sekaran2022local}.   
Such approaches include dynamical mean-field theory~\cite{metzner1989correlated,georges1996dynamical,georges1992numerical,kotliar2006electronic,maier2005quantum}, or variational embedding theory~\cite{lin2022variational,chen2020multiscale,khoo2021scalable}.

Subject of this article is a widely-used quantum embedding theory, namely, density matrix embedding theory (DMET)~\cite{DMET2012, DMET2013, tsuchimochi2015density, bulik2014density, wouters2016practical, cui2019efficient,sun2020finite,cui2020ground}.
The general idea of DMET is to partition the global quantum system into several quantum ``impurities'', each impurity being treated accurately via a high-level theory (such as full configuration interaction (FCI)~\cite{knowles1984new,olsen1990passing,vogiatzis2017pushing}, coupled cluster theory~\cite{Cizek1966}, density matrix renormalization group (DMRG)~\cite{White92}, etc.).
More precisely, the DMET methodology follows the procedure sketched out as: 
1) fragment the system, 
2) for each fragment, construct an interacting bath that describes the
coupling between the fragment and the remaining system, thus giving rise to a so-called impurity problem, 
3) solve an interacting problem for each impurity using a highly accurate method, 
4) extract properties of the system, 
5) perform step 2)--4) self-consistently in order to embed updated correlation effects back into the full system. 
Over the past years, a large variety of this general framework has been developed, including how
the bath space is defined (including the choice of low-level theory)~\cite{fertitta2018rigorous,nusspickel2020efficient,nusspickel2020frequency,yalouz2022quantum}, how the interacting cluster Hamiltonian is constructed and solved~\cite{nusspickel2022effective,potthoff2001two,lu2019natural,ganahl2015efficient,scott2021extending}, and the choice of self-consistent requirements~\cite{wu2019projected,wu2020enhancing,faulstich2022pure}.
This variety of DMET flavors has been successfully applied to a wide range of systems such as Hubbard models~\cite{DMET2012, bulik2014density, chen2014intermediate, boxiao2016, Zheng2017, zheng2017stripe, welborn2016bootstrap,senjean2018site,senjean2019projected}, quantum spin models~\cite{Fan15, gunst2017block,ricke2017performance}, and a number of strongly correlated molecular and periodic systems~\cite{DMET2013, wouters2016practical,cui2020ground,nusspickel2022systematic,bulik2014electron,pham2018can,hermes2019multiconfigurational,tran2019using,ye2018incremental,ye2019bootstrap,ye2019atom,ye2020bootstrap,ye2021accurate,tran2020bootstrap,meitei2023periodic,mitra2022periodic,mitra2021excited,ai2022efficient}. 
Recently, the application of DMET variants on quantum computers has been explored~\cite{liu2023bootstrap,vorwerk2022quantum,cao2022ab}.

In this article, we follow the computational procedure where the global information, at the level of the one-electron reduced density matrix (1-RDM), is made consistent between all the impurities with the help of a low-level Hartree-Fock (HF) type of theory. 
In the self-consistent-field DMET (SCF-DMET)\footnote{Throughout the paper, DMET refers to SCF-DMET. This is in contrast to one-shot DMET, in which the impurity problem is only solved once without self-consistent updates.}, this global information is then used to update the impurity problems in the next self-consistent iteration, until a consistency condition of the 1-RDM is satisfied between the high-level and low-level theories.

This article is organized as follows. In Section~\ref{sec:notations}, we introduce the many-body quantum model under investigation and its fragment decomposition, and set up some notation used in the sequel. In Section~\ref{sec:impurity_pb}, we present a mathematical formulation of the DMET impurity problem and introduce (formally) the high-level DMET map. The low-level DMET map and the DMET fixed point problem are defined (still formally) in Sections~\ref{sec:overview} and \ref{sec:DMET_pb} respectively. In Section~\ref{sec:main_results}, we state our main results:
\begin{enumerate}
\item in Proposition~\ref{prop:DMET0}, we show that for non-interacting systems, the exact ground-state density matrix is a fixed-point of the DMET map if (i) the system is gapped (Assumption (A1)), and (ii) the fragment decomposition satisfies a natural and rather mild condition (Assumption (A2)). Although this result is well-known in the physics and chemistry community, a complete mathematical proof was still missing;
\item in Theorem~\ref{thm:theory}, we prove that under two additional assumptions ((A3) and (A4)), the DMET fixed-point problem has a unique physical solution in the weakly-interacting regime, which is real-analytic in the coupling parameter $\alpha$. Assumption (A3) is related to some specific $\NElec$-representability condition inherent to the DMET approach, while Assumption (A4) has a physical interpretation in terms of linear response theory;
\item in Theorem~\ref{thm:DMET_HF}, we prove that in the weakly-interacting regime, DMET is exact at first order in $\alpha$.
\end{enumerate}
The numerical simulations reported in Section~\ref{sec:numerics} illustrate the above results and indicate that DMET does not seem to be exact at second order. Although, in the special case when there is only one site per fragment, Assumption (A4) is a consequence of Assumptions (A1)-(A3) (see Remark~\ref{rem:one_site_per_fragment}), the numerical simulations presented show that this is in general not the case. 
Further investigations using the H$_6$-model (vide infra) reveal the existence of a specific configuration ($\Theta_3$) for which only Assumption (A4) is not satisfied. In the vicinity of this configuration, DMET has at least two distinct solutions that arise from a transcritical bifurcation at $\Theta_3$.
In Section~\ref{sec:hl_DMET}, we formulate the impurity problem in more detail and discuss the domain of the high-level DMET map. In Section~\ref{sec:ll_DMET}, we study the $\NElec$-representability problem mentioned above and provide a simple criterion of local $\NElec$-representability directly connected to Assumption (A3). In order to improve the readability of the paper, we postponed the technical proofs to Section~\ref{sec:proofs}. For the reader's convenience, the main notations used throughout this article are collected in Table~\ref{tab:notations} in Appendix~\ref{sec:appendix}.

\section{The DMET formalism}

\subsection{The quantum many-body problem and its fragment decomposition}
\label{sec:notations}

We consider a physical system with $\DimH$ quantum sites, with one orbital per site, occupied by $1 \leq \NElec < \DimH$ electrons, and assume that magnetic effects (interaction with an external magnetic field, spin-orbit coupling, etc.) can be neglected. This allows us to work with real-valued wave-functions and density matrices. We set
\begin{align}
\HSpace := \R^\DimH \quad \text{(one-particle state space)}, \quad & \AtomicBasis := \{\AtomicVector\}_{\kappa \in \llbracket1,\DimH \rrbracket} \quad \text{(canonical basis of  $\R^\DimH$)}, \label{eq:one_particle_state_space} \\
\HSpaceN :=\bigwedge^n \HSpace \quad  \text{($n$-particle state space)}, \quad & \Fock(\HSpace) := \bigoplus_{n=0}^{\DimH} \HSpaceN \quad  \text{(real fermionic Fock space)}. \nonumber
\end{align}
We denote by $\Ann$ and $\Ann^\dagger$ the generators of the (real) CAR algebra associated with the canonical basis of $\HSpace$, i.e. 
$$
\Ann := \Ann[](\AtomicVector) \quad \mbox{and} \quad \Ann^\dagger = \Ann[]^\dagger (\AtomicVector).
$$
Recall that the maps
$$
\R^\DimH \ni f \mapsto \Ann[]^\dagger(f) \in \LinearMap(\Fock(\HSpace)) \quad \mbox{and} \quad \R^\DimH  \ni f \mapsto \Ann[](f) \in \LinearMap(\Fock(\HSpace)),
$$
are both linear in this setting since we work in a real Hilbert space framework. Here and below, $\LinearMap(E)$ is the space of linear operators from the finite-dimensional vector space $E$ to itself. We also define the number operator $\NumberOpe$ by
\begin{align*}
\NumberOpe:= \sum_{n=0}^\DimH n \, \widehat{\1}_{\HSpaceN} = \sum_{\kappa=1}^\DimH \Ann^\dagger \Ann && \mbox{(particle number operator)}.
\end{align*}
For each linear subspace $E$ of $\HSpace$, we denote the orthogonal projector on $E$ by $\Projector \in \LinearMap(\HSpace)$.
We assume that the Hamiltonian of the system in the second-quantized formulation reads 
\begin{equation} \label{eq:Hamiltonian}
\Hamiltonian := \sum_{\kappa,\lambda=1}^\DimH h_{\kappa \lambda} \Ann[\kappa]^\dagger \Ann[\lambda] + \frac{1}2  \sum_{\kappa, \lambda, \nu, \xi=1}^\DimH  V_{\kappa \lambda \nu \xi} \Ann[\kappa]^\dagger \Ann[\lambda]^\dagger \Ann[\xi]  \Ann[\nu],
\end{equation}
where the matrix $h \in \R^{\DimH \times \DimH}$ and the 4th-order tensor $V \in \R^{\DimH \times \DimH \times \DimH \times \DimH}$ satisfy the following symmetry properties:
$$
h_{\kappa \lambda}=h_{\lambda \kappa}\quad \mbox{and} \quad V_{\kappa\lambda\nu\xi} = V_{\nu\lambda\kappa\xi} = V_{\kappa\xi\nu\lambda}  = V_{\nu\xi\kappa\lambda}. 
$$
We denote by $\Grass[\NElec]$ the Grassmannian of rank-$\NElec$ orthogonal projectors in $\R^\DimH$:
\begin{equation}\label{eq:def_Grass}
\Grass[\NElec]:={\rm Gr}_\R(\NElec,\DimH)=\{ D \in  \R^{\DimH \times \DimH}_{\rm sym}  \; | \; D^2=D, \; \Trace(D)=\NElec \},
\end{equation}
and by $\CHGrass[\NElec]$ the convex hull of $\Grass[\NElec]$, i.e.
\begin{equation}\label{eq:def_CHGrass}
\CHGrass[\NElec]= \{ D \in \R^{\DimH \times \DimH}_{\rm sym} \; | \;0 \leq D \leq 1, \; \Trace(D)=\NElec \}.
\end{equation}
Physically, the set $\CHGrass[\NElec]$ corresponds to the set of (real-valued, mixed-state) $\NElec$-representable one-body density matrices with $\NElec$ electrons, and $\Grass[\NElec]$ is the set of one-body density matrices generated by (real-valued) Slater determinants in $\HSpaceN[\NElec]$. 

\medskip

We consider a fixed partition of the $\DimH$ sites into $\NFrag$ non-overlapping fragments $\{\FragIndic \}_{x \in \llbracket 1,\NFrag \rrbracket}$ of sizes $\{ \FragSize \}_{x \in \llbracket 1,\NFrag \rrbracket}$ such that $\FragSize < \NElec $ for all $x$. Up to reordering the sites, we can assume that the partition is the following:
\begin{equation}\label{eq:partition}
\llbracket 1, \DimH \rrbracket = 
\bigg\{ \underbrace{(1,\cdots,\FragSize[1])}_{\FragIndic[1]},\underbrace{(1+\FragSize[1],\cdots, \FragSize[1]+\FragSize[2])}_{\FragIndic[2]},\cdots, \underbrace{(1+\FragSize[1]+\cdots+\FragSize[\NFrag-1],\cdots,\DimH)}_{\FragIndic[\NFrag]} \bigg\}.
\end{equation} 

This partition corresponds to a decomposition of the space into $\NFrag$ fragment subspaces fulfilling
\begin{equation}\label{eq:dec_H1}
\HSpace = \Frag[1] \oplus \cdots \oplus \Frag[\NFrag] \quad \mbox{with} \quad \Frag := \Span(\AtomicVector, \; \kappa \in \FragIndic).
\end{equation}
For $M \in \mathbb{R}^{\DimH \times \DimH}_{\rm sym}$, we set
\begin{equation} \label{eq:PartitionProjDefinition}
  \PartitionProj(M):= \sum_{x=1}^{\NFrag} \Projector[x] M \Projector[x],
\end{equation}
where $\Projector[x]:=\Projector[\Frag]$ is the orthogonal projector on $\Frag$.
The operator $\PartitionProj \in \LinearMap(\mathbb{R}^{\DimH \times \DimH}_{\rm sym})$ is the orthogonal projector onto the set of block-diagonal matrices for the partition \eqref{eq:partition} (endowed with the Frobenius inner product).

\medskip

As we will see, a central intermediary in DMET is the diagonal blocks of the density matrix,
$P = \PartitionProj(D) \in \PartitionProj(\Grass)$. It is clear that these blocks
must satisfy $0 \le P_{x} \le 1$ and $\sum_{x=1}^{\NFrag} \Trace(P_x)=N$.
Conversely, it is easy to see that grouping these
blocks together into a block-diagonal matrix produces a matrix in
$\CHGrass$; therefore, we have
\begin{align}
  \DiagonalBlockCHGrass := \PartitionProj(\CHGrass)
  = \bigg\{ P = &
    \left(\begin{array}{cccc}
   P_1& 0 &\cdots & 0 \\
   0 &P_2 &\cdots &0\\
   \vdots & &\ddots &\vdots \\
   0 &0 &\cdots & P_{\NFrag}\\
   \end{array}\right) \nonumber \\ & \mbox{ s.t. } \forall 1 \le x \le \NFrag, \;  
P_x \in \R^{\FragSize \times \FragSize}_{\rm sym}, \; 0 \le P_x \le 1, \; \sum_{x=1}^{\NFrag} \Trace(P_x)=\NElec \bigg\} . \label{eq:def_PP}
\end{align}
From a geometrical viewpoint, $\DiagonalBlockCHGrass$ is a non-empty, compact, convex subset of an affine vector subspace of $\R^{\DimH \times \DimH}_{\rm sym}$ with base vector space
\begin{equation} \label{eq:defY}
  \cY := 
  \bigg\{ Y = 
    \left(\begin{array}{cccc}
   Y_1& 0 &\cdots & 0 \\
   0 &Y_2 &\cdots &0\\
   \vdots & &\ddots &\vdots \\
   0 &0 &\cdots & Y_{\NFrag}\\
   \end{array}\right) \mbox{ s.t. } \forall 1 \le x \le \NFrag, \;  
Y_x \in \R^{\FragSize \times \FragSize}_{\rm sym}, \sum_{x=1}^{\NFrag} \Trace(Y_x)=0 \bigg\} .
\end{equation}
The structure of the set $\PartitionProj(\Grass) \subset \DiagonalBlockCHGrass$ is a more subtle issue that we will investigate in Section~\ref{sec:ll_DMET}.

\subsection{The impurity high-level problem} 
\label{sec:impurity_pb}

Given one of the spaces $\Frag$ and a one-body density matrix $D \in \Grass$, we set:
\begin{equation}\label{eq:ImpuritySubspaceDefinition}
\Impurity{D}:=\Frag+D\Frag = D\Frag \oplus (1-D)\Frag \quad \mbox{($x$-th impurity subspace)}.
\end{equation}
We will assume in the following that 
\begin{equation}\label{eq:max_rank_assumption}
  \dim(D\Frag)= \dim((1-D)\Frag) = \dim(\Frag) = \FragSize \quad \mbox{(maximal-rank assumption)},
\end{equation}
so that $\dim(\Impurity{D}) = 2 \FragSize$. 
Decomposing $\Ran(D)$ and $\Ker(D)$ as
$$
\Ran(D) = D \Frag \oplus \Core{D}  \quad \mbox{and} \quad 
\Ker(D) = (1-D)\Frag \oplus  \Virt{D},
$$
we obtain the following decomposition of $\HSpace=\R^\DimH$:
$$
\HSpace = \Impurity{D} \oplus \underbrace{\Core{D} \oplus \Virt{D}}_{ =:\Env{D}}.
$$
Note that the space $\Core{D}$ has dimension $(\NElec - \FragSize)$. The matrix $D$ can be seen as the one-body density matrix associated with the Slater determinant
$$
\WFSlater = \WFimp  \wedge \WFcore  \quad \mbox{with} \quad 
\WFimp \in \bigwedge^{\FragSize}  D\Frag \quad \mbox{and} \quad 
\WFcore \in \bigwedge^{(\NElec-\FragSize)} \Core{D},
$$
where $\WFimp$ and $\WFcore$ are normalized.
More precisely, $\WFSlater$ is the Slater determinant built from an orthonormal basis of $\FragSize$ orbitals in $D \Frag$ and an orthonormal basis of $(\NElec-\FragSize)$ orbitals in $\Core{D}$. The so-defined wave-function $\WFSlater$ is unique up to an irrelevant sign. 

\medskip

We denote by $\NumberOpe_{\Frag} \in \LinearMap(\Fock(\HSpace))$ the projection of the number operator onto the fragment Fock space $\Fock(\Frag)$.
Solving the impurity problem aims at minimizing, for a given $\ChemicalPotential \in \mathbb{R}$ which will be specified later, the thermodynamic potential 
    \begin{equation}\label{eq:impurityPb1}
    \langle \WF  | (\Hamiltonian - \ChemicalPotential \NumberOpe_{\Frag} ) |\WF \rangle
    \end{equation}
over the set of normalized trial states in $\Fock(\HSpace)$ of the form
\begin{equation}\label{eq:impurityPb2}
 \WF = \WFimpTrial \wedge \WFcore
\end{equation}
 with $\WFcore$ fixed, and $\WFimpTrial$ in 
 \begin{equation}
\Fock(\Impurity{D}):= \bigoplus_{n=0}^{\FragSize} \bigwedge^n \Impurity{D} \quad  \mbox{($x$-th impurity Fock space)}. \nonumber
\end{equation}
The impurity Hamiltonian is the unique operator $\ImpurityHamiltonian{D}$ on $\Fock(\Impurity{D})$ such that
\begin{equation} \label{eq:def_impurity_Ham_1}
\forall \WFimpTrial \in \Fock(\Impurity{D}), \quad
\langle \WFimpTrial  | \ImpurityHamiltonian{D} | \WFimpTrial  \rangle =
\langle \WFimpTrial \wedge \WFcore | \Hamiltonian | \WFimpTrial \wedge \WFcore \rangle.
\end{equation}
For an explicit expression of $\ImpurityHamiltonian{D}$, see Proposition
\ref{prop:HimpDefinition}.

The impurity problem defined by \eqref{eq:impurityPb1}-\eqref{eq:impurityPb2} can then be reformulated as
    \begin{equation} \label{eq:ImpurityProblem}
    \min_{\WFimpTrial \in \Fock(\Impurity{D}), \lVert \WFimpTrial \rVert =1}
    \langle  \WFimpTrial |  \ImpurityHamiltonian{D} - \ChemicalPotential \NumberOpe_{\Frag} | \WFimpTrial \rangle \quad \mbox{(impurity problem)}.
    \end{equation}
In practice, this full-CI problem in the Fock space $\Fock(\Impurity{D})$ is solved by an approximate correlated wave-function method such as CASSCF, CCSD or DMRG for example, but we assume in this analysis that it can be solved exactly.

\medskip

If \eqref{eq:ImpurityProblem} has a non-degenerate ground state for all
$x$, we denote the one-body ground-state density matrices by $\GSDensity(D)$, seen as matrices in $\R^{\DimH \times \DimH}_{\mathrm{sym}}$, and finally set
\begin{equation}\label{eq:impurity_Fhl}
\boxed{\HLMapPartial(D):= \Projector[\Frag] \GSDensity(D)  \Projector[\Frag].}
\end{equation}
Let us remark incidentally that if the ground state of the impurity problem is degenerate, we can either consider $\HLMapPartial(D)$ as a multivalued function or define them from finite-temperature versions of \eqref{eq:ImpurityProblem}, which are strictly convex compact problems on the set of density operators on the Fock space, and therefore always have a unique minimizer. We will not proceed further in this direction and only consider here the case of impurity problems with non-degenerate ground states. 

The combination of the $\NFrag$ impurity problems introduced in \eqref{eq:ImpurityProblem} (see also \eqref{eq:impurity_Fhl}) gives rise to a high-level DMET map $\HLMap$ 
\begin{equation} \label{eq:HLmap0}
 \Grass \ni D \mapsto \HLMap(D) \in \DiagonalBlockCHGrass
\end{equation}
formally defined by 
\begin{equation}\label{eq:HLmap}
\boxed{\HLMap(D) := \sum_{x=1}^{\NFrag}  \HLMapPartial(D) \quad \mbox{(high-level map)} }
\end{equation}
with $\ChemicalPotential \in \mathbb{R}$ chosen such that $\Trace(\HLMap(D))=\NElec$. The domain of $\HLMap$ and the regularity properties of this map will be studied in Section~\ref{sec:hl_DMET}.

\subsection{The global low-level problem}
\label{sec:overview}

The low-level map is defined by 
\begin{equation}\label{eq:LowLevelMapDefinition}
  \boxed{\LLMap(P):= \mathop{\mathrm{argmin}}_{D \in \Grass ,\; \PartitionProj(D)=P} \EMF \left(  D \right) \quad \mbox{(low-level map)}, }
\end{equation}
where $\EMF$ is the Hartree-Fock (mean-field) energy functional of the trial density-matrix $D$. The latter reads 
\begin{equation} \label{eq:HF_functional}
\EMF(D):= \Trace(hD) + \frac 12 \Trace(J(D)D)-\frac 12 \Trace(K(D)D),
\end{equation}
where
\begin{equation} \label{eq:direct_exchange}
[J(D)]_{\kappa\lambda} := \sum_{\nu,\xi=1}^\DimH V_{\lambda\xi\kappa\nu} D_{\nu\xi} 
\quad \mbox{and} \quad 
[K(D)]_{\kappa\lambda}  := \sum_{\nu,\xi=1}^\DimH V_{\kappa\xi\nu\lambda} D_{\nu\xi}. 
\end{equation}
The existence and uniqueness of a minimizer
to~\eqref{eq:LowLevelMapDefinition} will be discussed in
Section~\ref{sec:ll_DMET}.

\subsection{The DMET problem}
\label{sec:DMET_pb}

Finally, the full DMET map is formally defined as the self-consistent solution to the system
\begin{empheq}[box=\widefbox]{align}
  D &= \LLMap(P) \in \Grass, \nonumber \\
  P &= \HLMap(D) \in \mathcal P. \nonumber
\end{empheq}
In particular, $D = \LLMap(P)$ implies that $P = \PartitionProj (D)$. Equivalently, we can formulate the problem as
$$
  \boxed{P = \DMETMap(P) := \HLMap(\LLMap(P)).}
$$

Assuming that the solution to this fixed-point problem exists and is unique, $P$ is expected to provide a good approximation of the diagonal blocks (in the decomposition \eqref{eq:dec_H1} of $\HSpace$) of the ground-state one-body density matrix of the interacting system.
The mathematical properties of this self-consistent loop will be studied in the next section, first for the non-interacting case, and second, for the interacting case in a perturbative regime.

\section{Main results}
\label{sec:main_results}

We now embed the Hamiltonian $H$ into the family of Hamiltonians
\begin{equation} \label{eq:Halpha}
\Hamiltonian_\alpha := \sum_{\kappa,\lambda=1}^\DimH h_{\kappa\lambda} \Ann[\kappa]^\dagger \Ann[\lambda] + \frac{\alpha}2  \sum_{\kappa,\lambda,\nu\xi=1}^\DimH V_{\kappa\lambda\nu\xi} \Ann[\kappa]^\dagger \Ann[\lambda]^\dagger \Ann[\xi] \Ann[\nu], \quad \alpha \in \R,
\end{equation}
acting on $\Fock(\HSpace)$.
For $\alpha=0$, we obtain the one-body Hamiltonian 
\begin{equation} \label{eq:def_H0}
\Hamiltonian_0 := \sum_{\kappa,\lambda=1}^\DimH h_{\kappa\lambda} \Ann[\kappa]^\dagger \Ann[\lambda]
\end{equation}
describing non-interacting particles, and we recover the original Hamiltonian $\Hamiltonian$ for $\alpha=1$. We denote by $\HLMap_{\alpha}$, $\LLMap_{\alpha}$, and $\DMETMap_{\alpha}$ the high-level, low-level, and DMET maps constructed from $\Hamiltonian_\alpha$.
\medskip

We first assume that the non-interacting problem is non-degenerate.
Denoting by $\varepsilon_n$ the $n$-th lowest eigenvalue of $h$ (counting multiplicities), this condition reads

\begin{description}
\item[(A1)] { $\varepsilon_N < 0 < \varepsilon_{N+1}$},
\end{description}
where without loss of generality we have chosen the Fermi level to be $0$. Assumption (A1) indeed implies that the ground-state of $\Hamiltonian_0$ in the $\NElec$-particle sector of the Fock space is non-degenerate, and that the ground-state one-body density is the rank-$\NElec$ orthogonal projector given by
\begin{equation}\label{eq:GS_DM}
\DGS=\1_{(-\infty,0]}(h).
\end{equation}
By perturbation theory, the ground state of $\widehat H_\alpha$ in the $\NElec$-particle sector is non-degenerate for all $\alpha \in (-\alpha_+, \alpha_+)$ for some $0 < \alpha_+ \le +\infty$. We denote by $D_\alpha^{\rm exact}$ the corresponding ground-state one-body density matrix.
As a consequence of analytic perturbation theory for hermitian matrices, the map $(-\alpha_+,\alpha_+) \ni \alpha \mapsto D_\alpha^{\rm exact} \in \R^{\DimH \times \DimH}_{\rm sym}$ is real-analytic.

\medskip

Second, we make the maximal-rank assumption:
\begin{description}
\item[(A2)] For all $1 \le x \le \NFrag$, $\dim(\DGS\Frag)= \dim((1-\DGS)\Frag) = \dim(\Frag) = \FragSize$.
\end{description}
Assumption (A2) implies that the impurity problem \eqref{eq:ImpurityProblem} for $\Hamiltonian=\Hamiltonian_0$ and $D=\DGS$ is well-defined for each $x$ and each $\mu$. 
We emphasize however that this does not prejudge that the so-obtained $\NFrag$ impurity problems are well-posed (i.e. have a unique ground-state) for a given value of $\mu$, nor {\it a fortiori} that $\DGS$ is in the domain of the high-level map $\HLMap_{0}$.
We will elaborate more on the meaning of Assumptions (A2) in Section~\ref{sec:hl_DMET}. 

\medskip

DMET is then consistent in the non-interacting case:

\begin{proposition}[$\PGS:=\PartitionProj(\DGS)$ is a fixed point of the DMET map for
  $\alpha=0$] 
\label{prop:DMET0} 
Under Assumptions (A1)-(A2), $\PGS:=\PartitionProj(\DGS)$ is a fixed point of the non-interacting DMET iterative scheme, i.e. $\PGS$ is in the domain of $\LLMap_0$, $\DGS$ is in the domain of $\HLMap_0$, and 
$\DMETMap_0(\PGS)=\PGS$.
 \label{prop:DGSFixedPoint}
\end{proposition}

\medskip

\begin{remark} \label{rem:hl_HF}
We formally define the {\em high-level Hartree-Fock} map 
$$
\HLMap_{\rm MF} : \Grass \to \DiagonalBlockCHGrass,
$$
as the high-level map constructed from the Hartree-Fock $\NElec$-body Hamiltonian
$$
\Hamiltonian_D^{\rm HF}:= \sum_{\kappa,\lambda=1}^\DimH [h^{\rm HF}(D)]_{\kappa\lambda} \Ann[\kappa]^\dagger \Ann[\lambda],
$$
where
\begin{equation}\label{eq:Fock_Hamiltonian}
h^{\rm HF}(D) = h+J(D)-K(D)
\end{equation}
is the one-particle mean-field (Fock) Hamiltonian.
Using exactly the same arguments as in the proof of Proposition~\ref{prop:DGSFixedPoint}, we obtain that the low-level map $\LLMap$ satisfies the mean-field consistency property 
$$
\LLMap(\HLMap_{\rm HF}(D_*))=D_*,
$$
for any Hartree-Fock ground state $D_*$.
We will make use of this important observation in the proof of Theorem~\ref{thm:theory}.
\end{remark}

\medskip

We now study the DMET equations in the perturbative regime of $\alpha$ small.
In order to use perturbative techniques, we need to determine the space in which we seek $P$.
Generically, at $\alpha \neq 0$, we expect $P$ to be equal to the block diagonal of the one-body density matrix, which is not a projector.
Therefore it is natural to seek $P$ in $\mathcal P = \PartitionProj(\CHGrass)$.
However, in the DMET method, $D$ is constrained to be a projector, and therefore $P$ will necessarily belong to $\PartitionProj(\Grass)$.
We will study in Section~\ref{sec:ll_DMET} the relationship between the two sets $\mathcal P$ and $\PartitionProj(\Grass)$ (the $\NElec$-representability problem), and in particular show that, in the regime of interest to DMET (many relatively small fragments, so that $\DimH \gg \max_{x} \FragSize$), the two sets are (generically) locally the same.
Therefore, it is natural to assume the local $\NElec$-representability condition:
\begin{description}
\item[(A3)] The linear map $\PartitionProj$ is surjective from
  $\mathcal T_{D_{0}} \Grass$ to $\mathcal Y$,
\end{description}
where $\cY$ is the vector subspace defined in \eqref{eq:defY}.
Indeed, $\mathcal P$ is a (non-empty, compact, convex) subset of the affine space $\PGS+\cY$ and Assumption (A2) implies that $\PGS \in \dps\mathop{\mathcal P}^\circ$, where $\dps\mathop{\mathcal P}^\circ$ is the interior of ${\mathcal P}$ in $\PGS+\cY$.
Thus $\cY$ can be identified with the tangent space at $\PGS$ to the manifold $\dps\mathop{\mathcal P}^\circ$.
By the local submersion theorem, this implies that any $P$ in the neighborhood of $P_{0}$ can be expressed as the block diagonal of a density matrix in the neighborhood of $D_{0}$ in $\Grass$.

\medskip

Our last assumption is concerned with the response properties of the impurity problems at the non-interacting level. 
Consider a self-adjoint perturbation $Y \in \R^{\DimH \times \DimH}_{\rm sym}$ of the one-particle Hamiltonian $h$, non-local but  block-diagonal in the fragment decomposition, i.e. such that $Y \in \R I_L +\cY$, and denote by $\widetilde \HLMap_{h+Y}(D)$ the non-interacting high-level map obtained by replacing $h$ with $h+Y$ (so that $\widetilde \HLMap_{h}(D)=\HLMap_0$).
Formally, we have
\begin{equation}\label{eq:def_R_formal}
\widetilde \HLMap_{h+Y}(\DGS) = 
\PGS + RY + o(\lVert Y \rVert),
\end{equation}
with $R : \R I_L +\cY \to \cY$ linear (the fact that $RY \in \cY$ is due to particle-number conservation).
The map $R$ can be interpreted as a non-interacting static 4-point density-density linear response function for frozen impurity spaces.
It follows from Assumption (A1) that constant perturbations do not modify the density matrix: $R(I_L)=0$.
Our fourth assumption reads:
\begin{description}
\item[(A4)] the 4-point linear response function $R : \cY \to \cY$ is invertible.
\end{description}

This condition is somewhat reminiscent of the Hohenberg-Kohn theorem from Density Functional Theory.
Together with the local inversion theorem, it implies that, locally around $h$, in the non-interacting case and for frozen impurity spaces $\Impurity{\DGS}$, the high-level map defines a one-to-one correspondence between non-local fragment potentials (up to a constant shift) and fragment density matrices.

\medskip

\begin{remark}\label{rem:one_site_per_fragment}
We will show in Section~\ref{sec:A4} that in the case when $\NFrag=\DimH$ (one site per fragment), it holds: under Assumptions (A1)-(A2),
\begin{align*}
\mbox{(A3) is satisfied} &\implies  
\DGS \mbox{ is an irreducible matrix} \iff 
\mbox{(A4) is satisfied}.
\end{align*}  
On the other hand, numerical simulations indicate that in the general case, Assumption (A4) is not a consequence of Assumptions (A1)-(A3).
\end{remark}

\medskip

We are now in position to state our main results.

\begin{theorem}[DMET is well-posed in the perturbative regime]\label{thm:theory}  Under assumptions (A1)-(A4),
there exist $0 < \widetilde\alpha_+  \le \alpha_+$, and a neighborhood $\Omega$ of $\DGS$ in $\Grass$ such that for all $\alpha \in (-\widetilde\alpha_+,\widetilde\alpha_+)$, the fixed-point DMET problem
\begin{equation}
P^{\rm DMET}_\alpha = \HLMap_\alpha(D^{\rm DMET}_\alpha), \quad D^{\rm DMET}_\alpha = \LLMap_\alpha(P^{\rm DMET}_\alpha) \nonumber
\end{equation}
has a unique solution $(D^{\rm DMET}_\alpha,P^{\rm DMET}_\alpha)$ with $D^{\rm DMET}_\alpha \in \Omega$ (otherwise stated, the DMET map for $H_\alpha$ has a unique fixed point $P^{\rm DMET}_\alpha$ in the neighborhood of $\PGS$).
In addition, the maps $(-\widetilde\alpha_+,\widetilde \alpha_+) \ni \alpha \mapsto  D^{\rm DMET}_\alpha \in \R^{\DimH \times \DimH}_{\rm sym}$  and $(-\widetilde\alpha_+,\widetilde \alpha_+) \ni \alpha \mapsto  P^{\rm DMET}_\alpha \in \R^{\DimH \times \DimH}_{\rm sym}$  are real-analytic and such that 
\begin{equation}
\DGS^{\rm DMET}=\DGS = \1_{(-\infty,0]}(h), \quad \PGS^{\rm DMET}=\PGS = \PartitionProj(\DGS). \nonumber
\end{equation}
\end{theorem}

As is standard, the first-order perturbation of the exact density matrix is given by the Hartree-Fock method.
DMET is able to reproduce this, and is therefore exact at first order:
\begin{theorem}[DMET is exact to first order] \label{thm:DMET_HF} Under Assumptions (A1)-(A4) and with the notation of Theorem~\ref{thm:theory}, it holds
\begin{align*}
D_{\alpha}^{\rm DMET} = D_{\alpha}^{\rm exact} + O(\alpha^{2}) = D_{\alpha}^{\rm HF} + O(\alpha^{2}),
\end{align*}
where $D_{\alpha}^{\rm HF}$ is the Hartree-Fock ground-state density matrix for $\widehat H_\alpha$, which is unique for $\alpha$ small enough.
\end{theorem}

\medskip

The numerical simulations reported in the next section show that such exactness property is not expected to hold at second order.

In the weakly interacting regime, the solution $D_{\alpha}^\mathrm{DMET}$ to the DMET fixed-point problem is the only physical one because it is the only one laying in the vicinity of $\DGS$, where the exact ground-state density matrix is known to be by analytic perturbation theory.

\section{Numerical simulations}
\label{sec:numerics}

In this section, we perform numerical investigations of DMET for two distinct test systems:
The first system is H$_{10}$ in a circular geometry which serves as a benchmark where DMET has been previously recognized for its exceptional performance~\cite{DMET2013}. 
By studying this system, we aim to reaffirm the efficacy of DMET and numerically showcase that DMET is exact to first order in the non-interacting limit.
However, to gain a comprehensive understanding of DMET's limitations, we also explore a second system which is an H$_6$ variant. 
This particular system allows us to numerically scrutinize the assumptions made in the analysis presented above. 
Through these numerical investigations, we aim to provide valuable insights into the mathematical structure of DMET, paving the way for further advancements and improvements in this promising computational approach.
Throughout this section, we denote by $\|\cdot\|_{\rm F}$ the Frobenius norm on matrix spaces. 

\subsection{H$_{10}$ ring}
\label{sec:H10}

We consider a circular arrangement of ten hydrogen atoms, with a nearest-neighbor distance of 1.5~$a_0$ between each pair of atoms (where $a_0 \simeq 0.529$ {\AA} is the Bohr radius). The system is treated using the STO-6G basis set and is half-filled, i.e., containing ten electrons. We partition the system into five fragments, each consisting of two atoms, as shown in Figure \ref{fig:H10_circ_mod}. 

\begin{figure}[!ht]
    \centering
    \includegraphics[width = 0.25\textwidth]{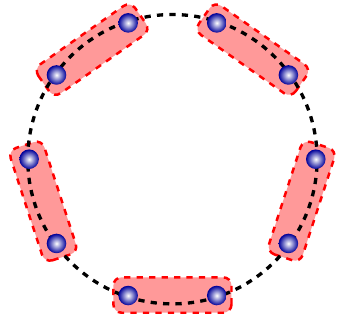}
    \caption{Depiction of the H$_{10}$ system in circular geometry. The red-shaded areas show the chosen fragmentation.}
    \label{fig:H10_circ_mod}
\end{figure}

In order to numerically confirm that DMET is exact to first order for this ``well-behaved'' system, we determine $P_\alpha$ for $\alpha \in [0,1]$ and compute $\Vert dP_\alpha/d\alpha \Vert_F$.  
Figure~\ref{fig:H6_firs_derv} compares the DMET result with the exact diagonalization result (abbreviated FCI). 
We clearly see that DMET is indeed exact to first order for the considered system.

\begin{figure}[!ht]
\centering
\begin{subfigure}[b]{0.48\textwidth}
    \centering
    \includegraphics[width=\textwidth]{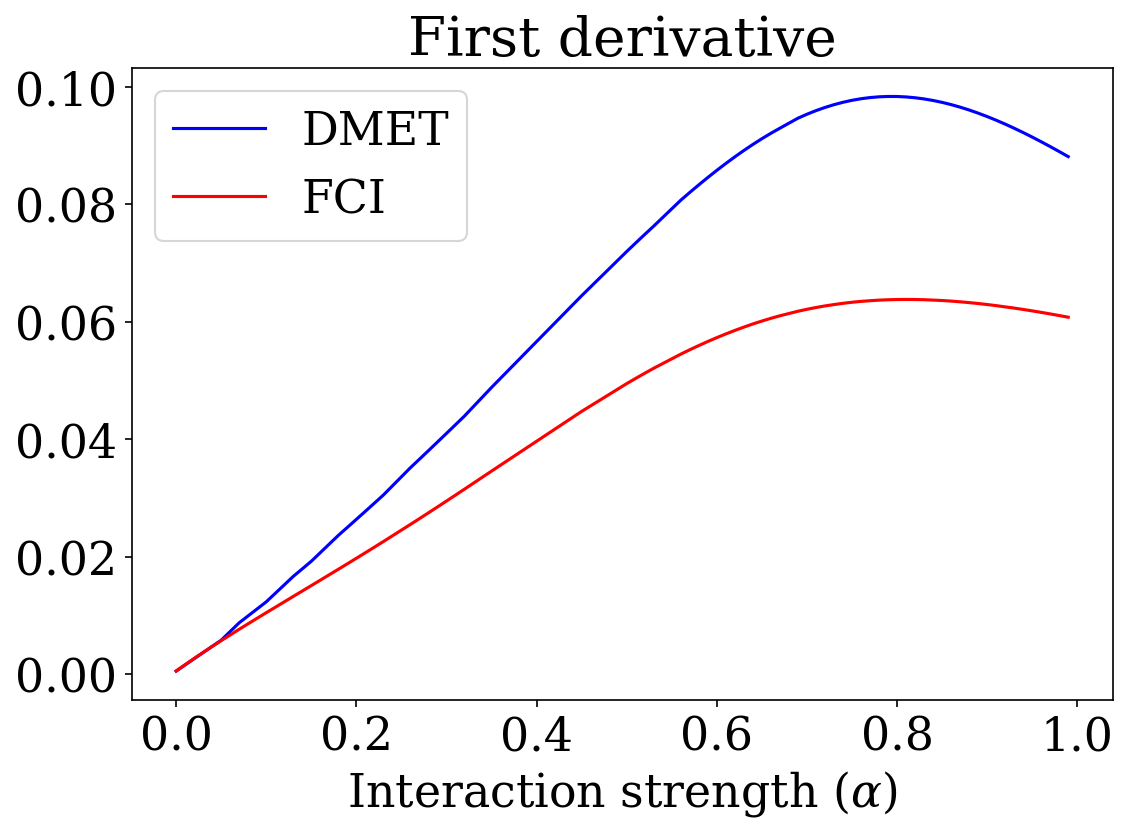}
    \caption{}
    \label{fig:h10_first_derivative}
\end{subfigure}
\hfill
\begin{subfigure}[b]{0.48\textwidth}
    \centering
    \includegraphics[width=\textwidth]{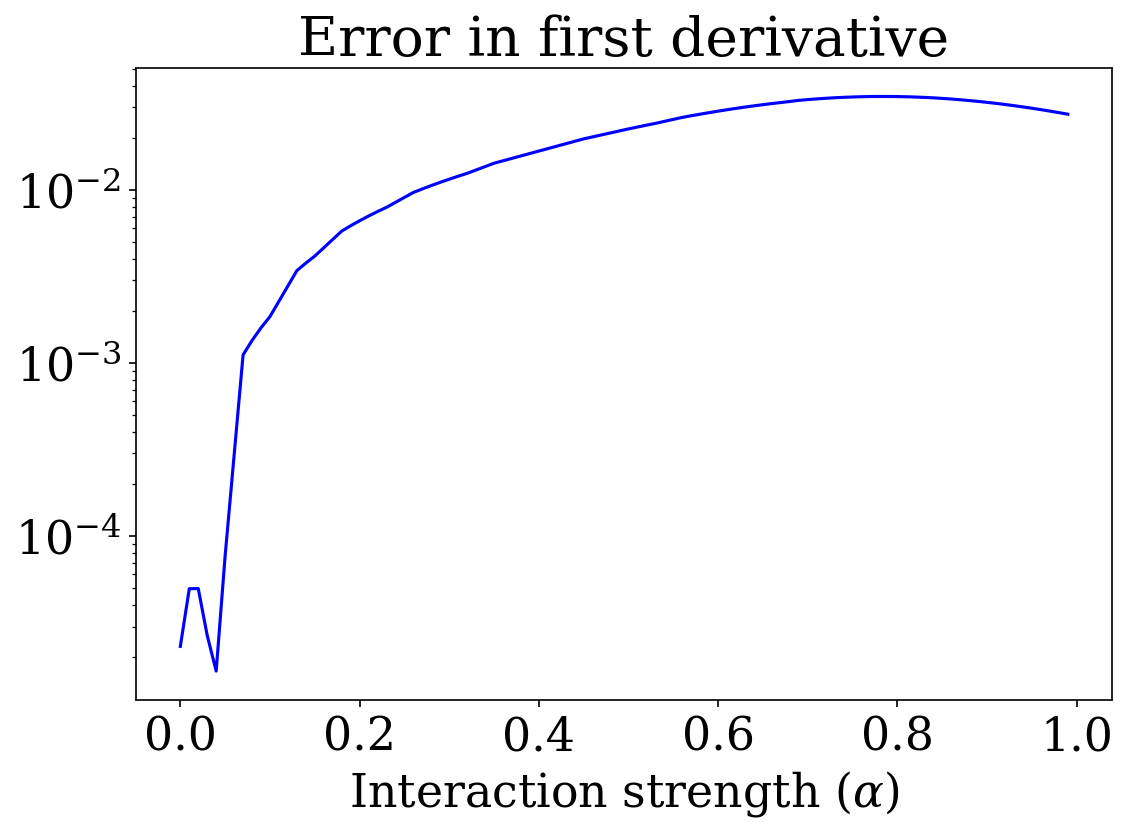}
    \caption{}
    \label{fig:h10_first_derivative_error}
\end{subfigure}
\caption{\label{fig:H6_firs_derv}
(\subref{fig:h10_first_derivative}) Shows $\Vert dP_\alpha/d\alpha \Vert_F$ for DMET and FCI, respectively 
(\subref{fig:h10_first_derivative_error}) Shows the error on $ dP_\alpha/d\alpha$ between DMET and FCI, measured in Frobenius norm.
}
\end{figure}

\subsection{H$_6$ model}
\label{sec:H6}

In this section, we will numerically investigate the assumptions required for the analysis presented in this article. 
To that end, we consider a non-interacting H$_6^{4-}$ system, undergoing the following transition on a circular geometry:
We begin by placing three hydrogen molecules in equilibrium geometry, i.e., bond length of $1.4\,a_0$, equidistantly on a circle of radius $3\,a_0$. We then dissociate each hydrogen molecule while maintaining a circular geometry. Specifically, we break each hydrogen molecule in such a way that the hydrogen atoms from neighboring molecules can form new molecules. We stop this transition at $\Theta=\Theta_{\rm max}$, when the hydrogen atoms from neighboring molecules form new hydrogen molecules in equilibrium geometry. We steer this transition with the angle $\Theta$ that measures the displacement of the individual hydrogen atoms relative to their initial positions. The dissociation is done in a manner that maintains the circular arrangement of the hydrogen atoms throughout the process, see Figure~\ref{fig:H6_example} for a schematic depiction of this process and a depiction of $\Theta$. The system is partitioned into 3 fragments that correspond to the initial molecules. Note that the fragments remain unchanged during the transition process.
In order to fulfill the $N$-representability condition \eqref{eq:NecCondLorRep} below (which is necessary for Assumption (A3) to be fulfilled), we dope the system with four additional electrons, i.e., 10 electrons in total. The system is discretized using the 6-31G basis set.

\begin{figure}[ht!]
    \centering
    \includegraphics[width=0.8\textwidth]{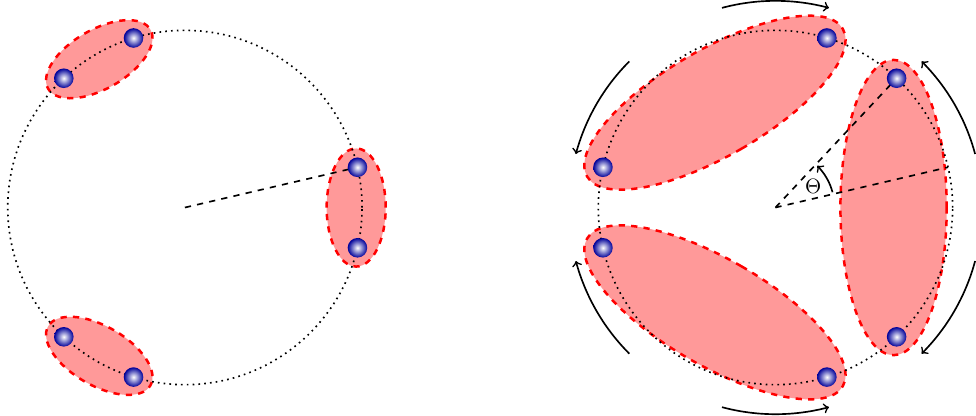}
    \caption{Schematic depiction of the considered H$_6$ transition. The left panel shows the initial configuration for $\Theta=0$; the right panel shows the final configuration $\Theta=\Theta_{\rm max}$. The red-shaded areas depict the imposed fragmentation. The arrows indicate the transition of the hydrogen atoms for $\Theta \in [0, \Theta_{\rm max}]$.}
    \label{fig:H6_example}
\end{figure}

In order to numerically depict Theorem~\ref{thm:DMET_HF}, we compute $P_\alpha$ and $D_\alpha$ using a mean-field theory approach (HF), DMET, and the exact diagonalization (FCI), and compare these quantities for $\alpha=0$ as well as their first derivatives with respect to $\alpha$.
Note that in the non-interacting limit, the mean-field theory is exact, which is reflected in our simulations. We indeed observe that $\sup_\Theta \Vert P_0^{\rm HF}(\Theta) - P_0^{\rm FCI}(\Theta) \Vert_F$ and $\sup_\Theta \Vert D_0^{\rm HF}(\Theta) - D_0^{\rm FCI}(\Theta) \Vert_F$ are equal to zero up to numerical accurary, while $\sup_\Theta \Vert P_0^{\rm DMET}(\Theta) - P_0^{\rm FCI} \Vert_F(\Theta)$, , $\sup_\Theta \Vert D_0^{\rm DMET}(\Theta) - D_0^{\rm FCI}(\Theta) \Vert_F$ are respectively of the order of $10^{-13}$ and $10^{-7}$ with the chosen convergence thresholds.
Figure~\ref{fig:H6_dervs} shows the first-order exactness of DMET in the non-interacting limit for the H$_6^{4-}$ model.

\begin{comment}
\begin{figure}[!ht]
\centering
\begin{subfigure}[b]{0.48\textwidth}
    \centering
    \includegraphics[width=\textwidth]{Graphics/Norm_difference_frag.png}
    \caption{}
    \label{fig:H6_diff}
\end{subfigure}
\hfill
\begin{subfigure}[b]{0.48\textwidth}
    \centering
    \includegraphics[width=\textwidth]{Graphics/Norm_difference_full.png}
    \caption{}
    \label{fig:H6_diff_full}
\end{subfigure}
\caption{\label{fig:H6_diffs}
(\subref{fig:H6_diff}) Shows $\Vert P_0^{\rm DMET} - P_0^{\rm FCI} \Vert_F$ and $\Vert P_0^{\rm HF} - P_0^{\rm FCI} \Vert_F$
(\subref{fig:H6_diff_full}) Shows $\Vert D_0^{\rm DMET} - D_0^{\rm FCI} \Vert_F$ and $\Vert D_0^{\rm HF} - D_0^{\rm FCI} \Vert_F$.
}
\end{figure}
\end{comment}

\begin{figure}[!ht]
\centering
\begin{subfigure}[b]{0.48\textwidth}
    \centering
    \includegraphics[width=\textwidth]{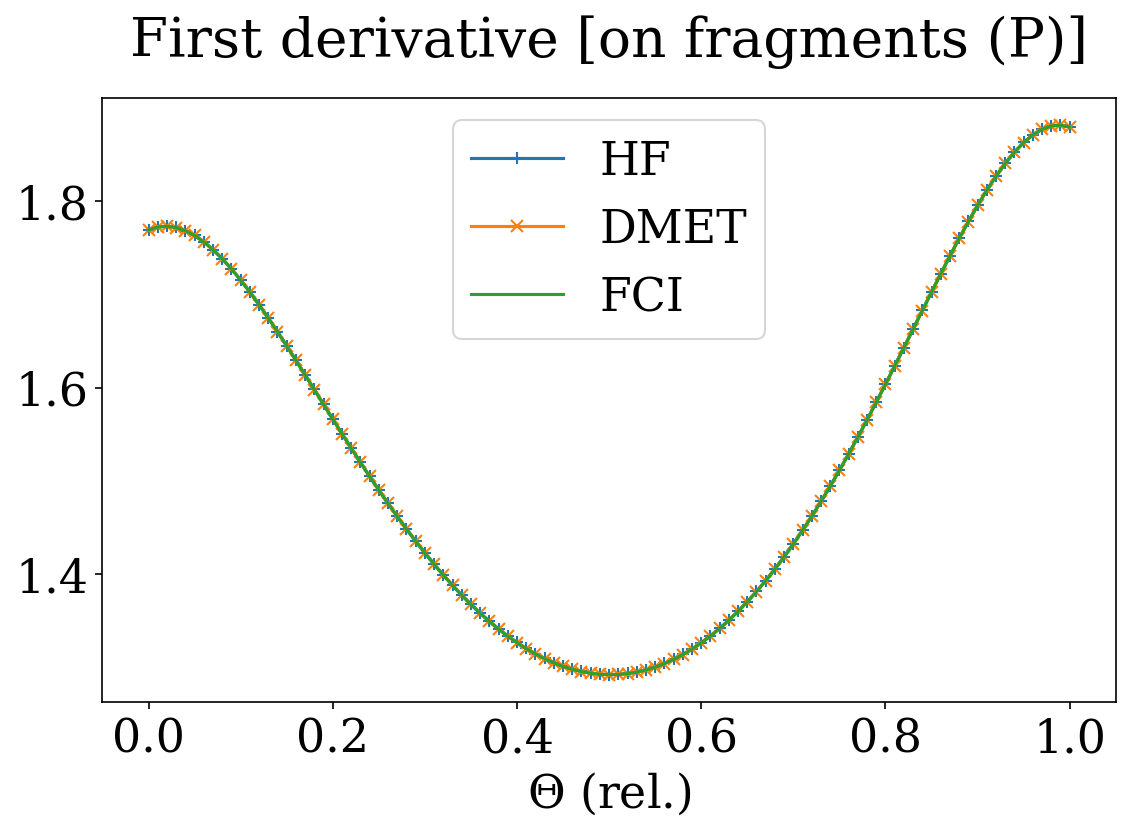}
    \caption{}
    \label{fig:H6_derv}
\end{subfigure}
\hfill
\begin{subfigure}[b]{0.48\textwidth}
    \centering
    \includegraphics[width=\textwidth]{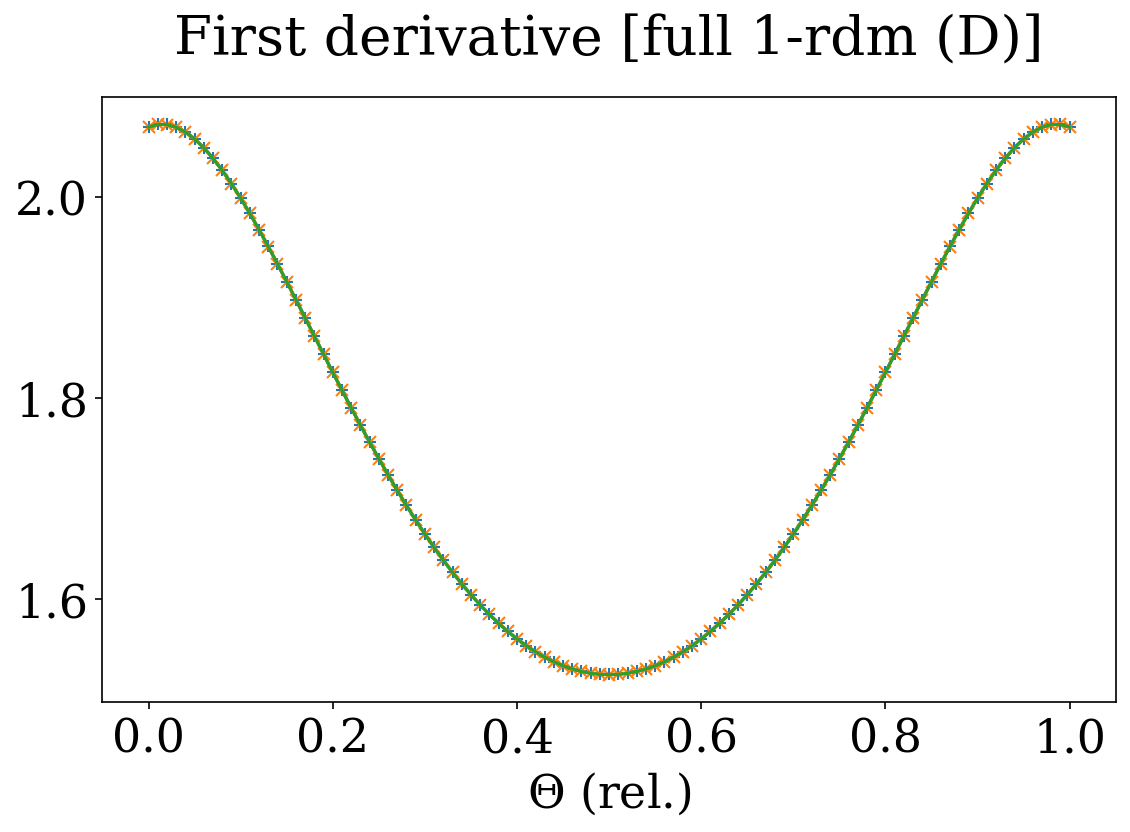}
    \caption{}
    \label{fig:H6_derv_full}
\end{subfigure}
\caption{\label{fig:H6_dervs}
(\subref{fig:H6_derv}) Shows $\Vert \partial_\alpha P_\alpha \big|_{\alpha = 0} \Vert_F$ for HF, DMET and FCI
(\subref{fig:H6_derv_full}) Shows $\Vert \partial_\alpha D_\alpha \big|_{\alpha = 0} \Vert_F$ for HF, DMET and FCI.
}
\end{figure}

Our numerical investigations include an analysis of Assumptions (A1)-(A4). 
We present a check of Assumptions (A1) and (A2) in Figure~\ref{fig:H6_ass1_2}. 
Assumption (A1) can be directly tested by calculating the HOMO-LUMO gap of the non-interacting Hamiltonian under consideration for each value of $\Theta$. 
Furthermore, Assumption (A2) can be tested by monitoring the behavior of the smallest and largest singular values of the matrix $P_0$ as a function of the variable $\Theta$ (see Lemma~\ref{lem:Omega_L}).

\pagebreak
\begin{figure}[!ht]
\centering %https://www.overleaf.com/project/6332e403b515e741382d08ed
\begin{subfigure}[b]{0.48\textwidth}
    \centering
    \includegraphics[width=\textwidth]{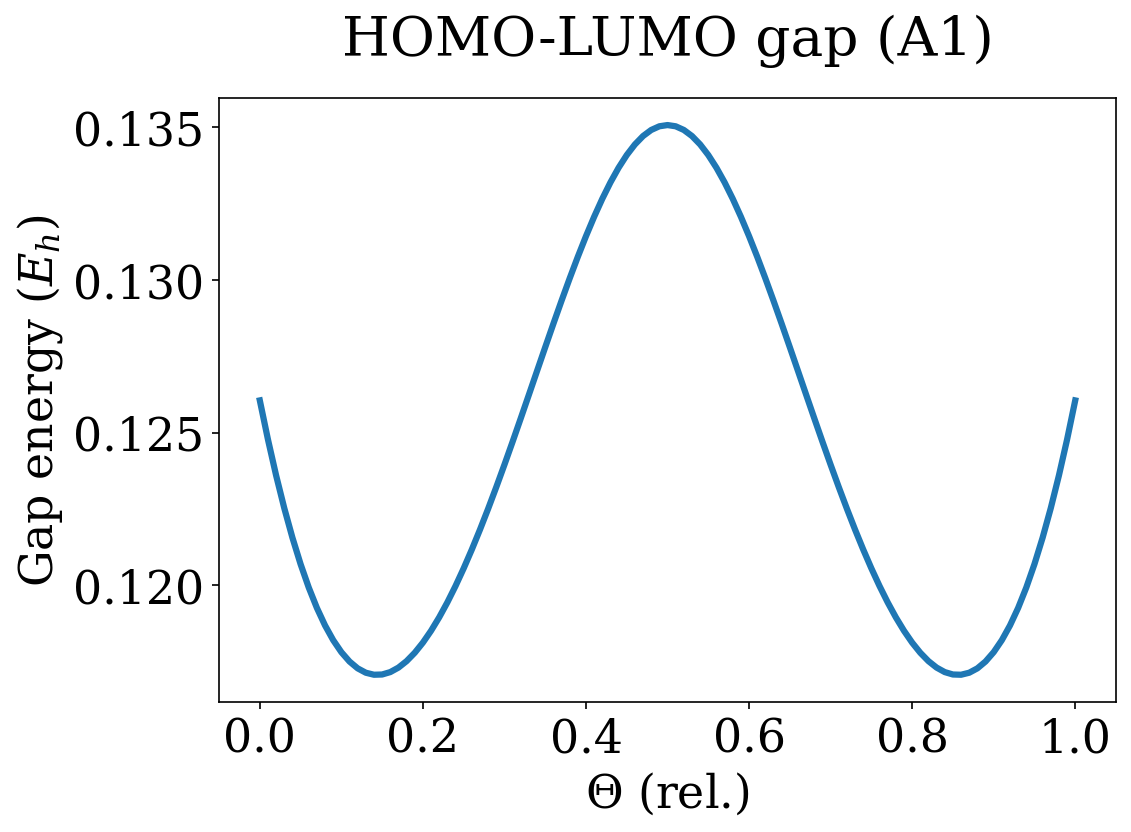}
    \caption{}
    \label{fig:H6_ass1}
\end{subfigure}
\hfill
\begin{subfigure}[b]{0.48\textwidth}
    \centering
    \includegraphics[width=\textwidth]{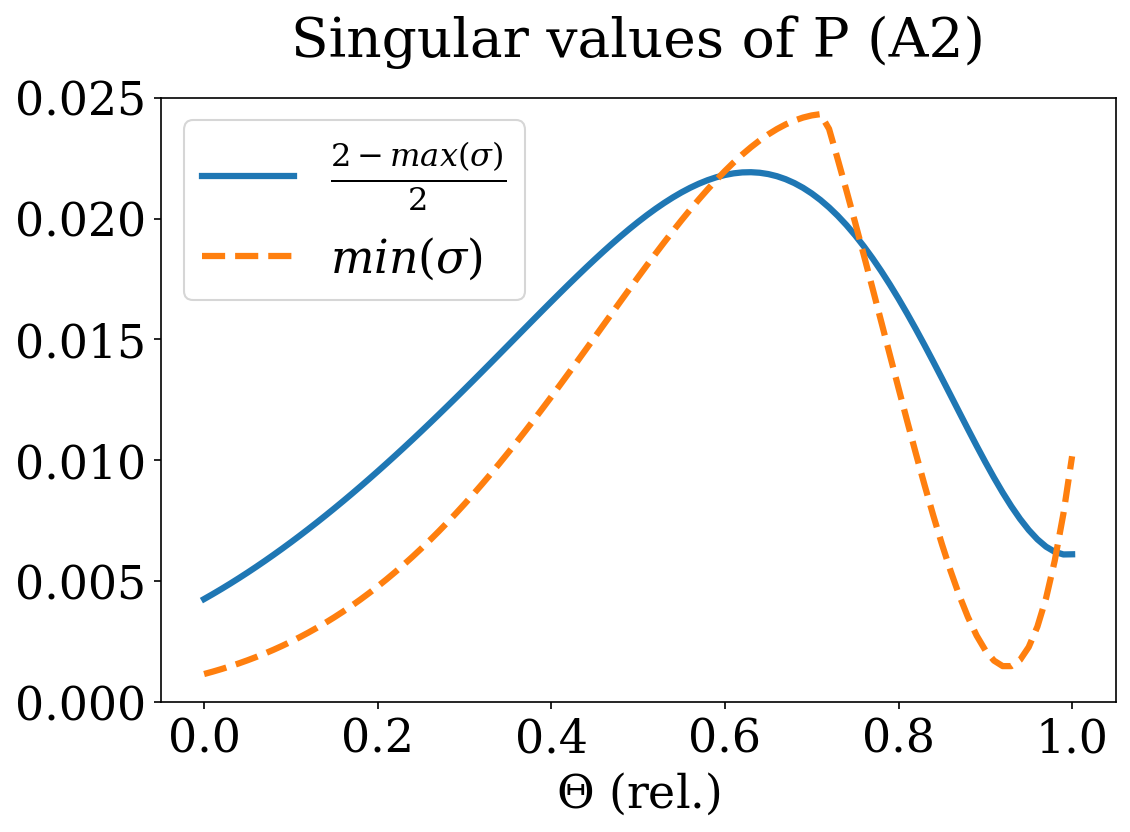}
    \caption{}
    \label{fig:H6_ass2}
\end{subfigure}
\caption{\label{fig:H6_ass1_2}
(\subref{fig:H6_ass1}) Shows the HOMO-LUMO gap for the H$_6$ model as a function of $\Theta$ for $\alpha = 0$. 
(\subref{fig:H6_ass2}) Shows the largest and smallest singular values of $P$ for the H$_6$ model as a function of $\Theta$ for $\alpha = 0$.
}
\end{figure}

The validity of assumptions (A3) and (A4) is tested in Figure~\ref{fig:H6_ass3_4} by monitoring the lowest eigenvalue of the operator $S:=({\rm Bd}|_{T_{D_0}\cD \to \cY})^*{\rm Bd}|_{T_{D_0}\cD \to \cY}$ (which corresponds to (A3)), and the smallest singular value of the operator $R|_{\cY \to \cY}$ (which corresponds to (A4)).

\begin{figure}[!ht]
\centering
\begin{subfigure}[b]{0.48\textwidth}
    \centering
    \includegraphics[width=\textwidth]{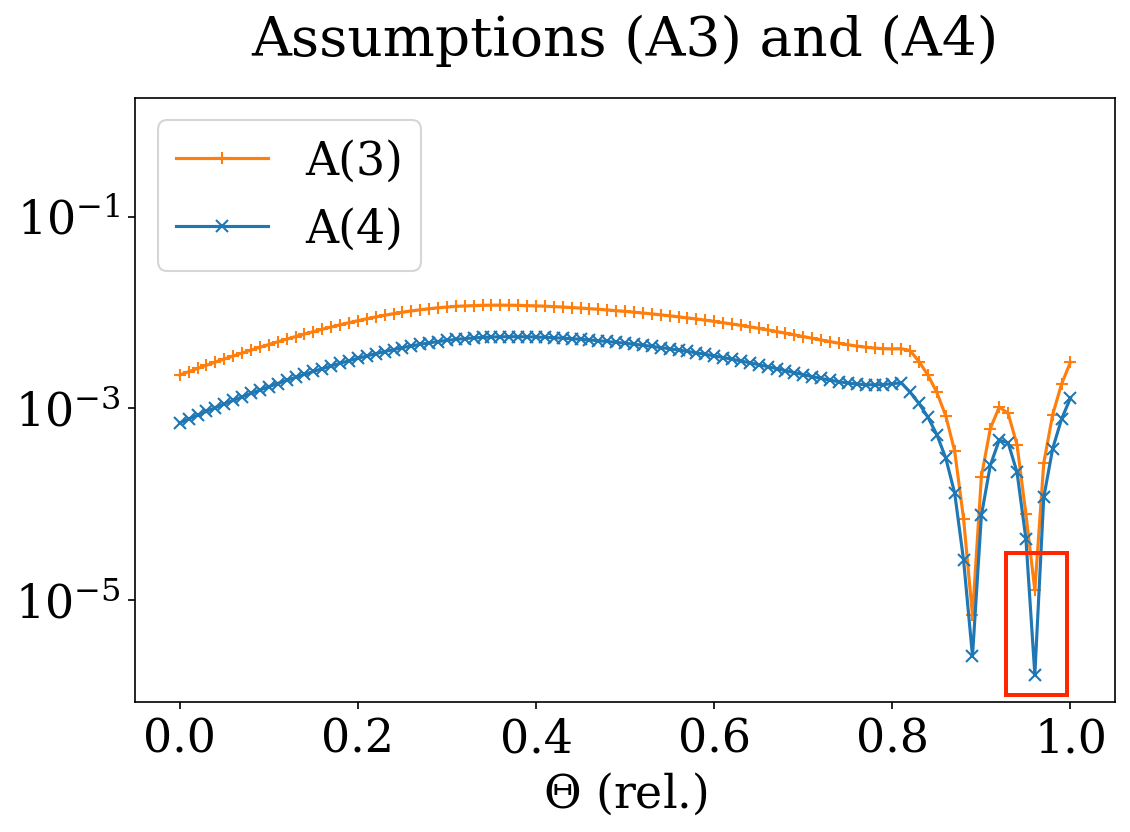}
    \caption{}
    \label{fig:H6_ass34}
\end{subfigure}
\hfill
\begin{subfigure}[b]{0.48\textwidth}
    \centering
    \includegraphics[width=\textwidth]{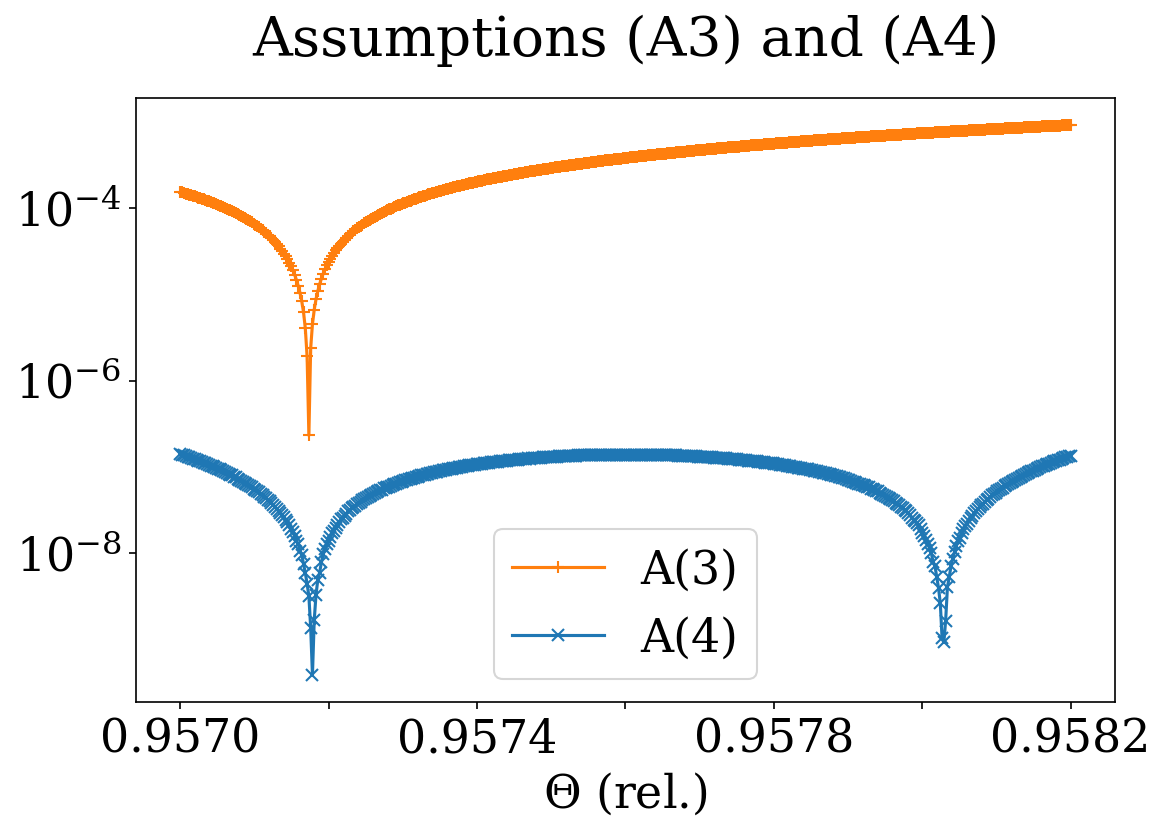}
    \caption{}
    \label{fig:H6_ass34_zoom}
\end{subfigure}
\caption{\label{fig:H6_ass3_4}
(\subref{fig:H6_ass34}) The orange line shows the lowest eigenvalue of $S:=({\rm Bd}|_{T_{D_0}\cD \to \cY})^*{\rm Bd}|_{T_{D_0}\cD \to \cY}$ for the H$_6$ model as a function of $\Theta$ for $\alpha = 0$ (which corresponds to (A3)), and the blue line shows the smallest singular value $\sigma_{\rm min}$ of $R|_{\cY \to \cY}$ (which corresponds to (A4)).
(\subref{fig:H6_ass34_zoom}) Shows a zoomed version of (\subref{fig:H6_ass34}) around the second (local) minimum.
}
\end{figure}

We see that Assumptions (A1) and (A2) are uniformly fulfilled over the whole range $[0,\Theta_{\rm max}]$. Assumption (A3) seems to be satisfied for all $\Theta$ except two values $\Theta_1 \simeq 0.885$ and  $\Theta_2 \simeq 0.957$. Careful testing around $\Theta_2$ shows that Assumption (A4) is additionally not satisfied at $\Theta_3 \simeq 0.958$, where all other assumptions are satisfied. This illustrates the fact that in the general case $N_f < L$, Assumption (A4) is independent of Assumptions (A1)-(A3) (see Remark~\ref{rem:one_site_per_fragment}). 

Figure~\ref{fig:H6_sderv} shows the Frobenius norms of the second derivative of $P_\alpha$ and $D_\alpha$ at $\alpha=0$ for HF, DMET, and FCI. We see that the three methods give different results, and that the result in Theorem~\ref{thm:DMET_HF} is therefore optimal. We also observe that for DMET, the second derivatives become noisy in the range of $\Theta$'s where Assumptions (A3) and (A4) are poorly or not satisfied. This is probably due to conditioning issues or to the use of convergence thresholds not directly connected to the computed quantity of interest. The numerical analysis of DMET is left for future work.

\pagebreak
\begin{figure}[!ht]
\centering
\begin{subfigure}[b]{0.48\textwidth}
    \centering
    \includegraphics[width=\textwidth]{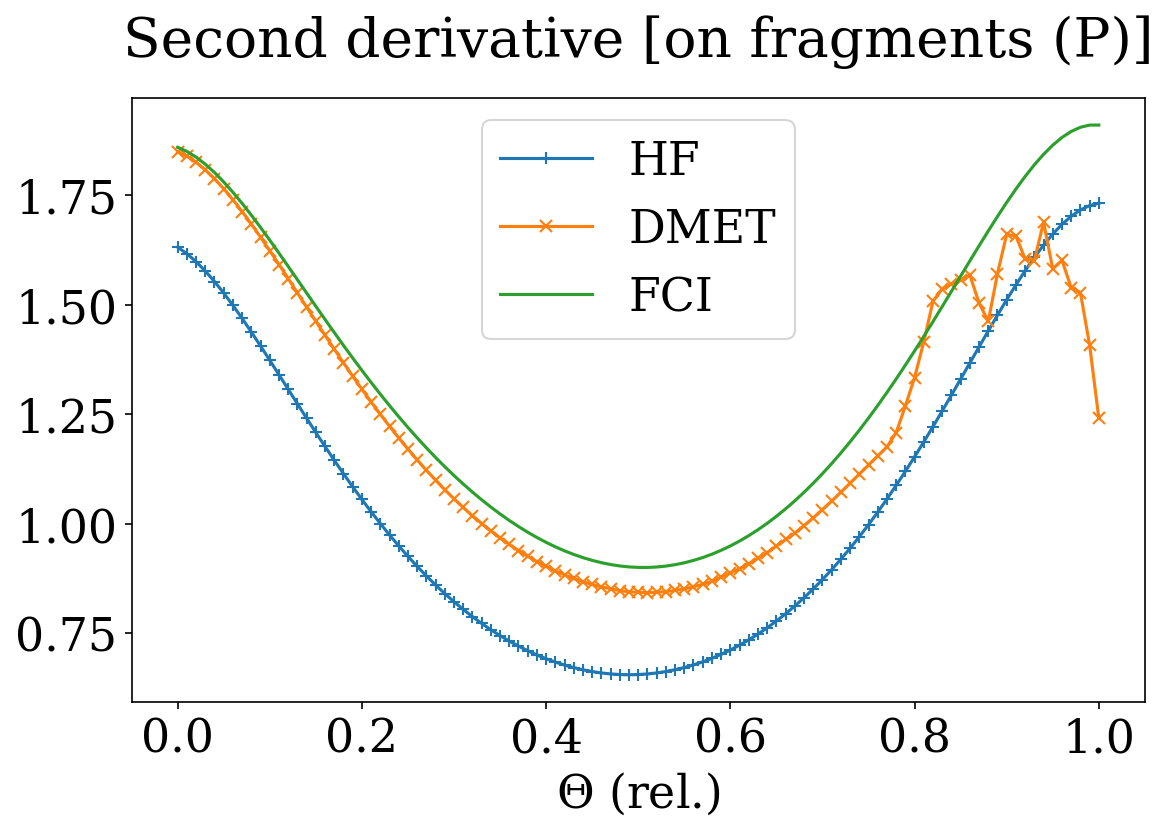}
    \caption{}
    \label{fig:H6_sderv}
\end{subfigure}
\hfill
\begin{subfigure}[b]{0.48\textwidth}
    \centering
    \includegraphics[width=\textwidth]{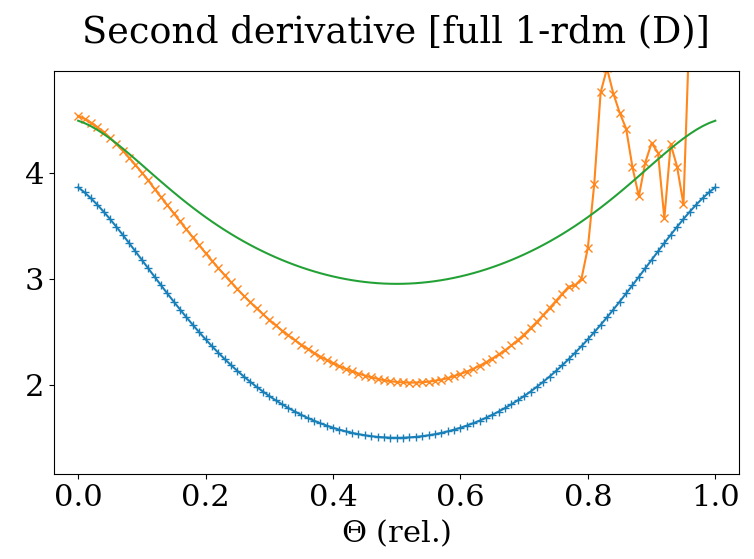}
    \caption{}
    \label{fig:H6_sderv_full}
\end{subfigure}
\caption{\label{fig:H6_sdervs}
(\subref{fig:H6_sderv}) shows $\Vert \partial^2_\alpha P_\alpha \big|_{\alpha = 0} \Vert_F$ for HF, DMET and FCI
(\subref{fig:H6_sderv_full}) shows $\Vert \partial^2_\alpha D_\alpha \big|_{\alpha = 0} \Vert_F$ for HF, DMET and FCI.
}
\end{figure}

We now investigate more closely the violation of the hypotheses at
$\Theta_{3}$, where $R$ is not invertible, but (A3) is still
satisfied. To that end, we compute the differential of $\DMETMap_0(\PGS)$
at $\PGS$, as a function of $\Theta$, and see that for $\Theta$ close to $\Theta_3$, $\DMETMap_0(P_0(\Theta))$ possesses a simple real eigenvalue 
which transitions from being positive (for $\Theta < \Theta_{3}$) to being negative (for $\Theta > \Theta_{3}$), with all other eigenvalues having negative real parts. As
is standard, this type of eigenvalue crossing generically gives rise
to a transcritical bifurcation \cite{strogatz2018nonlinear}. This
suggests the existence of another branch of solutions $P_{1}(\Theta)$ of $P = \DMETMap_0(P)$,
which collides with $P_{0}(\Theta)$ at $\Theta=\Theta_{3}$, and such
that the largest eigenvalue of the differential of $\DMETMap_0$ at
$P_{0}$ has the opposite sign to that at $P_{1}$.

To find this branch of solutions, we employ a Newton algorithm on
$\DMETMap_0$. Since we are looking at small differences, this requires
accurate computations of $\HLMap_{0}$ and $\LLMap_{0}$ as well as
their differentials (without resorting to finite differences). The
differential of $\HLMap_{0}$ is computed analytically by perturbation
theory (taking into account the self-consistent Fermi level). For
$\LLMap_{0}$, we implemented a manifold Newton algorithm to compute an
accurate solution of the problem defining the low-level solver. This
is done by, starting from the point $D_{n}$, parametrizing $D_{n+1}$
as $D(X)$ with an unconstrained matrix $X$ as in the proof of Lemma
\ref{lem:N-rep-loc}, and then performing a Newton step on the
Lagrangian $L(X,\Lambda)$ that corresponds to minimizing
$\EMF \left( D(X) \right)$ subject to $\PartitionProj(D(X))=P$. From
the Hessian of the Lagrangian one can also compute the differential of
$\LLMap_{0}$, and then ultimately of $\DMETMap_0$.

To initialize the Newton algorithm on $\DMETMap_{0}$ at a given
$\Theta$ close to $\Theta_{3}$, we start from $P_{0}$, and compute the eigenvector $Y$ of
$d\DMETMap_{0}$ associated with the eigenvalue that crosses zero. Then,
we run a Newton algorithm started from
$P_{0} + \alpha (\Theta-\Theta_{3}) Y$, where $\alpha$ is an
empirically chosen parameter (its precise determination involves
higher derivatives \cite{strogatz2018nonlinear}, which are cumbersome
to compute). We observe the two branches $P_{0}$ and $P_{1}$ shown in Figure
\ref{fig:bifurcation}, confirming the transcritical bifurcation. Let us emphasize that this bifurcation is not due to symmetry breaking, as can be shown from a detailed analysis of the solutions $P_{0}$ and $P_{1}$ (see Appendix~\ref{sec:appendixB}).

\begin{figure}[ht!]
    \centering
    \includegraphics[width=.6\textwidth]{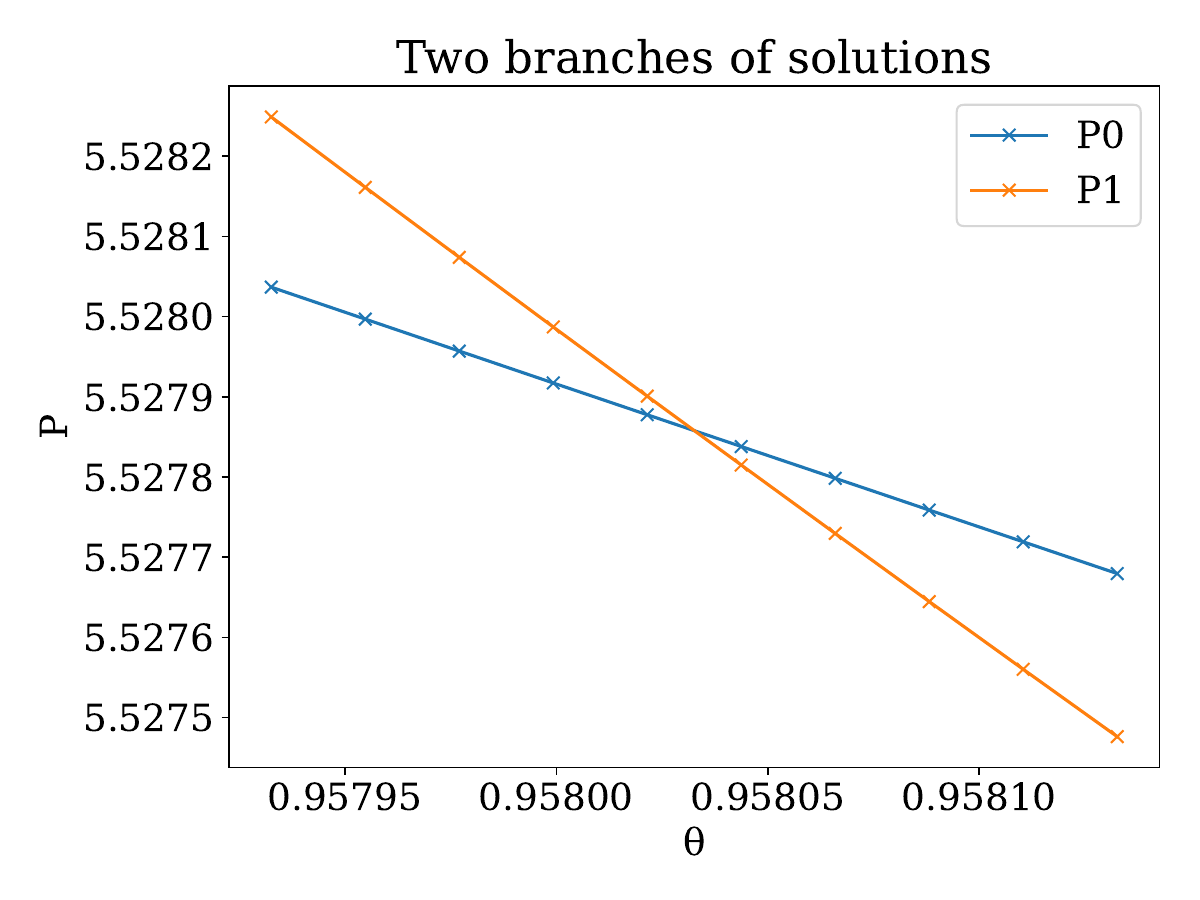}
    \caption{The two branches $P_{0}$ and $P_{1}$ (displayed are
      the scalars $\sum_{ij} P_{ij}$) as functions of $\Theta$ near $\Theta=\Theta_{3}$.
      %Placeholder, needs legend and styling
      }
    \label{fig:bifurcation}
\end{figure}

\pagebreak

\section{Impurity problems and high-level map}
\label{sec:hl_DMET}

\subsection{Impurity Hamiltonians}

It follows from the considerations in Section~\ref{sec:impurity_pb} that if 
\begin{equation}
\forall x \in \llbracket1,\NFrag\rrbracket, \quad \dim(D\Frag)=\dim((1-D)\Frag)= \FragSize. \nonumber
\end{equation}
the impurity problem is well-defined for each fragment since the maximal rank assumption~\eqref{eq:max_rank_assumption} is satisfied for each $\Frag$.
The next lemma gives useful equivalent characterizations of these conditions.

\medskip

Let us introduce the matrix
\begin{equation} \label{eq:def_Ex}
\MatFrag := \mat_{\AtomicBasis}(\AtomicVector, \kappa \in \FragIndic) =  \left( \begin{array}{c}  0_{\FragSize' \times \FragSize} \\  I_{\FragSize} \\  0_{\FragSize'' \times \FragSize} \end{array} \right) \in \R^{\DimH \times \FragSize}
\quad \mbox{with } \left\lbrace \begin{matrix} 
\FragSize':= \sum_{1 \leq x' < x } \FragSize[x']\\
\FragSize'' := \sum_{x < x' \le \NFrag} \FragSize[x'] \\
\end{matrix}\right.
\end{equation}
representing the orbitals of fragment $x \in \llbracket 1, \NFrag \rrbracket$, whose range is $X_{x}$.
We recall that $\dps\mathop{\mathcal P}^\circ$ denotes the interior of the set $\mathcal P=\PartitionProj(\CHGrass)$ in the affine space $\PGS+\cY$. 

\medskip

\begin{lemma}[Compatibility conditions] \label{lem:Omega_L}
Let $D \in \Grass$. The following assertions are equivalent: 
\begin{enumerate}
    \item $\PartitionProj(D) \in \dps\mathop{\mathcal P}^\circ$;
    \item $\forall x \in \llbracket1,\NFrag\rrbracket, \quad \dim(D\Frag)=\dim((1-D)\Frag)= \FragSize$;
    \item $\forall x \in \llbracket 1,\NFrag \rrbracket, \quad 0 < \MatFrag^T D \MatFrag < 1 \quad \mbox{(all the eigenvalues of  $\MatFrag^T D \MatFrag$ are in $(0,1)$)}$; 
    \item $\forall x \in \llbracket 1,\NFrag \rrbracket, \quad \MatFrag^T D \MatFrag \in \mathrm{GL}_\R(\FragSize) \quad \mbox{and} \quad \MatFrag^T (1-D) \MatFrag \in \mathrm{GL}_\R(\FragSize)$.
\end{enumerate}
If $D$ satisfies these conditions, we say that it is {\em compatible with the fragment decomposition}.
\end{lemma}

It is easily seen that if $D$ is compatible with the fragment decomposition, then the column vectors defined by the matrix
\begin{equation}\label{eq:matrix_C}
C^x(D) := \bigg(D\MatFrag \left(\MatFrag^T D \MatFrag \right)^{-1/2} \, \bigg| (1-D) \MatFrag \left(\MatFrag^T (1-D) \MatFrag \right)^{-1/2} \bigg) \in \R^{\DimH \times 2\FragSize}
\end{equation}
form an orthonormal basis of the impurity one-particle state space $\Impurity{D}$ defined in \eqref{eq:ImpuritySubspaceDefinition}.
More precisely, the first $\FragSize$ columns of $C^x(D)$ form an orthonormal basis of $D\Frag$ and its last $\FragSize$ columns form an orthonormal basis of $(1-D)\Frag$. Likewise, the column vectors of the matrix
\begin{equation}\label{eq:matrix_C_tilde}
\widetilde C^x(D) := \bigg(\MatFrag  \, \bigg| (1-\Projector[x]) D \MatFrag \left(\MatFrag^T D (1-\Projector[x]) D \MatFrag \right)^{-1/2} \bigg) \in \R^{\DimH \times 2\FragSize}
\end{equation}
form an orthonormal basis of $\Frag\oplus (1-\Projector[x])D\Frag$.

\medskip

We denote by $\Ann[j]^x(D)$ and $\Ann[j]^x(D)^\dagger$, $1 \le j \le 2 \FragSize$
the annihilation and creation operators in the basis of the columns
of $C^x(D)$:
\begin{align*}
  \Ann[j]^x(D) = \sum_{\kappa=1}^{\DimH}  (C^x(D))_{\kappa j}\Ann[\kappa], \quad \Ann[j]^x(D)^\dagger = \sum_{\kappa=1}^{\DimH}  (C^x(D))_{\kappa j}\Ann[\kappa]^\dagger.
\end{align*}
These operators allow for an explicit form of the impurity Hamiltonian $\ImpurityHamiltonian{D}$ as follows.
\begin{proposition}[Impurity Hamiltonian]\label{prop:HimpDefinition}
Let $D \in \Grass$ be compatible with the fragment decomposition. The $x$-th impurity Hamiltonian $\ImpurityHamiltonian{D}$
is the operator on $\Fock(\Impurity{D})$ given by
\begin{align} \label{eq:impurity_Hamiltonian}
\ImpurityHamiltonian{D} &= E^{\mathrm{env}}_x(D)+ \sum_{i,j=1}^{2\FragSize} \left[ C^x(D)^T \left(h+J({\mathfrak D}^x(D))-K({\mathfrak D}^x(D))\right) C^x(D) \right]_{ij} \Ann[i](D)^\dagger \Ann[j](D)  \nonumber \\
& +\frac 12 \sum_{i,j,k,\ell=1}^{2\FragSize} [V^x(D)]_{ijkl}  \Ann[i](D)^\dagger \Ann[j](D)^\dagger \Ann[\ell](D)  \Ann[k](D),
\end{align}
where 
\begin{itemize}
\item the Coulomb and exchange matrices $J({\mathfrak D}^x(D)) \in \R^{\DimH \times \DimH}$ and $K({\mathfrak D}^x(D)) \in \R^{\DimH \times \DimH}$  for the $x$-th impurity are constructed from the density matrix
\begin{equation}\label{eq:Dxd}
{\mathfrak D}^x(D):= D - D\MatFrag (\MatFrag^TD\MatFrag)^{-1}\MatFrag^TD  \in {\rm Gr}(\NElec-\FragSize,\DimH);
\end{equation}
\item the rank-4 tensor $V^x(D)$ is given by 
\begin{equation} \label{eq:tensor_VxD}
[V^x(D)]_{ijkl} := \sum_{\kappa,\lambda,\nu,\xi=1}^\DimH V_{\kappa\lambda\nu\xi} [C^x(D)]_{\kappa i} [C^x(D)]_{\lambda j}  [C^x(D)]_{\nu k}  [C^x(D)]_{\xi \ell}; 
\end{equation}
\item the value of the (irrelevant) constant $E^{\mathrm{env}}_x(D)$ is given in \eqref{eq:E_env}.
\end{itemize}
\end{proposition}

Note that the matrix ${\mathfrak D}^x(D)$ is in fact the one-body density matrix associated with the Slater determinant $\WFcore$ (see Section~\ref{sec:impurity_pb}).

\subsection{Domain of the high-level map}

A matrix $D \in \Grass$ is in the domain of the high-level map $\HLMap$ formally defined in Section~\ref{sec:overview} if and only if 
\begin{enumerate}
\item $D$ is compatible with the fragment decomposition (see Lemma~\ref{lem:Omega_L}), in such a way that the impurity problem  \eqref{eq:ImpurityProblem} is well defined for each $x$;
\item the set 
\begin{align*}
M_D:=\bigg\{ \ChemicalPotential \in \R \; \bigg| \; &\forall x, \mbox{ the impurity problem~\eqref{eq:ImpurityProblem} has a unique ground-state 1-RDM $P_{x,D,\mu}$}, \\
& \mbox{and } \sum_{x=1}^{\NFrag} \Trace\left( \Projector[x] P_{x,D,\ChemicalPotential} \Projector[x] \right) = \NElec \bigg\}
\end{align*}
is non-empty;
\item the function 
\begin{equation}
\cF_D: M_D \ni \ChemicalPotential \mapsto \sum_{x=1}^{\NFrag} \Projector[x] P_{x,D,\ChemicalPotential} \Projector[x] \in \DiagonalBlockCHGrass \nonumber
\end{equation}
is a constant over $M_D$, which we denote by $\HLMap(D)$.
\end{enumerate}

In the proof of Theorem~\ref{thm:theory}, we will study $\HLMap_\alpha$ in the non-interacting ($\alpha=0$) and weakly interacting ($|\alpha|$ small ) cases.
We will see that in these regimes the domain of $\HLMap_\alpha$ contains a neighborhood of $\DGS$ in $\Grass$.

\section{$\NElec$-representability and low-level map}
\label{sec:ll_DMET}

In this section, we focus our study on the low level map defined in \eqref{eq:LowLevelMapDefinition}.
Clearly, \eqref{eq:LowLevelMapDefinition} has minimizers if and only if $P \in \PartitionProj(\Grass)$ (otherwise, the feasible set of the minimization problem is empty). 

\medskip

The next Lemma covers the extreme cases of minimal ($\NFrag=2$) and maximal ($\NFrag=\DimH$) numbers of fragments. 

\begin{lemma}[Global $\NElec$-representability] \label{lem:N-rep} $\,$
\begin{enumerate}
    \item If $\NFrag=\DimH$ (one site per fragment), then $\PartitionProj(\Grass[\NElec])=\PartitionProj(\CHGrass[\NElec])=\DiagonalBlockCHGrass$.
    \item If $\NFrag=2$ and $\DimH \ge 3$, then $\PartitionProj(\Grass[\NElec])\subsetneq \PartitionProj(\CHGrass[\NElec])=\DiagonalBlockCHGrass$.
    More precisely, 
    \begin{align*}
    & \PartitionProj(\Grass[\NElec])=\big\{P \in \DiagonalBlockCHGrass  \; \big|\; \forall 0 < n < 1, \; {\rm dim}({\rm Ker}(\Projector[1] P \Projector[1] -n)) = {\rm dim}({\rm Ker}(\Projector[2] P \Projector[2] -(1-n)) \big\}.
    \end{align*}
\end{enumerate}
\end{lemma}

Our analysis of the DMET method in the non-interacting and weakly perturbative settings relies on the following weaker $\NElec$-representability result.

\begin{definition}[Local $\NElec$-representability] Let $D \in \Grass$ be compatible with the fragment decomposition. We say that the local $\NElec$-representability condition is satisfied at $D$ if the linear map $\PartitionProj$ is surjective from $T_D\Grass$ to $\cY$.
\end{definition}

\medskip

Note that Assumption (A3) can be rephrased as: the local $\NElec$-representability condition is satisfied at $\DGS$.

\medskip

A necessary condition for the local $\NElec$-representability condition to be satisfied at some $D \in \GrassCompatiblePartition$ is that 
\begin{equation}\label{eq:NecCondLorRep}
\NElec(\DimH-\NElec) = \dim(\Grass) = \dim(\R^{N_{\rm v} \times \NElec}) \ge \dim(\cY) = \sum_{x=1}^{\NFrag} \frac{\FragSize(\FragSize+1)}{2} - 1.
\end{equation}

If $\NFrag=\DimH$ (one site per fragment), the above condition reads $\NElec(\DimH-\NElec) \ge \DimH-1$, and is therefore  satisfied for any $1 \le \NElec \le \DimH-1$, i.e. for any non-trivial case. On the other hand, if $\NFrag=2$ and $\DimH=2\FragSize[1]=2\FragSize[2]$ (two fragments of identical sizes), the necessary condition reads $\NElec(\DimH-\NElec) \ge \frac{\DimH(\DimH+2)}4 - 1$ and is never satisfied as soon as $\DimH \ge 3$. This result is in agreement with the global $\NElec$-representability results in Lemma~\ref{lem:N-rep}. In  usual DMET calculations, condition  \eqref{eq:NecCondLorRep} is always satisfied, so that, generically, $\DiagonalBlockCHGrass$ and $\PartitionProj(\Grass)$ coincide in the neighbourhood of $\PGS$.
\medskip
  
The next lemma provides a sufficient local $\NElec$-representability criterion. 

\begin{lemma}[A local $\NElec$-representability criterion] \label{lem:N-rep-loc} 
Let $D \in \cD$ be compatible with the fragment decomposition (i.e. $D\in \GrassCompatiblePartition$). The following assertions are equivalent:
\begin{enumerate}
    \item the local $\NElec$-representability condition is satisfied at $D$;
    \item the only matrices $M \in \R^{\DimH \times \DimH}_{\rm sym}$ commuting with both $D$ and the matrices $\Projector[x]$ for all $1 \le x \le \NFrag$ are of the form $M = \lambda I_L$ for some $\lambda \in \R$;
    \item if $\Phi \in \R^{\DimH \times \DimH}$ is an orthogonal matrix such that 
    \begin{equation} \label{eq:D=PhiIPhi}
D = \Phi \left( \begin{array}{cc} I_\NElec & 0 \\ 0 & 0 \end{array} \right) \Phi^T,
\end{equation}
then the linear map 
\begin{equation}\label{eq:Nrepmap}
\R^{(\DimH-\NElec) \times \NElec} \ni X \mapsto \sum_{x=1}^{\NFrag} \Projector[x] \Phi \left( \begin{array}{cc} 0 & X^T \\ X & 0 \end{array} \right) \Phi^T \Projector[x] \in  \cY
\end{equation}
is surjective.
\end{enumerate}
\end{lemma}

The third assertion of Lemma~\ref{lem:N-rep-loc} gives a practical way to check the local $\NElec$-representability criterion: it suffices to (i) diagonalize $D$ in order to write it as in \eqref{eq:D=PhiIPhi} (the columns of $\Phi \in O(\DimH)$ form an orthonormal basis of eigenvectors of $D$), (ii) assemble the matrix of the linear map \eqref{eq:Nrepmap} in the canonical bases of $\R^{(\DimH-\NElec) \times \NElec}$ and $\cY$, and (iii) check whether the number of positive singular values of this matrix is equal to $\dim(\cY)=\sum_{x=1}^{\NFrag} \frac{\FragSize(\FragSize+1)}2 - 1$.

\section{Proofs}
\label{sec:proofs}

\subsection{Proof of Lemma~\ref{lem:Omega_L}}

Let $D \in \Grass$. 

\medskip

\noindent
2) $\iff$ 3). Assume that
\begin{equation}
\forall 1 \le x \le \NFrag, \quad \dim(D\Frag)=\dim((1-D)\Frag)= \FragSize. \nonumber
\end{equation}
Since $D^2=D$, we have for all $y \in \R^{\FragSize}$,
\begin{equation}\label{eq:Dchi_rx}
y^T (\MatFrag^T D \MatFrag) y = 
y^T (\MatFrag^T D^2 \MatFrag) y= (D (\MatFrag y))^T  (D(\MatFrag y)) = |D(\MatFrag y)|^2,
\end{equation}
and therefore, 
\begin{equation}
0 \le y^T (\MatFrag^T D \MatFrag) y = |D(\MatFrag y)|^2 \le |\MatFrag y |^2 = |y|^2. \nonumber
\end{equation}
Thus $0 \le \MatFrag^T D \MatFrag \le 1$ in the sense of hermitian matrices.
Assume now that $y^T (\MatFrag^T D \MatFrag) y=0$.
Then, $\MatFrag y \in \Ker(D)$.
But we also have $\MatFrag y \in \Frag$.
Since $\dim (D \Frag)=\FragSize$, this implies that $y=0$.
Thus $0 < \MatFrag^T D \MatFrag$ in the sense of hermitian matrices.
Likewise, we have $\MatFrag^T D \MatFrag < 1$.
This proves that 2) $\implies$ 3).
Conversely, if for all $1 \le x \le \NFrag$, $0 < \MatFrag^T D \MatFrag$, we infer from \eqref{eq:Dchi_rx} that $D(\MatFrag y)=0$ implies $y=0$, hence that $\dim(D\Frag)=\FragSize$.
Likewise, $\MatFrag^T D \MatFrag < 1$ implies $\dim((1-D)\Frag)=\FragSize$.
Therefore, 3) $\implies$ 2).

\medskip

\noindent
3) $\iff$ 4). Since $0 < \MatFrag^T D \MatFrag$ is equivalent to $\MatFrag^T D \MatFrag \in {\rm GL}_\R(\FragSize)$ and $\MatFrag^T D \MatFrag < 1$ is equivalent to $\MatFrag^T(1-D)\MatFrag \in {\rm GL}_\R(\FragSize)$, we conclude that 3) $\iff$ 4).
\medskip

\noindent
Lastly, it follows from the definition of ${\mathcal P}$ that 
\begin{equation}
P =\left(\begin{array}{cccc}
P_1& 0 &\cdots & 0 \\
0 &P_2 &\cdots &0\\
\vdots & &\ddots &\vdots \\
0 &0 &\cdots & P_{\NFrag}\\
\end{array}\right) \in \dps\mathop{\mathcal P}^\circ \quad \iff \quad (\forall 1 \le x \le \NFrag, \; 0 < P_x=\MatFrag^T P \MatFrag < 1).
\end{equation}
This shows that 1) $\iff$ 3), which concludes the proof.

\subsection{Proof of Proposition~\ref{prop:HimpDefinition}}%AKMARKER

Let $D \in \Grass$ and $1 \le x \le \NFrag$. Let us first concatenate the matrix $C^x(D) \in \R^{\DimH \times 2\FragSize}$ introduced in~\eqref{eq:matrix_C} with a matrix $C^x_{\rm env}(D)\in \R^{\DimH \times (\DimH-2\FragSize)}$ in order to form an orthogonal matrix
$$
{\mathfrak C}^x(D) = (C^x(D)|C^x_{\rm env}(D)) \in O(\DimH).
$$
The column vectors of ${\mathfrak C}^x(D)$ define an orthonormal basis of $\HSpace=\R^\DimH$ adapted to the decomposition $\HSpace=\Impurity{D} \oplus \Env{D}$. The generators of the real CAR algebra associated with this basis are given by
\begin{align*}
    \Ann[i]^x(D)&=\sum_{\kappa=1}^\DimH {\mathfrak C}^x(D)_{\kappa i} \Ann[\kappa], \qquad 
    \Ann[i]^x(D)^\dagger =\sum_{\kappa=1}^\DimH {\mathfrak C}^x(D)_{\kappa i} \Ann[\kappa]^\dagger,
\end{align*}
so that the Hamiltonian 
$$
    \Hamiltonian= \sum_{\kappa, \lambda=1}^\DimH h_{\kappa \lambda} \Ann[\kappa]^\dagger \Ann[\lambda] + \frac{1}2  \sum_{\kappa, \lambda , \nu ,\xi=1}^\DimH V_{\kappa \lambda \nu \xi} \Ann[\kappa]^\dagger \Ann[\lambda]^\dagger \Ann[\xi]  \Ann[\nu]
$$
can be rewritten as 
\begin{align*}
    \Hamiltonian &=\sum_{i,j=1}^\DimH [h^x(D)]_{ij} \Ann[i]^x(D)^\dagger \Ann[j]^x(D) + \frac{1}2  \sum_{i,j,k,l=1}^\DimH [V^x(D)]_{ijkl} \, \Ann[i]^x(D)^\dagger \Ann[j]^x(D)^\dagger \Ann[l]^x(D) \Ann[k]^x(D)
    \end{align*}
    with 
    \begin{align*}
    [h^x(D)]_{ij}&:=\sum_{\kappa, \lambda=1}^\DimH h_{\kappa \lambda} {\mathfrak C}^x(D)_{\kappa i} {\mathfrak C}^x(D)_{\lambda j}  \quad \mbox{i.e.} \quad h^x(D)={\mathfrak C}^x(D)^T h{\mathfrak C}^x(D)
    \end{align*}
    and
    \begin{align*}
    [V^x(D)]_{ijkl}&:= \sum_{\kappa, \lambda,  \nu, \xi=1}^\DimH V_{\kappa \lambda \nu \xi} {\mathfrak C}^x(D)_{\kappa i} {\mathfrak C}^x(D)_{\lambda  j} {\mathfrak C}^x(D)_{\nu k} {\mathfrak C}^x(D)_{\xi l} .
\end{align*}
Note that if $1 \le i,j,k,l \le 2\FragSize$,
$$
[h^x(D)]_{ij}=[C^x(D)^ThC^x(D)]_{ij} \quad \mbox{and} \quad [V^x(D)]_{ijkl} := \sum_{\kappa, \lambda,  \nu, \xi=1}^\DimH V_{\kappa \lambda \nu \xi} C^x(D)_{\kappa i} C^x(D)_{\lambda  j} C^x(D)_{\nu k} C^x(D)_{\xi l},
$$
in agreement with \eqref{eq:tensor_VxD}.
Let $\WF \in \Fock(\HSpace)$ be of the form 
$$
\WF=\WFimpTrial \wedge \WFcore \quad \mbox{with} \quad \WFimpTrial \in \Fock(\Impurity{D})  \quad \mbox{and} \quad \WFcore \in \bigwedge^{(\NElec-\FragSize)} \Core{D}.
$$
We have
\begin{align*}
    \langle \WF | \Hamiltonian |\WF \rangle = 
    \big\langle \WFimpTrial\wedge \WFcore \big| &\sum_{i,j=1}^\DimH [h^x(D)]_{ij} \Ann[i]^x(D)^\dagger \Ann[j]^x(D) \\& + \frac{1}2  \sum_{i,j,k,l=1}^\DimH [V^x(D)]_{ijkl} \Ann[i]^x(D)^\dagger \Ann[j]^x(D)^\dagger \Ann[l]^x(D) \Ann[k]^x(D) \big| \WFimpTrial\wedge \WFcore \big\rangle .
\end{align*}
The terms in the Hamiltonian which change the number of particles in the impurity space or the environment do not contribute. The terms which act only on the environment subspace yield a term proportional to $\|\Psi^{\rm imp}_{x,D}\|^2$. Expanding the above expression, we thus obtain
\begin{align*}
    \langle \WF | \Hamiltonian |\WF \rangle =a_1+a_2+a_3 + a_4+a_5+a_6+a_7
    \end{align*}
    with 
    \begin{align*}
    a_1 :&=\sum_{i,j=1}^{2 \FragSize} [h^x(D)]_{ij} \langle \WFimpTrial\wedge \WFcore|  \Ann[i]^x(D)^\dagger \Ann[j]^x(D)|\WFimpTrial\wedge \WFcore \rangle \\
    & = \sum_{i,j=1}^{2 \FragSize} [h^x(D)]_{ij} \langle \WFimpTrial |  \Ann[i]^x(D)^\dagger \Ann[j]^x(D)|\WFimpTrial  \rangle \\
    & = \sum_{i,j=1}^{2 \FragSize} [C^x(D)^T h C^x(D)]_{ij} \langle \WFimpTrial |  \Ann[i]^x(D)^\dagger \Ann[j]^x(D)|\WFimpTrial  \rangle,
    \end{align*}
    \begin{align*}
    a_2:&=  \sum_{i,j=2\FragSize+1}^\DimH [h^x(D)]_{ij} \langle \WFimpTrial\wedge \WFcore| \Ann[i]^x(D)^\dagger \Ann[j]^x(D) | \WFimpTrial\wedge \WFcore\rangle \\
    &= \left( \sum_{i,j=2\FragSize+1}^\DimH [h^x(D)]_{ij} \langle  \WFcore| \Ann[i]^x(D)^\dagger \Ann[j]^x(D) |  \WFcore\rangle  \right)  \|\Psi^{\rm imp}_{x,D}\|^2, 
    \end{align*}
    \begin{align*}
    a_3:&= \sum_{i=2\FragSize+1}^\DimH \sum_{j=1}^{2\FragSize}  [h^x(D)]_{ij}\langle \!\!\underbrace{\WFimpTrial\wedge \WFcore}_{\substack{\text{$\FragSize$ part. in imp.}\\ \text{$(N-\FragSize)$ part. in env.}}} \!\!\! |\underbrace{ \Ann[i]^x(D)^\dagger \Ann[j]^x(D) | \WFimpTrial\wedge \WFcore\rangle}_{\substack{\text{ $(\FragSize+1)$ part. in imp.}\\ \text{$(N-\FragSize-1)$ part. in env.}}} \\
    &= 0, 
    \end{align*}
    \begin{align*}
    a_4:&= \sum_{i=1}^{2\FragSize} \sum_{j=2\FragSize+1}^\DimH   [h^x(D)]_{ij} \langle \!\!\underbrace{\WFimpTrial\wedge \WFcore}_{\substack{\text{$\FragSize$ part. in imp.}\\ \text{$(N-\FragSize)$ part. in env.}}} \!\!\! |\underbrace{ \Ann[i]^x(D)^\dagger \Ann[j]^x(D) | \WFimpTrial\wedge \WFcore\rangle}_{\substack{\text{ $(\FragSize-1)$ part. in imp.}\\ \text{$(N-\FragSize+1)$ part. in env.}}}  \\
    &= 0, 
    \end{align*}
    \begin{align*}
    a_5:&= \frac{1}{2} \sum_{i,j,k,l=1}^{2\FragSize}[V^x(D)]_{ijkl} \langle \WFimpTrial\wedge \WFcore|\Ann[i]^x(D)^\dagger \Ann[j]^x(D)^\dagger \Ann[l]^x(D) \Ann[k]^x(D)| \WFimpTrial\wedge \WFcore\rangle \\ 
    &= \frac{1}{2} \sum_{i,j,k,l=1}^{2\FragSize} [V^x(D)]_{ijkl} \langle \WFimpTrial |\Ann[i]^x(D)^\dagger \Ann[j]^x(D)^\dagger \Ann[l]^x(D) \Ann[k]^x(D)| \WFimpTrial \rangle,
    \end{align*}
    \begin{align*}
    a_6:&=\frac{1}{2} \sum_{i,j,k,l=2\FragSize+1}^{\DimH} [V^x(D)]_{ijkl} \langle \WFimpTrial\wedge \WFcore| \Ann[i]^x(D)^\dagger \Ann[j]^x(D)^\dagger \Ann[l]^x(D) \Ann[k]^x(D) | \WFimpTrial\wedge \WFcore\rangle \\
    &= \left( \frac{1}{2} \sum_{i,j,k,l=2\FragSize+1}^{\DimH} [V^x(D)]_{ijkl} \langle \WFcore| \Ann[i]^x(D)^\dagger \Ann[j]^x(D)^\dagger \Ann[l]^x(D) \Ann[k]^x(D) | \WFcore\rangle \right) \|\Psi^{\rm imp}_{x,D}\|^2,
    \end{align*} 
    \begin{align*}
    a_7:&=\frac{1}{2} \sum_{i,k=1}^{2\FragSize}  \sum_{j,l=2\FragSize+1}^{\DimH} ([V^x(D)]_{ijkl} \! - \! [V^x(D)]_{ijlk} \!  -  \! [V^x(D)]_{jikl} \! +  \! [V^x(D)]_{jilk}) \\ & \qquad \qquad \times \underbrace{\langle \WFimpTrial\wedge \WFcore| \Ann[i]^x(D)^\dagger \Ann[j]^x(D)^\dagger \Ann[l]^x(D) \Ann[k]^x(D) | \WFimpTrial\wedge \WFcore\rangle}_{\langle \WFimpTrial | \Ann[i]^x(D)^\dagger \Ann[k]^x(D)|\WFimpTrial\rangle \langle \WFcore|\Ann[j]^x(D)^\dagger \Ann[l]^x(D)| \WFcore\rangle}.
\end{align*}
Noticing that
\begin{align}
 \forall 2\FragSize+1 \le j, l \le \DimH, \quad    \langle \WFcore|\Ann[j]^x(D)^\dagger \Ann[l]^x(D)| \WFcore\rangle = ({\mathfrak C}^x(D)^T D {\mathfrak C}^x(D))_{jl}
\end{align}
we get
\begin{align*}
a_7&=\frac{1}{2} \sum_{i,k=1}^{2\FragSize}  \sum_{j,l=2\FragSize+1}^{\DimH} ([V^x(D)]_{ijkl} \! - \! [V^x(D)]_{ijlk} \!  -  \! [V^x(D)]_{jikl} \! +  \! [V^x(D)]_{jilk}) \\ & \qquad \qquad \times({\mathfrak C}^x(D)^T D {\mathfrak C}^x(D))_{jl}  \langle \WFimpTrial | \Ann[i]^x(D)^\dagger \Ann[k]^x(D)|\WFimpTrial\rangle \\
&= \sum_{i,j=1}^{2\FragSize} \bigg(  \sum_{k,l=2\FragSize+1}^{\DimH}  ([V^x(D)]_{ikjl} \! - \! [V^x(D)]_{iklj} \!  -  \! [V^x(D)]_{kijl} \! +  \! [V^x(D)]_{kilj}) ({\mathfrak C}^x(D)^T D {\mathfrak C}^x(D))_{kl}  \bigg)  \\ & \qquad \qquad \times\langle \WFimpTrial | \Ann[i]^x(D)^\dagger \Ann[j]^x(D)|\WFimpTrial\rangle .
\end{align*}
It holds for all $1 \le i,j \le 2\FragSize$,
\begin{align*}
& \sum_{k,l=2\FragSize+1}^{\DimH}  [V^x(D)]_{ikjl} ({\mathfrak C}^x(D)^T D {\mathfrak C}^x(D))_{kl} \\
&\; =  \sum_{k,l=2\FragSize+1}^{\DimH}  \sum_{\kappa,\lambda,\nu,\xi,\sigma,\tau=1}^\DimH V_{\kappa\lambda\nu\xi} [{\mathfrak C}^x(D)]_{\kappa i} [{\mathfrak C}^x(D)]_{\lambda k}  [{\mathfrak C}^x(D)]_{\nu j} [{\mathfrak C}^x(D))]_{\xi l}  [{\mathfrak C}^x(D)]_{\sigma k} D_{\sigma\tau} [{\mathfrak C}^x(D)]_{\tau l} \\
&\; =  \sum_{k,l=1}^{\DimH-2\FragSize}  \sum_{\kappa,\lambda,\nu,\xi,\sigma,\tau=1}^\DimH V_{\kappa\lambda\nu\xi} [C^x(D)]_{\kappa i} [C^x_{\rm env}(D)]_{\lambda k}  [C^x(D)]_{\nu j} [C^x_{\rm env}(D))]_{\xi l}  [C^x_{\rm env}(D)]_{\sigma k} D_{\sigma\tau} [C^x_{\rm env}(D)]_{\tau l} \\
& \; =   \sum_{\kappa,\nu=1}^\DimH [C^x(D)]_{\kappa i} \left(   \sum_{\lambda,\xi,\sigma,\tau=1}^\DimH V_{\kappa\lambda\nu\xi} \left( \sum_{k=1}^{\DimH-2\FragSize}   [C^x_{\rm env}(D)]_{\lambda k}  [C^x_{\rm env}(D)]_{\sigma k} \right) \right. \\
&  \qquad \qquad \qquad \qquad \qquad \qquad \qquad \times D_{\sigma\tau} \left.  \left( \sum_{l=1}^{\DimH-2\FragSize}   [C^x_{\rm env}(D)]_{\tau l}  [C^x_{\rm env}(D)]_{\xi l} \right)   \right) [C^x(D)]_{\nu j} \\
& \; =   \sum_{\kappa,\nu=1}^\DimH [C^x(D)]_{\kappa i} \left(   \sum_{\lambda,\xi,\sigma,\tau=1}^\DimH V_{\kappa\lambda\nu\xi} \left( C^x_{\rm env}(D) C^x_{\rm env}(D)^T D C^x_{\rm env}(D) C^x_{\rm env}(D)^T  \right)_{\lambda\xi} \right) [C^x(D)]_{\nu j} \\
& \;= \left[ C^x(D)^T J(\widetilde{\mathfrak D}^x(D)) C^x(D) \right]_{ij},
\end{align*}
with, recalling that ${\mathfrak C}^x(D) = (C^x(D)|C^x_{\rm env}(D))$ is an orthogonal matrix,
\begin{align*}
\widetilde {\mathfrak D}^x(D):&=  C^x_{\rm env}(D) C^x_{\rm env}(D)^T D C^x_{\rm env}(D) C^x_{\rm env}(D)^T \\
&= (1-C^x(D) C^x(D)^T) D (1-C^x(D) C^x(D)^T)  \\
&= D - D\MatFrag(\MatFrag^TD\MatFrag)^{-1}\MatFrag^TD \\
&= {\mathfrak D}^x(D) \qquad \mbox{(see \eqref{eq:Dxd})}.
\end{align*}
Using similar arguments, we get
$$
a_7=\sum_{i,k=1}^{2\FragSize}   \left[ C^x(D)^T (J({\mathfrak D}^x(D))-K({\mathfrak D}^x(D))) C^x(D) \right]_{ij}
\langle \WFimpTrial | \Ann[i]^x(D)^\dagger \Ann[j]^x(D)|\WFimpTrial\rangle.
$$
We finally obtain
$$
\langle \Psi |\Hamiltonian|\Psi\rangle = \langle \WFimpTrial|\ImpurityHamiltonian{D}|\WFimpTrial\rangle,
$$
where $\ImpurityHamiltonian{D}$ is given by \eqref{eq:impurity_Hamiltonian} with 
    \begin{align}
    E^\mathrm{env}(D)&= \sum_{i,j=2\FragSize+1}^\DimH  [h^x(D)]_{ij} \langle  \WFcore| \Ann[i]^x(D)^\dagger \Ann[j]^x(D) | \WFcore\rangle \nonumber \\ & \quad +\frac{1}{2} \sum_{i,j,k,l=2\FragSize+1}^\DimH [V^x(D)]_{ijkl} \langle  \WFcore| \Ann[i]^x(D)^\dagger \Ann[j]^x(D)^\dagger \Ann[l]^x(D) \Ann[k]^x(D) |\WFcore\rangle .\label{eq:E_env}
\end{align}

\subsection{Proof of Lemma~\ref{lem:N-rep}}

The first assertion is a direct consequence of \cite[Theorem
6]{Kasison2002}.

\medskip

We now prove the second assertion. Let
$$
{\mathcal K}:= \big\{P=(P_1,P_2) \in {\mathbb R}^{\FragSize[1] \times \FragSize[1]}_{\rm sym} \times {\mathbb R}^{\FragSize[2] \times \FragSize[2]}_{\rm sym} \; \big| 
\forall 0 < n < 1, \; {\rm dim}({\rm Ker}(P_1-n)) = {\rm dim}({\rm Ker}(P_2-(1-n)) \big\}. 
$$
Let $P=(P_1,P_2) \in \PartitionProj(\Grass)$ and 
$D \in \Grass$ be such that $\PartitionProj(D)=P$. 
Let $U_1$ and $U_2$ be two orthogonal matrices of sizes $(\FragSize[1] \times \FragSize[1])$ and $(\FragSize[2] \times \FragSize[2])$ respectively, and $D_1=\mbox{diag}(m_1,\cdots,m_{\FragSize[1]})$ and $D_2=\mbox{diag}(m_1',\cdots,m'_{\FragSize[2]})$ two diagonal matrices  with entries in the range $[0,1]$ ranked such that $m_1 \ge \cdots \ge m_{\rm \FragSize[1]}$ and $m_1' \le \cdots \le m_{\rm \FragSize[2]}'$, such that $P_1=U_1D_1U_1^T$ and $P_2=U_2D_2U_2^T$. It holds 
$$
D= \left( \begin{array}{cc} U_1 & 0 \\ 0 & U_2 \end{array} \right)   \left( \begin{array}{cc} D_1 & C \\ C^T & D_2 \end{array} \right) \left( \begin{array}{cc} U_1^T & 0 \\ 0 & U_2^T \end{array} \right)  \quad \mbox{for some } C \in {\mathbb R}^{\FragSize[1] \times \FragSize[2]}.
$$
The condition $D^2=D$ reads
$$
CC^T = D_1-D_1^2, \quad C^TC = D_2-D_2^2, \quad C-D_1C-CD_2=0,
$$
that is 
$$
\forall 1 \le i \le \FragSize[1], \quad \forall 1 \le j \le \FragSize[2], \quad \sum_{k=1}^{\FragSize[2]} C_{ik}^2 = m_i-m_i^2, \quad \sum_{k=1}^{\FragSize[1]} C_{kj}^2 = m'_j-{m_j'}^2, \quad (1-m_i-m'_j) C_{ij}=0.
$$
This implies that $C_{ij}=0$ unless $m'_j=1-m_j$ and that $C_{ij}=0$ whenever $m_i=0$ or $1$, or $m_j'=0$ or $1$. It follows that 
\begin{align}
& \left( \begin{array}{cc} U_1^T & 0 \\ 0 & U_2^T \end{array} \right)  D  \left( \begin{array}{cc} U_1 & 0 \\ 0 & U_2 \end{array} \right)  \nonumber \\
& \qquad = \left( \begin{array}{ccccc|ccccc} I_{r_1} &  &  &  &  &  &  \\  & n_1 I_{d_1} &  &  &  & &  C_1 \\   &  & \ddots &  &  &  &  & \ddots  \\  &  &  & n_\ell I_{d_\ell} &  &  &  & & C_\ell  \\  &  &  &  & 0_{s_1}  &  & \\ \hline
 & & & & & 0_{s_2}  &  &  &  &  \\  & C_1^T & &&  & & (1-n_1) I_{d_1'}  &  &  & \\   &  & \ddots & & & & &  \ddots &  &    \\  & & & C_\ell^T & & &  &  & (1-n_\ell) I_{d_\ell'} &      \\  &  &  &  & & & & & & I_{r_2} 
  \end{array} \right), \label{eq:decomp_D}
 \end{align}
 with $0 < n_\ell < \cdots < n_1 < 1$. Using again the idempotency of $D$, we obtain the relations $C_jC_j^T = n_j(1-n_j)I_{d_j}$ and  $C_j^TC_j = n_j(1-n_j)I_{d_j'}$. Taking the trace leads to $d_j=d_j'$. Therefore, $P \in {\mathcal K}$ so that $\PartitionProj(\Grass) \subset {\mathcal K}$. 
 \medskip
 
 Conversely, let $P \in {\mathcal K}$ and $U_1$, $U_2$, $D_1$, $D_2$ as before. Then $U_1^T P_1 U_1$ and $U_2^T P_2 U_2$ read as the diagonal blocks of the right-hand side of \eqref{eq:decomp_D} with $d_j=d_j'$ for all $j$. Setting $C_j=\sqrt{n_j(1-n_j)}I_{d_j}$, the matrix $D$ defined by \eqref{eq:decomp_D} is in ${\mathcal M}_{\rm S}$ and satisfies $\PartitionProj(D)=P$. Hence, $P \in \PartitionProj(\Grass)$ and therefore ${\mathcal K} \subset  \PartitionProj(\Grass)$. 

\subsection{Proof of Lemma~\ref{lem:N-rep-loc}}

Let $N_{\rm v}:=\DimH-N$. For $X \in \R^{N_{\rm v} \times N}$ such that $\|X\| < 1/2$, we set
\begin{align*}
f_\Phi(X):&=\Phi \left( \begin{array}{cc} \frac 12\left( I_{N} + (I_{N}-4X^TX)^{1/2} \right) & X^T \\ X & \frac 12\left( I_{N_{\rm v}} - (I_{N_{\rm v}}-4XX^T)^{1/2} \right)
\end{array} \right) \Phi^T, \\ g_\Phi(X):&=\PartitionProj(f_\Phi(X)).
\end{align*}
The map $f_\Phi$ provides a local system of coordinates of $\Grass$ in the vicinity of $D$. Therefore, the local $\NElec$-representability condition is satisfied at $D$ if and only if the map 
$$
d_0g_\Phi: \R^{N_{\rm v} \times N} \ni X \mapsto d_0g_\Phi = \sum_{x=1}^{\NFrag} \Projector[x] \Phi \left( \begin{array}{cc} 0 & X^T \\ X & 0 \end{array} \right) \Phi^T \Projector[x] \in \cY
$$
is surjective. This proves the equivalence between the first and third assertions of the lemma. 

\medskip

Writing $\Phi$ as $\Phi=(\Phi^{\rm occ}| \Phi^{\rm virt})$ with $\Phi^{\rm occ} \in \R^{\DimH \times \NElec}$ and $\Phi^{\rm virt} \in \R^{\DimH \times N_{\rm v}}$, the adjoint of $d_0g_\Phi$ is given by 
$$
d_0g_\Phi^* : \cY \ni Y \mapsto d_0g_\Phi^*(Y)=2{\Phi^{\rm virt}}^T Y \Phi^{\rm occ} \in \R^{N_{\rm v} \times N}.
$$
We therefore have for all $Y \in \cY$,
\begin{align}\label{eq:BB*}
(d_0g_\Phi d_0g_\Phi^*) Y &= 2 \sum_{x =1}^{\NFrag} \Projector[x] \left( (1-D)YD + DY(1-D) \right) \Projector[x], 
\end{align}
and therefore
\begin{align*}
\|d_0g_\Phi^*(Y)\|^2&= \Trace\left( Y (d_0g_\Phi d_0g_\Phi^*)(Y)) \right) = 2  \Trace \left( Y \sum_{x =1}^{\NFrag} \Projector[x] \left( (1-D)YD + DY(1-D) \right) \Projector[x] \right) \\
&= 2  \Trace \left(  \sum_{x =1}^{\NFrag} \Projector[x] Y \Projector[x] \left( (1-D)YD + DY(1-D) \right)  \right) \\
&= 2  \Trace \left(  Y \left( (1-D)YD + DY(1-D) \right)  \right) = 4 \| (1-D)YD \|^2.
\end{align*}
Thus 
$$
\forall Y \in \cY, \quad \|d_0g_\Phi^*(Y)\|=2\|(1-D)YD \|.
$$
The map $d_0g_\Phi$ is surjective if and only if its adjoint is injective. Thus the criterion is satisfied if and only if
$$
\forall Y \in \cY, \quad (1-D)YD=0 \quad \Rightarrow \quad Y = 0.
$$
As $D$ is an orthogonal projector, $(1-D)YD=0$ if and only if $Y$ commutes with $D$. In addition, a matrix $Y \in \R^{\DimH \times \DimH}_{\rm sym}$ is in $\cY$ if and only if (i) it commutes with all the $\Projector[x]$'s, and (ii) its trace is equal to $0$. Thus, the criterion is satisfied if and only if any zero trace matrix $Y \in \R^{\DimH \times \DimH}_{\rm sym}$ commuting with $D$ and the $\Projector[x]$'s is the null matrix. Lastly, this condition is equivalent to: any matrix $Y \in \R^{\DimH \times \DimH}_{\rm sym}$ commuting with $D$ and the $\Projector[x]$'s is of the form $\lambda I_L$ for some $\lambda \in \R$. This completes the proof of the second statement.

\subsection{Proof of Proposition~\ref{prop:DMET0}}

For $\alpha=0$, the low-level map is formally given by 
\begin{align} \label{eq:Fll_0}
\LLMap_{0}(P) &= \mathop{\rm argmin}_{D \in \Grass, \; \PartitionProj(D)=P} \Trace(hD) \qquad \mbox{(formal)}.
\end{align}
Under Assumption (A1) (i.e. $\varepsilon_N < 0 <  \varepsilon_{N+1}$), $\DGS$ is the unique minimizer of 
$$
\mathop{\rm argmin}_{D \in \Grass} \Trace(hD).
$$
Since $\PartitionProj(\DGS)=\PGS$ (by definition of $\PGS$), $\DGS$ is the unique minimizer of \eqref{eq:Fll_0} for $P=\PGS$. Thus, $\PGS$ is in the domain of $\LLMap_{0}$ and $\LLMap_0(\PGS)=\DGS$. 

\medskip

For $\alpha=0$, the high-level map takes the simple formal expression 
\begin{align*}
\HLMap_{0}(D) &= \sum_{x=1}^{\NFrag} \Projector[x] C^x(D) \1_{(-\infty,0]}\left( C^x(D)^T (h-\mu \Projector[x]) C^x(D)  \right) C^x(D)^T  \Projector[x] \qquad \mbox{(formal)},
\end{align*}
where $C^x(D)$ is defined in \eqref{eq:matrix_C} and  $\ChemicalPotential \in \R$ is such that
$$
\sum_{x=1}^{\NFrag}  \Trace\left(\Projector[x] C^x(D) \1_{(-\infty,0]}\left( C^x(D)^T  (h-\ChemicalPotential \Projector[x]) C^x(D) \right) C^x(D)^T  \Projector[x]   \right) = \NElec.
$$
Therefore, a matrix $D \in \Grass$ is in the domain of $\HLMap_{0}$ if and only if
\begin{enumerate}
\item the set
$$
M_D:=\left\{\ChemicalPotential \in \R \; \big| \; 
\sum_{x=1}^{\NFrag}  \Trace\left(\Projector[x] C^x(D) \1_{(-\infty,0]}\left( C^x(D)^T  (h-\ChemicalPotential\Projector[x]) C^x(D) \right) C^x(D)^T  \Projector[x]   \right) = \NElec \right\}
$$
is non-empty;
\item the function
$$
{\mathcal F}_D: M_D \ni \ChemicalPotential \mapsto \sum_{x=1}^{\NFrag} \Projector[x] C^x(D) \1_{(-\infty,0]}\left( C^x(D)^T (h-\ChemicalPotential \Projector[x]) C^x(D)  \right) C^x(D)^T  \Projector[x] \in \R^{\DimH \times \DimH}_{\rm sym}
$$
is constant over $M_D$. Its value is an element of $\DiagonalBlockCHGrass$, which we denote by $\HLMap_0(D)$.
\end{enumerate}

Let us prove that under Assumptions (A1) and (A2), $\DGS$ belongs to the domain of $\HLMap_0$ and $\HLMap_0(\DGS)=\PGS$.

\medskip

First, we observe that for each $1 \le x \le \NFrag$, the space $\Impurity{0} := \Frag + \DGS \Frag$ is $\DGS$-invariant since $\DGS$ is a projector. The linear operator $\DGS$ on $\R^\DimH$ therefore has a a block-diagonal operator representation in the decomposition $\Impurity{0} \oplus \Impurity{0}^\perp$ of $\HSpace=\R^\DimH$:
$$
\DGS \equiv  \left( \begin{array}{cc}
\DGS^x & 0 \\
0 & \widetilde {\DGS^x} 
\end{array} \right) \qquad \mbox{(in the decomposition $\HSpace = \Impurity{0} \oplus \Impurity{0}^\perp$)},
$$

where $\DGS^x$ and $\widetilde{\DGS^x}$ are both orthogonal projectors. The corresponding representation of $h$ is not necessarily block-diagonal:
$$
h \equiv  \left( \begin{array}{cc} h^x & h^x_{\rm OD} \\ {h^x_{\rm OD}}^T & \widetilde{h^x} \end{array} \right) \qquad \mbox{(in the decomposition $\HSpace = \Impurity{0} \oplus \Impurity{0}^\perp$)}.
$$

Let us now focus on the operator $h^x$. To lighten the notation, we set 
$$
D_{0,x}:=\MatFrag^T \DGS \MatFrag.
$$
We infer from Assumption~(A2) and Lemma~\ref{lem:Omega_L} that ${\rm dim}(\DGS \Frag)={\rm dim}((1-\DGS)\Frag)=\FragSize$ and  
$$
C^x_0:=C^x(\DGS)=\left(\DGS\MatFrag D_{0,x}^{-1/2}|  (1-\DGS)\MatFrag (1-D_{0,x})^{-1/2}\right)
$$
forms an orthonormal basis of $\Impurity{0}$. In this basis, the operator $h^x$ is represented by the matrix 
\begin{align}
{\mathfrak h}^x &:=  {C^x_0}^T h C^x_0 =  \left( \begin{array}{cc} {\mathfrak h}^x_- & 0 \\ 0 & {\mathfrak h}^x_+ \end{array} \right), \label{eq:matrix_hr}
\end{align}
with 
\begin{align}
{\mathfrak h}^x_-:&= D_{0,x}^{-1/2}  \MatFrag^T \DGS h \DGS \MatFrag D_{0,x}^{-1/2}, \label{eq:hr-} \\
{\mathfrak h}^x_+:&=  (1-D_{0,x})^{-1/2} \MatFrag^T (1-\DGS)h(1-\DGS) \MatFrag  (1-D_{0,x})^{-1/2}. \label{eq:hr+}
\end{align}
The zeros in the off-diagonal blocks of ${\mathfrak h}^x$ come from the fact that $\DGS h(1-\DGS)=(1-\DGS)h\DGS=0$ since $h$ and $\DGS$ commute. In addition, we have 
\begin{align}
\varepsilon_1 \DGS \le \DGS h\DGS &= \sum_{i=1}^N \varepsilon_i \phi_i \phi_i^T \le \varepsilon_N \DGS  \label{eq:bounds_D0hD0}, \\
\varepsilon_{N+1} (1-\DGS) \le (1-\DGS)h(1-\DGS)&= \sum_{a=N+1}^\DimH \varepsilon_a \phi_a \phi_a^T  \le \varepsilon_{\DimH} (1-\DGS). 
\label{eq:bounds_(1-D0)h(1-D0)}
\end{align}
Combining \eqref{eq:hr-} and \eqref{eq:bounds_D0hD0} on the one hand, and \eqref{eq:hr+} and \eqref{eq:bounds_(1-D0)h(1-D0)} on the other hand, we obtain
\begin{align}
\varepsilon_1 I_{\FragSize} \le {\mathfrak h}^x_-  \le \varepsilon_N I_{\FragSize} \qquad \mbox{and} \qquad 
\varepsilon_{N+1} I_{\FragSize} \le  {\mathfrak h}^x_+\le  \varepsilon_{\DimH} I_{\FragSize}. \label{eq:lower_bound_h+}
\end{align}
We therefore have
\begin{equation}\label{eq:1hx}
\1_{(-\infty,0]}({\mathfrak h}^x) = \1_{(-\infty,0)}({\mathfrak h}^x) = \left( \begin{array}{cc} I_{\FragSize} & 0 \\ 0 & 0 \end{array} \right) , \qquad \1_{[0,\infty)}({\mathfrak h}^x) = \1_{(0,\infty)}({\mathfrak h}^x) = \left( \begin{array}{cc} 0 & 0 \\ 0 & I_{\FragSize}  \end{array} \right),
\end{equation}
and thus
\begin{align*}
\sum_{r=1}^{\NFrag} \Projector[x] C_0^x \1_{(-\infty,0]}({\mathfrak h}^x) {C_0^x}^T \Projector[x] 
& = \sum_{r=1}^{\NFrag} \Projector[x] C_0^x \left( \begin{array}{cc} I_{\FragSize} & 0 \\ 0 & 0 \end{array} \right) {C_0^x}^T \Projector[x] \\
&= \sum_{r=1}^{\NFrag} (\MatFrag\MatFrag^T) \DGS \MatFrag  D_{0,x}^{-1} 
\MatFrag^T \DGS (\MatFrag\MatFrag^T) \\
&= \sum_{r=1}^{\NFrag} \Projector[x] \DGS \Projector[x] = \PartitionProj(\DGS) = \PGS.
\end{align*}
As $\Trace(\PGS)=N$, we have $0 \in M_{\DGS}$ and ${\mathcal F}_{\DGS}(0)=\PGS$. Let us now show that $M_{\DGS}=\{0\}$. It holds
$$
\Projector[x] \equiv  \left( \begin{array}{cc} \Projector[x]^x & 0 \\ 0 & 0 \end{array} \right) \qquad \mbox{(in the decomposition $\HSpace = \Impurity{0} \oplus \Impurity{0}^\perp$)},
$$
and in the basis defined of $\Impurity{0}$ defined by $C^x_0$, the orthogonal projector $\Projector[x]^x$ is represented by the matrix
\begin{equation}\label{eq:matrix_pix}
{\mathfrak p}^x := {C^x_0}^T \Projector[x] C^x_0 = \left( \begin{array}{cc} D_{0,x}  & D_{0,x}^{1/2}(1-D_{0,x})^{1/2} \\ (1-D_{0,x})^{1/2}D_{0,x}^{1/2}  & (1-D_{0,x})\ \end{array} \right). 
\end{equation}
We therefore have in particular ${{\mathfrak p}^x}^2={\mathfrak p}^x={{\mathfrak p}^x}^T$. Consider the function  
\begin{align*}
\R \ni \mu \mapsto  \zeta(\mu):&=\sum_{x=1}^{\NFrag}  \Trace\left(\Projector[x] C_x^0 \1_{(-\infty,0]}\left( {C_x^0}^T  (h-\mu \Projector[x]) C_x^0 \right) {C_x^0}^T  \Projector[x]   \right)
  \\
  &= \sum_{x=1}^{\NFrag}  \Trace\left({\mathfrak p}^x \1_{(-\infty,0]}\left( {\mathfrak h}^x-\mu{\mathfrak p}^x \right)  \right) \\
  & = \sum_{x=1}^{\NFrag}  \Trace\left({{\mathfrak p}^x} \1_{(-\infty,0]}\left( {\mathfrak h}^x-\mu{\mathfrak p}^x  \right)   {{\mathfrak p}^x}\right) \ge 0.
\end{align*}
We already know that $\zeta(0)=N$. We see from \eqref{eq:lower_bound_h+} that $0$ is not in the spectrum of ${\mathfrak h}$ for all~$x$. By a simple continuity argument, we obtain that for $|\mu|$ small enough, $0$ is not in the spectrum of ${\mathfrak h}^x-\mu{\mathfrak p}^x$ for all $x$. We therefore have 
\begin{equation}\label{eq:def_zeta}
\zeta(\mu)=\sum_{x=1}^{\NFrag}  \frac{1}{2\pi i} \oint_{\mathcal C} \Trace\left({\mathfrak p}^x \left( z-({\mathfrak h}^x-\mu{\mathfrak p}^x) \right)^{-1}  \right) \, dz \qquad \mbox{(for $|\mu|$ small enough)}, 
\end{equation}
where $\mathcal C$ is e.g. a circle in the complex plane, centered on the negative real axis, containing $0$ and of large enough radius. It follows that $\zeta$ is analytic in the vicinity of $0$ and that
\begin{align} \label{eq:zeta'0}
\zeta'(0)&= - \sum_{x=1}^{\NFrag}  \frac{1}{2\pi i} \oint_{\mathcal C} \Trace\left({\mathfrak p}^x \left( z-{\mathfrak h}^x \right)^{-1} {\mathfrak p}^x
  \left( z-{\mathfrak h}^x \right)^{-1}  \right) \, dz = \sum_{x=1}^{\NFrag} \langle {\mathfrak p}^x ,  {\mathfrak L}_x^+  {\mathfrak p}^x \rangle,
\end{align}
where ${\mathfrak L}_x^+$ is the linear operator on $\R^{2\FragSize \times 2\FragSize}_{\rm sym}$ defined by 
\begin{equation}\label{eq:Lx+}
\forall M \in \R^{2\FragSize \times 2\FragSize}_{\rm sym}, \quad {\mathfrak L}_x^+M = 
- \frac{1}{2\pi i} \oint_{\mathcal C}  \left( z-{\mathfrak h}^x \right)^{-1} M
  \left( z-{\mathfrak h}^x \right)^{-1}   \, dz,
\end{equation}
which can alternatively be defined by the linear response formula
\begin{equation}\label{eq:linear_response}
\1_{(-\infty,0]}({\mathfrak h}^x+ M) = \1_{(-\infty,0]}({\mathfrak h}^x) - {\mathfrak L}_x^+M + o(\|M\|).
\end{equation}
Let us diagonalize the real symmetric matrix ${\mathfrak h}^x$ as
$$
{\mathfrak h}^x = \sum_{n=1}^{2\FragSize} \widetilde \varepsilon_{x,n} \widetilde\phi_{x,n} \widetilde\phi_{x,n}^T \quad \mbox{with} \quad
\widetilde \varepsilon_{x,1} \le \cdots \le \widetilde \varepsilon_{x,2\FragSize}, \quad \widetilde\phi_{x,m}^T\widetilde\phi_{x,n}=\delta_{mn},
$$
with (using \eqref{eq:lower_bound_h+})
$$
\forall 1 \le i \le \FragSize, \quad \forall \FragSize \le a \le 2\FragSize, \quad \widetilde\varepsilon_{x,i} \le \varepsilon_N < 0 < \varepsilon_{N+1} \le \widetilde\varepsilon_{x,a}.
$$
 Using Cauchy residue formula, we get
\begin{align} \label{eq:L+OD}
\forall M  = \left( \begin{array}{cc} M^{--} & {M^{+-}}^T \\ M^{+-} & M_{++} \end{array} \right) \in \R^{2\FragSize \times 2\FragSize}_{\rm sym}, \quad  {\mathfrak L}_x^+M =   \left( \begin{array}{cc} 0 & N(M^{+-})^T \\ N(M^{+-}) & 0 \end{array} \right)
\end{align}
with 
\begin{equation}\label{eq:NM} \forall 1 \le m,n \le \FragSize, \quad 
[N(M^{+-})]_{mn}=\frac{[M^{+-}]_{mn}}{\widetilde\varepsilon_{x,m+\FragSize}-\widetilde\varepsilon_{x,n}}.
\end{equation}
The operator ${\mathfrak L}_x^+$ is self-adjoint and positive. Denoting by $\gamma:=\varepsilon_{N+1}-\varepsilon_N > 0$ the HOMO-LUMO gap, we have
\begin{align}
\forall M  = \left( \begin{array}{cc} M^{--} & M^{-+} \\ M^{+-} & M_{++} \end{array} \right) \in \R^{2\FragSize \times 2\FragSize}_{\rm sym}, \quad \langle M, {\mathfrak L}_x^+M \rangle & \ge  2\gamma^{-1} \|M^{-+}\|^2 \label{eq:LB_Liouvillian_PI}.
\end{align}
Indeed, we have
\begin{align*}
\langle M, {\mathfrak L}_x^+M \rangle = 2 \sum_{x=1}^{\NFrag}  \sum_{i=1}^{\FragSize}  \sum_{a=\FragSize+1}^{2\FragSize} \frac{|\widetilde\phi_{x,i}^T M \widetilde\phi_{x,a}|^2}{\widetilde\varepsilon_{x,a} -\widetilde\varepsilon_{x,i} } &\ge 2 \gamma^{-1} \sum_{x=1}^{\NFrag}  \sum_{i=1}^{\FragSize}  \sum_{a=\FragSize+1}^{2\FragSize} |\widetilde\phi_{x,i}^T M \widetilde\phi_{x,a}|^2 \\ & = 2 \gamma^{-1} \|\1_{(-\infty,0)}({\mathfrak h}^x)  M \1_{(0,+\infty)}({\mathfrak h}^x)\|^2 
= 2\gamma^{-1} \|M^{-+}\|^2 .
\end{align*}
Let 
$$
{\mathcal J}_0:= \left\{ \mu \in \R \; \bigg| \; \prod_{x=1}^{\NFrag} {\rm det}\left({\mathfrak h}^x-\mu{\mathfrak p}^x\right) =0 \right\}.
$$
Since $\mu \mapsto {\rm det}\left({\mathfrak h}^x-\mu{\mathfrak p}^x\right)$ is a polynomial of degree $\FragSize$, the set  ${\mathcal J}_0$ contains at most $\DimH$ elements. By similar arguments as above, the function $\zeta$ is real-analytic and non-decreasing on each connected components of $\R \setminus {\mathcal J}_0$. At each $\mu_0 \in {\mathcal J}_0$, the jump of $\zeta$ is given by
$$
\zeta(\mu_0+0)-\zeta(\mu_0-0) = \sum_{x=1}^{\NFrag}   \Trace\left({{\mathfrak p}^x} \1_{\{0\}}\left( {\mathfrak h}^x-\mu_0{\mathfrak p}^x  \right)   {{\mathfrak p}^x}\right) \ge 0.
$$
The function $\zeta$ is therefore nondecreasing on $\R$. As a consequence, the set $M_{\DGS}$ is an interval $I_{\DGS}$ containing $0$. Using \eqref{eq:matrix_pix}, \eqref{eq:zeta'0} and \eqref{eq:LB_Liouvillian_PI}, we get
\begin{align*}
\zeta'(0)& \ge 2 \gamma^{-1} \sum_{x=1}^{\NFrag} \|D_{0,x}^{1/2}(1-D_{0,x})^{1/2}\|^2 = 2 \gamma^{-1}  \sum_{x=1}^{\NFrag} \Trace(D_{0,x}(1-D_{0,x})) > 0,
\end{align*}
since, in view of Lemma~\ref{lem:Omega_L}, all the eigenvalues of the symmetric matrix $D_{0,x}(1-D_{0,x})$ are positive. Thus $M_{\DGS}=\{0\}$. This proves that $\DGS$ is in the domain of $\HLMap_0$ and that $\HLMap(\DGS)=\PGS$. 

\medskip

Combining this result with the previously established relation $\LLMap_0(\PGS)=\DGS$, we obtain that $\PGS$ is a fixed point of the DMET map for $\alpha=0$.

\subsection{Proof of Theorem~\ref{thm:theory}}

We endow $\Grass$ with the Riemannian metric induced by the Frobenius inner product on $\R^{\DimH \times \DimH}_{\rm sym}$.
For $\eta > 0$, we set
$$
\omega_\eta := \left\{ P \in \DiagonalBlockCHGrass \; | \; \|P-\PGS\| < \eta \right\} \quad \mbox{and} \quad \Omega_\eta := \left\{ D \in \Grass \; | \; \|D-\DGS\| < \eta \right\}.
$$

\subsubsection{Low-level map in the perturbative regime}
\label{sec:LLMPR}

Let us introduce the maps
\begin{align}
 &g: \Grass \to \cY && \mbox{s.t.} \quad  \forall D \in \Grass, \quad g(D):=\PartitionProj(D)-\PGS,  \nonumber \\
 &a : \Grass \to \R &&  \mbox{s.t.} \quad  \forall D \in \Grass, \quad a(D) := \Trace(hD), \nonumber \\
 &b: \Grass \to \R && \mbox{s.t.} \quad  \forall D \in \Grass, \quad b(D) := \frac 12 \Trace\left((J(D)-K(D)) D\right), \nonumber \\
 & E: \R \times \Grass \to \R  && \mbox{s.t.} \quad  \forall (\alpha,D) \in \R \times \Grass, \quad E(\alpha,D) := \EMF_\alpha(D)=a(D) + \alpha b(D). \nonumber
\end{align}
Since the maps $\PartitionProj,J,K: \R^{\DimH \times \DimH}_{\rm sym} \to \R^{\DimH \times \DimH}_{\rm sym}$ are linear, the maps $g$, $a$, $b$ and $E$ are real-analytic. With this notation,  we have
$$
\left( \mbox{Assumption~(A3)} \right) \iff \left( B:=d_{\DGS}g=\PartitionProj: T_{\DGS}\Grass \to \cY \mbox{ surjective}\right).
$$

\begin{lemma}[Low-level map in the perturbative regime] \label{lem:ll_perturbative} Under Assumptions (A1)-(A3), there exists $\alpha_{\rm LL} > 0$ and $0 < \eta_{\rm LL} < \frac 12$ such that 
\begin{enumerate}
\item $\omega_{\eta_{\rm LL}} \subset {\rm Dom}(\LLMap_\alpha)$ for all $\alpha \in (-\alpha_{\rm LL},\alpha_{\rm LL})$;
\item the function $(\alpha,P) \mapsto \LLMap_\alpha(P)$ is real-analytic on $(-\alpha_{\rm LL},\alpha_{\rm LL}) \times \omega_{\eta_{\rm LL}}^\cY$.
\end{enumerate}
\end{lemma}

\begin{proof} The first assertion means that for all $(\alpha,P)  \in (-\alpha_{\rm LL},\alpha_{\rm LL}) \times \omega_{\eta_{\rm LL}}$, the problem 
\begin{equation} \label{eq:min_EMF_alpha_M}
\min_{D \in \Grass \; | \; \PartitionProj(D)=P} \EMF_\alpha(D) = \min_{D \in \Grass \; | \; g(D)=P-\PGS} E(\alpha,D) 
\end{equation}
has a unique minimizer, which we denote by $\LLMap_\alpha(P)$. 

\medskip

Using Lemma~\ref{lem:N-rep-loc} and the submersion theorem, we deduce from Assumptions (A2)-(A3) that there exists $\eta > 0$ and $C \in \R_+$ such that for all $P \in \omega_\eta$, the set $\PartitionProj^{-1}(P)$ is nonempty and there exists $D_P \in \PartitionProj^{-1}(P)$ such that $\|D_P-\DGS \| \le C \|P-\PGS\|$. Let $D_{\alpha,P}$ be a minimizer of $\EMF_\alpha$ on $\PartitionProj^{-1}(P)$. Such a minimizer exists since $\EMF_\alpha$ is continuous on $\Grass$ and $\PartitionProj^{-1}(P)$ is a nonempty compact subset of $\Grass$, and satisfies the optimality conditions
\begin{equation} \label{eq:min_EMF_alpha_M_Euler}
\nabla_{\Grass} E(\alpha,D_{\alpha,P}) + d_{D_{\alpha,P}}g^* \Lambda_{\alpha,P} = 0, \quad g(D_{\alpha,P})=P-\PGS,
\end{equation}
where $\nabla_{\Grass} E(\alpha,D_{\alpha,P}) \in T_{D_{\alpha,P}} \Grass$ is the gradient at $D_{\alpha,P}$ of the function $\Grass \ni D \to E(\alpha,D) \in \R$ for the Riemannian metric induced with the Frobenius inner product, and $\Lambda_{\alpha,P} \in \cY$ the Lagrange multiplier of the constraint $g(D_{\alpha,P})=P-\PGS$.

Denoting by 
$$
C_{\rm nl}:= \frac 12 \max_{D \in \Grass} |\Trace((J(D)-K(D))D)|,
$$
we have 
\begin{align}
\EMF_\alpha(D_{\alpha,P})  &\le \EMF_\alpha(D_{P}) \le \EMF_0(D_{P}) + \alpha C_{\rm nl} \le \EMF_0(\DGS) + \|h\|  \|P-\PGS\| + \alpha C_{\rm nl} . \label{eq:EMF_alpha_UB}
\end{align}
To obtain a lower bound of $\EMF_\alpha(D_{\alpha,P})$, we use that
\begin{align*}
\forall D \in \Grass, \quad \EMF_0(D) = \Trace(hD) \ge \EMF_0(\DGS) + \frac \gamma 2 \|D-\DGS\|^2.
\end{align*}
This inequality is classical, but we recall its proof for the sake of completeness. For $M \in \R^{\DimH \times \DimH}_{\rm sym}$ we set
$$
M^{--}:=\DGS M \DGS, \quad M^{-+}:=\DGS M(1-\DGS), \quad M^{+-}:=(1-\DGS)M\DGS, \quad M^{++}:=(1-\DGS)M(1-\DGS).
$$
Let $D \in \Grass$ and $Q:=D-\DGS$. Since $\DGS=\1_{(-\infty,0]}(h)$, we have 
$$
h^{-+}=h^{+-}=0, \quad h^{--} \le \varepsilon_N, \quad h^{++} \ge \varepsilon_{N+1}, \quad Q^{++} \ge 0, \quad Q^{--} \le 0,
$$
and we deduce from the fact that both $D$ and $\DGS$ are rank-$\NElec$ orthogonal projectors that
$$
 Q^2=Q^{++}-Q^{--} \quad \mbox{and} \quad  \Trace(Q^{++})+\Trace(Q^{--})=0.
$$
Combining all the above properties, we obtain
\begin{align}
 \forall D \in \Grass, \quad a(D)  & =\Trace(hD) \nonumber \\ &= \Trace(h\DGS) +\Trace(h(D-\DGS)) \nonumber  \\
& = a(\DGS) + \Trace\left( h^{++} Q^{++} \right) + \Trace \left( h^{--} Q^{--} \right) \nonumber \\
& \ge a(\DGS) + \varepsilon_{N+1} \Trace\left(  Q^{++} \right) + \varepsilon_{N} \Trace \left(Q^{--} \right) \nonumber \\
&= a(\DGS) + \frac\gamma 2 \Trace\left(  Q^{++} - Q^{--} \right) \nonumber\\
&= a(\DGS) + \frac \gamma 2  \|D-\DGS\|^2. \label{eq:a_coercive}
\end{align}
As $\EMF_0(D)=a(D)$, this implies that
\begin{align*}
&\EMF_\alpha(D_{\alpha,P}) \ge  \EMF_0(D_{\alpha,P}) - \alpha C_{\rm nl} \ge \EMF_0(\DGS)+ \frac \gamma 2 \|D_{\alpha,P}-\DGS\|^2 - \alpha C_{\rm nl} .
\end{align*}
Combining this result with \eqref{eq:EMF_alpha_UB}, we obtain
$$
\|D_{\alpha,P}-\DGS\|^2 \le  2\gamma^{-1} \left( 2\alpha C_{\rm nl} +  \|h\| P-\PGS \| \right).
$$
This implies in particular that for $|\alpha|$ and $\|P-\PGS\|$ small enough, any minimizer $D_{\alpha,P}$ of \eqref{eq:min_EMF_alpha_M} is close to $\DGS$. To conclude, it suffices to prove that for $|\alpha|$ and $\|P-\PGS \|$ small enough, \eqref{eq:min_EMF_alpha_M_Euler} has a unique critical point close to $\DGS$. 
This leads us to introduce the function 
$$
\Theta : (\R \times {\mathcal P}) \times (\Grass \times \cY) \ni ((\alpha,P),(D,\Lambda)) \mapsto \Theta((\alpha,P),(D,\Lambda)) \in T_D\Grass \times \cY
$$
defined by 
$$
\Theta((\alpha,P),(D,\Lambda)):= \left( \nabla_{\Grass}E(\alpha,D) + (d_{D}g)^* \Lambda , g(D)-(P-\PGS) \right).
$$
As $\DGS$ is the unique minimizer of $D \mapsto E(0,D)$ on $\Grass$ and $\PGS=\PartitionProj(\DGS)$, we have $\nabla_{\Grass}E(0,\DGS)=0$ and $g(\DGS)=0$, so that 
$$
\Theta((0,\PGS),(\DGS,0))=(0,0).
$$
In addition, denoting by 
\begin{equation}\label{eq:defA}
A:= D^2_{\Grass}a(\DGS) : T_{\DGS} \Grass \to T_{\DGS} \Grass
\end{equation}
the Hessian at $\DGS$ of the function $a$ for the Riemannian metric induced by the Frobienius inner product, we have
$$
\forall (Q,\Lambda) \in T_{\DGS} \Grass  \times \cY, 
\quad \left[ d_{D,\Lambda}\Theta ((0,\PGS),(\DGS,0)) \right]
\begin{pmatrix}
  Q\\\Lambda
\end{pmatrix}
% (Q,\Lambda)
=
\begin{pmatrix}
  A&B^{*}\\
  B&0
\end{pmatrix}
\begin{pmatrix}
  Q\\\Lambda
\end{pmatrix},
$$
where we recall that $B:=d_{\DGS}g$. In view of \eqref{eq:a_coercive}, we have
\begin{equation}\label{eq:D2a_defpos}
\forall Q \in T_{\DGS} \Grass, \quad  \langle Q,AQ\rangle \ge \gamma \|Q\|^2.
\end{equation}
Since $A$ is coercive and $B: T_{\DGS} \Grass \to \cY$ is surjective,
it follows from the Schur complement formula that the map 
$$
d_{D,\Lambda}\Theta ((0,\PGS),(\DGS,0)) : T_{\DGS}\Grass \times \cY \to T_{\DGS}\Grass \times \cY
$$
is invertible. It follows from the real-analytic implicit function theorem on manifolds that there exists $\alpha_{\rm LL} > 0$, $\eta > 0$ and $\eta > 0$, such that for all $(\alpha,P) \in (-\alpha_{\rm LL},\alpha_{\rm LL}) \times \omega_{\eta}$, \eqref{eq:min_EMF_alpha_M_Euler} has a unique solution $(D_{\alpha,P},\Lambda_{\alpha,P})$ with $D_{\alpha,P} \in \omega_{\eta}$ and the map $(\alpha,P) \mapsto D_{\alpha,P}$ is real-analytic on $(-\alpha_{\rm LL},\alpha_{\rm LL}) \times \omega_{\eta}$.
\end{proof}

\subsubsection{High-level map in the perturbative regime} \label{sec:HL-PT}

The following result states that the high-level map $(\alpha,D) \mapsto \HLMap_\alpha(D)$ is well-defined and real-analytic on a neighborhood of $(0,\DGS)$. 

\begin{lemma}[High-level map in the perturbative regime] \label{lem:hl_perturbative} 
Under Assumptions (A1)-(A2), there exists $\alpha_{\rm HL} > 0$ and $0 < \eta_{\rm HL} < \frac 12$ such that 
\begin{enumerate}
\item  $\Omega_{\eta_{\rm HL}} \subset {\rm Dom}(\HLMap_\alpha)$ for all $\alpha \in (-\alpha_{\rm HL},\alpha_{\rm HL})$;
\item the function $(\alpha,D) \mapsto \HLMap_\alpha(D)$ is real-analytic on $(-\alpha_{\rm HL},\alpha_{\rm HL}) \times \Omega_{\eta_{\rm HL}}$.
\end{enumerate}
\end{lemma}

\begin{proof} For $D \in \Grass$ compatible with the fragment decomposition, we set
\begin{align}\label{eq:matrix_h_tilde}
[\widetilde h_x(D)]_{\kappa\lambda} &:=\left[ \widetilde C^x(D)^T h \widetilde C^x(D)\right]_{\kappa\lambda} =  \sum_{\kappa'\lambda' = 1}^{\DimH} [\widetilde C^x(D)]_{\kappa,\kappa'} [\widetilde C^x(D)]_{\lambda,\lambda'}  h_{\kappa'\lambda'}, \\
[\widetilde V_x(D)]_{\kappa\lambda\nu\xi}&:= \sum_{\kappa'\lambda'\nu'\xi' = 1}^{\DimH} [\widetilde C^x(D)]_{\kappa,\kappa'} [\widetilde C^x(D)]_{\lambda,\lambda'} [\widetilde C^x(D)]_{\nu,\nu'}  [\widetilde C^x(D)]_{\xi,\xi'} V_{\kappa'\lambda'\nu'\xi'},
\end{align}
where $\widetilde C^x(D)$ is defined in Lemma~\ref{lem:Omega_L}.
Denoting by $c_\kappa$, $c_\kappa^\dagger$, $1 \le \kappa \le 2\FragSize$ the generators of the CAR algebra on $\Fock(\R^{2\FragSize})$ associated with the canonical basis of $\R^{2\FragSize}$, the high-level map can be formally written as
\begin{align} \label{eq:def_Fhl_2_formal}
\HLMap_\alpha(D)=\sum_{x =1}^{\NFrag} \sum_{\kappa,\lambda=1}^{\FragSize} e_{\FragSize'+\kappa} \Trace_{\Fock(\R^{2N_x})}\left( \Gamma_{\alpha,x,D,\mu} c_\kappa^\dagger c_\lambda \right) e_{\FragSize'+\lambda}^T \quad \mbox{(formal)},
\end{align}
where $\Gamma_{\alpha,x,D,\mu} \in \LinearMap(\Fock(\R^{2\FragSize}))$ is the ground-state (many-body) density matrix associated with the grand-canonical impurity Hamiltonian
$$
\widetilde H_{\alpha,x,D,\mu}^{\rm imp}  := \sum_{\kappa,\lambda=1}^{2\FragSize} [\widetilde h_x(D)]_{\kappa\lambda} c_\kappa^\dagger c_\lambda + 
 \alpha \sum_{\kappa,\lambda,\nu,\xi=1}^{2\FragSize} [\widetilde V_x(D)]_{\kappa\lambda\nu\xi} c_\kappa^\dagger c_\lambda^\dagger c_\xi c_\nu -  \mu \sum_{\kappa=1}^{\FragSize} c_\kappa^\dagger c_\kappa,
$$
the parameter $\mu \in \R$ being chosen such that
$$
\sum_{x =1}^{\NFrag} \sum_{\kappa,\lambda=1}^{\FragSize}  \Trace_{\Fock(\R^{2N_x})}\left( \Gamma_{\alpha,x,D,\mu} c_\kappa^\dagger c_\lambda \right) =N.
$$
The results established in the proof of Proposition~\ref{prop:DMET0} can be rephrased as follows: under Assumptions (A1)-(A2), 
\begin{enumerate}
\item the impurity Hamiltonian $\widetilde H_{0,x,\DGS,0}^{\rm imp}$ has a non-degenerate ground-state for each $x$ and that it holds
$$
\sum_{x =1}^{\NFrag} \sum_{\kappa,\lambda=1}^{\FragSize}  \Trace_{\Fock(\R^{2N_x})}\left( \Gamma_{0,x,\DGS,0} c_\kappa^\dagger c_\lambda \right) = N;
$$
\item the function
$$
\R \ni \mu \mapsto \sum_{x =1}^{\NFrag} \sum_{\kappa,\lambda=1}^{\FragSize}  \Trace_{\Fock(\R^{2N_x})}\left( \Gamma_{0,x,\DGS,\mu} c_\kappa^\dagger c_\lambda \right) \in \R
$$
is non-decreasing, real-analytic in the neighborhood of $\mu=0$, and its derivative at $\mu=0$ is positive.
\end{enumerate}
Since the maps
$$
\Grass \ni D \mapsto[\widetilde h_x(D]_{\kappa\lambda} \in \R \quad \mbox{and} \quad \Grass \ni D \mapsto [\widetilde V_x(D)]_{\kappa\lambda\nu\xi} \in \R
$$
are real-analytic in the neighborhood of $\DGS$, we deduce from Kato's analytic perturbation theory and the implicit function theorem that there exists $\alpha_{\rm HL} > 0$, $\eta_{\rm HL} > 0$, and $\mu_{\rm HL} > 0$ such that 
\begin{enumerate}
\item for each $(\alpha,D,\mu) \in (-\alpha_{\rm HL},\alpha_{\rm HL}) \times \Omega_{\eta_{HL}} \times (-\mu_{\rm HL},\mu_{\rm HL})$, the impurity Hamiltonian $H^{\rm imp}_{\alpha,x,D,\mu}$ has a non-degenerate ground-state for each $x$; we denote by $\Gamma_{\alpha,x,D,\mu(\alpha,X)}$ the corresponding ground-state many-body density matrix;
\item for each $(\alpha,D) \in (-\alpha_{\rm HL},\alpha_{\rm HL}) \times \Omega_{\eta_{HL}}$, there exists a unique $\mu(\alpha,D) \in (-\mu_{\rm HL},\mu_{\rm HL})$ such that
$$
\sum_{x =1}^{\NFrag} \sum_{\kappa,\lambda=1}^{\FragSize}  \Trace_{\Fock(\R^{2N_x})}\left( \Gamma_{\alpha,x,D,\mu(\alpha,D)} c_\kappa^\dagger c_\lambda \right) = N;
$$
\item the maps $(\alpha,D) \mapsto \mu(\alpha,D)$, $(\alpha,D) \mapsto \Gamma_{\alpha,x,D,\mu(\alpha,D)}$, and 
$$
(\alpha,D) \mapsto \HLMap_\alpha(D):= \left(\sum_{x =1}^{\NFrag} \sum_{\kappa,\lambda=1}^{\FragSize} e_{\FragSize'+\kappa} \Trace_{\Fock(\R^{2N_x})}\left( \Gamma_{\alpha,x,D,\mu(\alpha,D)} c_\kappa^\dagger c_\lambda \right) e_{\FragSize'+\lambda}^T  \right)
$$
are real-analytic on $(-\alpha_{\rm HL},\alpha_{\rm HL}) \times \Omega_{\eta_{\rm HL}}$.
\end{enumerate}
This proves the two assertions of Lemma~\ref{lem:hl_perturbative}.
\end{proof}

\subsubsection{Existence, uniqueness, and analyticity}
\label{sec:EUA}

We infer from Lemma~\ref{lem:ll_perturbative} and Lemma~\ref{lem:hl_perturbative} that there exist $\alpha_{\rm DMET} > 0$ and $\eta_{\rm DMET} > 0$ such that the function 
$$
 (-\alpha_{\rm DMET},\alpha_{\rm DMET}) \times \omega_{\eta_{\rm DMET}} \ni \alpha, P \mapsto  \Phi(\alpha,P):=F_\alpha^{\rm DMET}(P)-P := F_\alpha^{\rm HL}(F_\alpha^{\rm LL}(P))-P \in \cY
$$
is well-defined and real-analytic, and we know from Proposition~\ref{prop:DMET0} that
$$
\Phi(0,\PGS)=0.
$$
To complete the proof of Theorem~\ref{thm:theory}, we have to check that the function $\Phi$ satisfies all the hypotheses of the implicit function theorem, namely that 
\begin{equation}\label{eq:dPPhi}
d_P\Phi(0,\PGS) = (d_{\DGS}F_0^{\rm HL}) \, (d_{\PGS}\LLMap_0) - I_\cY : \cY \to \cY
\end{equation}
is invertible.

\medskip

Let us first compute $d_{\PGS}\LLMap_0: \cY \to T_{\DGS}\Grass$. Differentiating the equality
$$
\forall P \in \omega_\eta, \quad 
\Theta((0,P),(\LLMap_0(P),\Lambda_{0,P})) = (0,0),
$$
we obtain that the derivatives at $\PGS$ of the functions $\omega_\eta \ni P \mapsto \LLMap_0(P) \in \Grass$ and $\omega_\eta \ni P \mapsto \lambda(P):=\Lambda_{0,P} \in \cY$ are characterized by the relation
$$
\forall Y \in \cY, \quad
\underbrace{[d_P\Theta((0,\PGS),(\DGS,0))] Y}_{=(0,-Y)}
+ \underbrace{[d_{(D,\Lambda)}\Theta((0,\PGS),(\DGS,0))]((d_{\PGS}\LLMap_0) Y, (d_{\PGS}\lambda) Y)}_{=(A[(d_{\PGS}\LLMap_0) Y]+B^*(d_{\PGS}\lambda) Y),B[(d_{\PGS}\LLMap_0) Y] } = 0,
$$
from which we infer that
\begin{equation} \label{eq:dFll}
d_{\PGS}\LLMap_0 = A^{-1}B^* (BA^{-1}B^*)^{-1}.
\end{equation}

Let us now compute $d_{\DGS}F_0^{\rm HL}: T_{\DGS}\Grass \to \cY$. We have
$$
\forall D \in \Omega_{\eta_{\rm HL}}, \quad F_0^{\rm HL}(D) = \sum_{x=1}^{\NFrag} \Projector[x] C^x(D) \1_{(-\infty,0]} \left( C^x(D)^T (h-\mu(0,D) \Projector[x]) C^x(D) \right) C^x(D)^T \Projector[x],
$$
where the function 
$$
\Grass \ni D \mapsto C^x(D)=(\underbrace{D\MatFrag(\MatFrag^TD\MatFrag)^{-1/2}}_{C^x_-(D)}| \underbrace{(1-D)\MatFrag(\MatFrag^T(1-D)\MatFrag)^{-1/2}}_{C^x_+(D)}) \in \R^{\DimH \times (2\FragSize)}
$$
has been introduced in \eqref{eq:matrix_C}. Setting as previously $C^x_0:=C^x(\DGS)$, and denoting by $M(Q):=[d_{\DGS}C^x](Q)$ and $\ell(Q):=[d_{D}\mu(0,\DGS)](Q)$, we get
\begin{align*}
d_{\DGS}\HLMap(Q) = &
\sum_{x=1}^{\NFrag} \Projector[x] \left( M(Q)  \1_{(-\infty,0]} \left( {\mathfrak h}^x \right) {C^x_0}^T + C^x_0\1_{(-\infty,0]} \left( {\mathfrak h}^x \right) M(Q)^T \right) \Projector[x] \\
& -  \sum_{x=1}^{\NFrag} \Projector[x]C^x_0   {\mathfrak L}_x^+ \left( M(Q)^T h C^x_0 + {C^x_0}^T h M(Q) - \ell(Q) {\mathfrak p}^x  \right) {C^x_0}^T \Projector[x].
\end{align*}
Using \eqref{eq:1hx}, we obtain
\begin{align*}
    M(Q)  \1_{(-\infty,0]} \left( {\mathfrak h}^x \right) {C^x_0}^T + C^x_0\1_{(-\infty,0]} \left( {\mathfrak h}^x \right) M(Q)^T &= [d_{\DGS}C^x_-(Q)] [C^x_-(\DGS)]^T + C^x_-(\DGS)[d_{\DGS}C^x_-(Q)]^T \\&= d_{\DGS}[C^x_-{C^x_-}^T](Q).
\end{align*}
This implies that 
\begin{align*}
\Projector[x] \left( M(Q)  \1_{(-\infty,0]} \left( {\mathfrak h}^x \right) {C^x_0}^T + C^x_0\1_{(-\infty,0]} \left( {\mathfrak h}^x \right) M(Q)^T \right) \Projector[x]&= 
d_{\DGS}[\Projector[x] C^x_-{C^x_-}^T\Projector[x]](Q).
\end{align*}
Since 
\begin{align*}
\Projector[x] C^x_-(D){C^x_-(D)}^T\Projector[x] &= (\MatFrag\MatFrag^T) (D\MatFrag(\MatFrag^TD\MatFrag)^{-1/2}) ((\MatFrag^TD\MatFrag)^{-1/2}\MatFrag^TD)  (\MatFrag\MatFrag^T) = \Projector[x] D \Projector[x],
\end{align*}
we get $d_{\DGS}[\Projector[x] C^x_-{C^x_-}^T\Projector[x]](Q)=\Projector[x]Q\Projector[x]$ and therefore
\begin{align*}
\sum_{x=1}^{\NFrag}\Projector[x] \left( M(Q)  \1_{(-\infty,0]} \left( {\mathfrak h}^x \right) {C^x_0}^T + C^x_0\1_{(-\infty,0]} \left( {\mathfrak h}^x \right) M(Q)^T \right) \Projector[x]&= \PartitionProj(Q) =BQ.
\end{align*}
Next, observing that for all $Q \in T_{\DGS}\Grass$,
\begin{align*}
d_{\DGS}C^x_-(Q)&= \DGS \MatFrag S_-(Q) + Q \MatFrag (\MatFrag^T\DGS \MatFrag)^{-1/2}, \\
d_{\DGS}C^x_+(Q)&= (1-\DGS) \MatFrag S_+(Q) - Q \MatFrag (\MatFrag^T(1-\DGS)\MatFrag)^{-1/2},
\end{align*}
with $Q \mapsto S_\pm(Q) \in \R^{\FragSize \times \FragSize}$ linear and 
\begin{equation}\label{eq:QinTD}
Q=\DGS Q(1-\DGS)+(1-\DGS)Q\DGS,
\end{equation}
we obtain that 
\begin{align*}
M(Q)^T h C^x_0 + {C^x_0}^T h M(Q) &= \left( \begin{array}{cc} * & N(Q)^T \\ N(Q) & * \end{array} \right) 
\end{align*}
with
\begin{align*}
N(Q):&= (\MatFrag^T(1-\DGS)\MatFrag)^{-1/2}\MatFrag^T \left( (1-\DGS)hQ - Qh\DGS \right) \MatFrag(\MatFrag^T \DGS \MatFrag)^{-1/2} \\
&=  (\MatFrag^T(1-\DGS)\MatFrag)^{-1/2}\MatFrag^T (1-\DGS)[h,Q]  \DGS \MatFrag(\MatFrag^T \DGS \MatFrag)^{-1/2}. 
\end{align*}
We thus have
\begin{align*}
M(Q)^T h C^x_0 + {C^x_0}^T h M(Q) &= \left( \begin{array}{cc} * & 0 \\ 0 & * \end{array} \right) - {C_0^x}^T [\DGS,[h,Q]] C_0^x,
\end{align*}
which implies, using \eqref{eq:L+OD}, 
$$
{\mathfrak L}_x^+ \left( M(Q)^T h C^x_0 + {C^x_0}^T h M(Q) - \ell(Q) {\mathfrak p}^x  \right) = 
{\mathfrak L}_x^+ \left(-{C_0^x}^T [\DGS,[h,Q]] C_0^x - \ell(Q) {\mathfrak p}^x  \right).
$$
We therefore obtain
$$
d_{\DGS}F_0^{\rm HL} = B+L,
$$
with $L: T_{\DGS}\Grass \to \cY$ given by
\begin{equation} \label{eq:def_L}
\forall Q \in T_{\DGS}\Grass, \quad 
LQ:=  \sum_{x=1}^{\NFrag} \Projector[x]C^x_0   {\mathfrak L}_x^+ \left({C_0^x}^T [\DGS,[h,Q]] C_0^x + \ell(Q) {\mathfrak p}^x  \right) {C^x_0}^T \Projector[x].
\end{equation}
Combining with \eqref{eq:dFll}, and setting 
\begin{equation}\label{eq:def_R}
R:=LA^{-1}B^* : \cY \to \cY,
\end{equation}
we obtain
$$
d_{P}\Phi(0,\PGS)= (B+L) (A^{-1}B^* (BA^{-1}B^*)^{-1})-I_\cY = R (BA^{-1}B^*)^{-1}.
$$
To conclude, we just have to show that the map $R$ rigorously defined by \eqref{eq:def_R} actually coincides with the 4-point response function formally defined by \eqref{eq:def_R_formal} (the latter is bijective by Assumption (A4)). We have for all $Q \in T_{\DGS}\Grass$ and $Y \in \cY$,
\begin{align*}
\langle Q, B^* Y \rangle_{T_{\DGS}\Grass} & = \langle B Q, Y \rangle_{\cY} = \Trace((BQ)Y) \\
&=  \Trace\left( \left( \sum_{x=1}^{\NFrag} \Projector[x] Q \Projector[x]\right)Y \right) = 
\sum_{x=1}^{\NFrag} \Trace\left( \Projector[x] Q \Projector[x] Y \right) = \sum_{x=1}^{\NFrag} \Trace\left(  Q \Projector[x] Y \Projector[x] \right) \\
&= \Trace\left( Q \left( \sum_{x=1}^{\NFrag} \Projector[x] Y \Projector[x]\right) \right) = \Trace(QY)
= \Trace (Q (\underbrace{\DGS Y(1-\DGS)+(1-\DGS)Y\DGS)}_{\in T_{\DGS}\Grass}) \\ &= \langle Q, \DGS Y(1-\DGS)+(1-\DGS)Y\DGS \rangle_{T_{\DGS}\Grass}.
\end{align*}
Therefore 
\begin{equation} \label{eq:B*}
Y \in \cY, \quad B^*Y= \DGS Y(1-\DGS)+(1-\DGS)Y\DGS.
\end{equation}
By a classical calculation (see e.g.~\cite[Section 2.2]{CKL_2021}), we have
\begin{equation} \label{eq:A}
\forall Q \in \cY, \quad AQ = - [\DGS,[h,Q]].
\end{equation}
It is also easily checked that
\begin{align} \label{eq:B*Y}
{C^x_0}^T (B^*Y) C^x_0 = {C^x_0}^T  \left(\DGS Y(1-\DGS )+  (1-\DGS )Y\DGS \right) {C^x_0}  = \left( \begin{array}{cc} * & 0 \\ 0 & * \end{array} \right) +  {C^x_0}^T  Y {C^x_0}.
\end{align}
Putting together~\eqref{eq:L+OD} and \eqref{eq:def_L}-\eqref{eq:B*Y} yields
\begin{align} \label{eq:def_R_math}
RY & =  \sum_{x=1}^{\NFrag} \Projector[x] C^x_0 {\mathfrak L}_{x}^+ \left( {C^x_0}^T \left(Y-\widetilde \ell(Y) \Projector[x]\right) {C^x_0}\right) {C^x_0}^T \Projector[x], 
\end{align}
where 
\begin{equation}\label{eq:def_ltilde}
\widetilde \ell(Y):=\ell(A^{-1}B^*Y) = \Trace(G Y ) \quad \mbox{with} \quad G:= \sum_{x=1}^{\NFrag}  {C^x_0}  {\mathfrak L}_x^+ \left( {\mathfrak p}^x \right) {C^x_0}^T \in \R^{\DimH \times \DimH}_{\rm sym}.
\end{equation}
Using the notation introduced in \eqref{eq:def_R_formal}, we have
\begin{align*}
\widetilde \HLMap_{h+Y}(\DGS )
&=
 \sum_{x=1}^{\NFrag} \Projector[x] C^x_0 \1_{(-\infty,0]} \left( {C^x_0}^T (h+Y-\mu_{Y} \Projector[x]) C^x_0 \right) {C^x_0}^T \Projector[x] ,
\end{align*}
where $\mu_{Y} \in \R$ is chosen such that 
$\Trace(\widetilde \HLMap_{h+Y}(\DGS ))=N$. Using similar perturbation argument as in Section~\ref{sec:HL-PT}, one can check that $\widetilde \HLMap_{h+Y}(\DGS )$ is well-defined for $Y \in \cY$ small enough, and that
\begin{align*}
\widetilde \HLMap_{h+Y}(\DGS )
&=  
 \sum_{x=1}^{\NFrag} \Projector[x] C^x_0 \1_{(-\infty,0]} \left( {\mathfrak h}^x + ({C^x_0}^T (Y - \mu_{Y} \Projector[x]) C^x_0 \right) {C^x_0}^T \Projector[x] \\
 &= \widetilde \HLMap_{h}(\DGS ) +
 \sum_{x=1}^{\NFrag} \Projector[x] C^x_0 {\mathfrak L}_x^+ \left( {C^x_0}^T (Y + \mu_{Y} \Projector[x]) C^x_0 \right) {C^x_0}^T \Projector[x] + o(\|Y\|),
\end{align*}
with $\mu_Y=\widetilde\ell(Y)$ by particle conservation. This shows that the map $R$ defined by \eqref{eq:def_R_math}-\eqref{eq:def_ltilde} actually coincides with the 4-point response function in Assumption (A4).

\subsubsection{About Assumptions (A3) and (A4) in the one-site-per-fragment setting} \label{sec:A4}

Let us show that when $\NFrag=\DimH$, we have under Assumptions (A1)-(A2),
\begin{align*}
\mbox{(A3) are satisfied} &\implies  
\DGS  \mbox{ is an irreducible matrix} \iff  
\mbox{(A4) is satisfied}.
\end{align*}  
Throughout this section, we assume that (A1)-(A2) are fulfilled.

\medskip

Let us first show that (A3) implies that $\DGS $ is irreducible.
We deduce from the second assertion of Lemma~\ref{lem:N-rep-loc} that (A3) is satisfied if and only if the only matrices in $\R^{\DimH \times \DimH}_{\rm sym}$ which commute with $\DGS $ and all the $\Projector[x]$'s are the multiples of the identity matrix. When $\NFrag=\DimH$, the matrices in $\R^{\DimH \times \DimH}_{\rm sym}$ which commute with all the $\Projector[x]$ are the diagonal matrices. The diagonal matrices $\Lambda=\mbox{diag}(\lambda_1,\cdots,\lambda_L)$ which commute with $\DGS $ are the ones for which
$$
\forall 1 \le i,j \le \DimH, \quad \lambda_i [\DGS ]_{ij} = [\DGS ]_{ij} \lambda_j.
$$
If $\DGS $ was reducible, then one could find a permutation matrix $P \in O(\DimH)$ such that $P\DGS P^{-1}$ is a $2 \times 2$ block-diagonal matrix. The matrix $P\mbox{diag}(1,\cdots,1,2,\cdots,2)P^{-1}$, where the numbers of entries $1$ and $2$ match the sizes of the blocks of $P\DGS P^{-1}$, would then be a diagonal matrix which commutes with $\DGS $ and is not proportional to the identity matrix. We reach a contradiction. Thus, (A3) implies that $\DGS $ is irreducible.

\medskip

Let us now show the equivalence 
$$
\DGS  \mbox{ is an irreducible matrix} \iff  
\mbox{(A4) is satisfied}.
$$
We have for all $Y \in \cY$,
\begin{align*}
\|RY\|^2 &= \Trace\left( (RY) (RY) \right) \\
& = \sum_{x,x'=1}^{\NFrag} \Trace \left(  \Projector[x] C^x_0 {\mathfrak L}_{x}^+ \left( {C^x_0}^T \left(Y-\widetilde \ell(Y) \Projector[x]\right) {C^x_0}\right) {C^x_0}^T \Projector[x]  \Projector[x']  C^{x'}_0 {\mathfrak L}_{x'}^+ \left( {C^{x'}_0}^T \left(Y-\widetilde \ell(Y) \Projector[x']\right) {C^{x'}_0}\right) {C^{x'}_0}^T \Projector[x']  \right) \\
& = \sum_{x=1}^{\NFrag} \Trace \left(  \Projector[x] C^x_0 {\mathfrak L}_{x}^+ \left( {C^x_0}^T \left(Y-\widetilde \ell(Y) \Projector[x]\right) {C^x_0}\right) {C^x_0}^T \Projector[x] C^{x}_0 {\mathfrak L}_{x}^+ \left( {C^{x}_0}^T \left(Y-\widetilde \ell(Y) \Projector[x]\right) {C^{x}_0}\right) {C^{x}_0}^T \Projector[x]  \right) \\
& = \sum_{x=1}^{\NFrag} \Trace \left( {\mathfrak p}_{x} {\mathfrak L}_{x}^+ \left( {C^x_0}^T \left(Y-\widetilde \ell(Y) \Projector[x]\right) {C^x_0}\right)  {\mathfrak p}_{x}  {\mathfrak L}_{x}^+ \left( {C^{x}_0}^T \left(Y-\widetilde \ell(Y) \Projector[x]\right) {C^{x}_0}\right)   \right) \\
& = \sum_{x=1}^{\NFrag} \|{\mathfrak p}^x  {\mathfrak L}_{x}^+ \left( {C^x_0}^T \left(Y-\widetilde \ell(Y) \Projector[x]\right) {C^x_0}\right)  {\mathfrak p}^x\|^2.
\end{align*}
Using \eqref{eq:matrix_pix} and \eqref{eq:L+OD}-\eqref{eq:NM}, we obtain after straightforward algebraic manipulations that
\begin{align*}
(RY=0) &\iff \left( \forall 1 \le x \le \NFrag, \; {\mathfrak p}^x  {\mathfrak L}_{x}^+ \left( {C^x_0}^T \left(Y-\widetilde \ell(Y) \Projector[x]\right) {C^x_0}\right)  {\mathfrak p}^x =0\right) \\
& \iff \left( \forall 1 \le x \le \NFrag, \; (1-D_{0,x})^{1/2} \widetilde N_x(Y) D_{0,x}^{1/2} + D_{0,x}^{1/2} \widetilde N_x(Y)^T(1-D_{0,x})^{1/2} = 0
\right),
\end{align*}
with 
$$
\widetilde N_x(Y):=N\left( (1-D_{0,x})^{-1/2} \MatFrag^T (1-\DGS ) Y\DGS \MatFrag D_{0,x}^{-1/2}  - \widetilde\ell(Y)D_{0,x}^{1/2} (1-D_{0,x})^{1/2} \right). 
$$
In the case when $\NFrag=\DimH$, we have $\FragSize=1$ for all $x$, and thus, $D_{0,x}$ and $\widetilde N(Y)$ are scalar quantities. We then have in this special case by assumption (A2),
\begin{align*}
(RY=0) &\iff \left( \forall 1 \le x \le \NFrag, \; N_x(Y) = 0
\right) \iff \left( M y = \widetilde \ell(Y) z \right),
\end{align*}
where $y = (Y_{11}, \cdots, Y_{LL})^T \in \R^\DimH$, $z = (D_{0,1}(1-D_{0,1}), \cdots , D_{0,\DimH}(1-D_{0,\DimH}))^T \in \R^\DimH$, and $M \in \R^{\DimH \times \DimH}_{\rm sym}$ is the matrix with entries
$$
M_{xx}= [D_{0}]_{xx}-[D_{0}]_{xx}^2, \quad M_{xx'}= -[D_{0}]_{xx'}^2 \; \mbox{if } x \neq x'.
$$
Still by Assumption (A2), $\sum_{x=1}^{\NFrag} z_x > 0$, and therefore using the fact that $\DGS $ is an orthogonal projector (hence that $\sum_{x=1}^{\NFrag} [\DGS ]_{x,x'}^2=[\DGS ^2]_{xx}=[\DGS ]_{xx}$), we get
$$
(M y = \widetilde \ell(Y) z) \implies \left( \widetilde\ell(Y) = \frac{\sum_{x,x'=1}^{\NFrag} M_{x,x'} y_{x'} }{\sum_{x=1}^{\NFrag} z_x}
=  \frac{\sum_{x=1}^{\NFrag} [\DGS ]_{x,x} y_{x} - \sum_{x,x'=1}^{\NFrag} [\DGS ]_{x,x'}^2 y_{x'} }{\sum_{x=1}^{\NFrag} z_x}  = 0\right) .
$$
Therefore, 
\begin{align*}
(RY=0) &\iff \left( M y = 0 \right).
\end{align*}
The matrix $M$ is hermitian, diagonal dominant with positive diagonal elements and non-positive off-diagonal elements, and such that
$$
\forall 1 \le x \le \NFrag, \quad M_{xx}= - \sum_{x' \neq x} M_{xx'}.
$$
Therefore the kernel of $M$ is reduced to $\R (1,\cdots,1)^T$ if and only if $M$ is irreducible. Besides, we see from the expressions of the coefficients of $M$ and Assumption (A2) that $M$ is irreducible if and only if $\DGS $ is irreducible. We conclude that $R$ is injective, hence bijective, if and only if $\DGS $ is irreducible.

\subsection{Proof of Theorem~\ref{thm:DMET_HF}}

\subsubsection{Perturbation expansion in the Fock space}
\label{sec:PTE_1}

This calculation is classical in the physics and chemistry literature, but we report it here for the sake of completeness. Consider a family of Hamiltonians $(\widehat H_\alpha)_{\alpha \in \R}$ of the form 
$$
\widehat H_\alpha := \widehat H_0 + \alpha (\widehat W_1 + \widehat W_2)
$$
on the real Fock space ${\rm Fock}(\R^{N_b})$ where
$$
\widehat H_0 := \sum_{m,n=1}^{N_b} [h_0]_{mn} c_m^\dagger c_n \quad \mbox{and} \quad \widehat W_1 := \sum_{m,n=1}^{N_b} [W_1]_{mn} c_m^\dagger c_n 
$$
are one-body Hamiltonians and 
$$
\widehat W_2 := \frac 12 \sum_{m,n,p,q,=1}^{N_b} [W_2]_{mnpq} c_m^\dagger c_n^\dagger c_q c_p
$$
is a two-body Hamiltonian. 

\medskip

Let us provisionally assume that $h_0$ is diagonal, and more precisely that
$$
h_0 = \mbox{diag}(\varepsilon_{1}^0 , \cdots , \varepsilon_{N_b}^0) \quad \mbox{with} \quad \varepsilon_{1}^0 \le  \cdots \le \varepsilon_{\mathcal N}^0 < 0 < \varepsilon_{{\mathcal N}+1}^0 \le \cdots \varepsilon_{N_b}^0.
$$
This amounts to working in a molecular orbital basis set of the unperturbed one-body Hamiltonian $h_0$ and assuming that the Fermi level $\epsilon_{\rm F}$ for having $\mathcal N$ particles in the ground state can be chosen equal to zero.
The ground state $\Psi_0$ of $\widetilde H_0$ in the ${\mathcal N}$-particle sector then is unique and so is the one of $\widetilde H_\alpha$ for $\alpha$ small by perturbation theory. We have
$$
\Psi_0 = \frac{1}{\sqrt{{\mathcal N}!}} c_{\mathcal N}^\dagger \cdots c_1^\dagger |0\rangle, \quad E_0 := \langle \Psi_0|\widehat H_0 | \Psi_0\rangle =  \sum_{i=1}^{\mathcal N} \varepsilon_{i}^0.
$$
 Denoting by $d(\alpha)$ the ground-state one-body reduced density matrix of $\widehat H_\alpha$, the map $\alpha \mapsto d(\alpha)$ is real-analytic in the neighborhood of $0$ and 
$$
d(\alpha) = d_0+\alpha d_1 + O(\alpha^2) \quad \mbox{with} \quad d_0 :=  \left( \begin{array}{cc} I_{\mathcal N} & 0 \\ 0 & 0 \end{array} \right).
$$
In addition, we have
$$
[d_1]_{mn}=\langle \Psi_1 | c_m^\dagger c_n | \Psi_0 \rangle + \langle \Psi_0 | c_m^\dagger c_n | \Psi_1 \rangle,
$$
where $\Psi_1$ is the first-order perturbation of the ground-state wave-function $\Psi_0$, solution to
$$
(\widehat H_0 -E_0) \Psi_1 = - \Pi_{\Psi_0^\perp}\left( (\widehat W_1+\widehat W_2) \Psi_0 \right), \qquad \Psi_1 \in \Psi_0^\perp.
$$
For $1 \le i_1 < \cdots < i_r \le  {\mathcal N}$ (occupied orbitals) and $m+1 \le a_1 < \cdots < a_r \le N_b$ (virtual orbitals), we set
$$
\Phi_0^0 := \Psi_0 \quad \mbox{and} \quad \Phi_{i_1\cdots i_r}^{a_1\cdots a_r} = c_{a_r}^\dagger \cdots c_{a_1}^\dagger c_{i_1}\cdots c_{i_r} \Phi_0^0 .
$$
The $\Phi_{i_1\cdots i_r}^{a_1\cdots a_r}$'s ($0 \le r \le \min({\mathcal N},N_b-{\mathcal N})$, $1 \le i_1 < \cdots < i_r \le {\mathcal N}$, $a_1 < \cdots < a_r \le N_b$, form an orthonormal basis of eigenfunctions of the restriction of $\widehat H_0$ to the $\mathcal{N}$-particle sector and it holds
$$
\widehat H_0 \Phi_{i_1\cdots i_r}^{a_1\cdots a_r} = E_{i_1\cdots i_r}^{a_1\cdots a_r}  \Phi_{i_1\cdots i_r}^{a_1\cdots a_r} \quad \mbox{with} \quad 
E_{i_1\cdots i_r}^{a_1\cdots a_r} = E_0 + \sum_{s=1}^r \varepsilon_{a_s} - \sum_{s=1}^r \varepsilon_{i_s}.
$$
We thus obtain the sum-over-state formula
$$
\Psi_1 = - \sum_{1 \le r \le \min({\mathcal N},N_b-{\mathcal N})} \sum_{1 \le i_1 < \cdots < i_r \le {\mathcal N}} \sum_{{\mathcal N}+ 1 \le a_1 < \cdots < a_r \le N_b} 
\frac{\langle \Phi_{i_1\cdots i_r}^{a_1\cdots a_r} | \widehat W_1+\widehat W_2 | \Phi_0^0 \rangle}{E_{i_1\cdots i_r}^{a_1\cdots a_r} - E_0} \Phi_{i_1\cdots i_r}^{a_1\cdots a_r},
$$
yielding
\begin{align*}
[d_1]_{mn}=  - \sum_{1 \le r \le \min({\mathcal N},N_b-{\mathcal N})} \sum_{1 \le i_1 < \cdots < i_r \le {\mathcal N}}&  \sum_{{\mathcal N}+1 \le a_1 < \cdots < a_r \le N_b} 
\frac{\langle \Phi_{i_1\cdots i_r}^{a_1\cdots a_r} | \widehat W_1+\widehat W_2 | \Phi_0^0 \rangle}{E_{i_1\cdots i_r}^{a_1\cdots a_r} - E_0} \\
& \times 
\left( \langle \Phi_{i_1\cdots i_r}^{a_1\cdots a_r} | c_m^\dagger c_n | \Phi_0^0\rangle + 
 \langle  \Phi_0^0 | c_m^\dagger c_n | \Phi_{i_1\cdots i_r}^{a_1\cdots a_r} \rangle \right).
\end{align*}
Since $ \langle \Phi_{i_1\cdots i_r}^{a_1\cdots a_r} | a_m^\dagger a_n | \Phi_0^0\rangle =0$ if $r \ge 2$, and 
\begin{align*}
\langle \Phi_{i}^{a} | c_m^\dagger c_n | \Phi_0^0\rangle &= \delta_{n,i}\delta_{m,a}, \\
\langle \Phi_{i}^{a} | c_m^\dagger c_n^\dagger c_q c_p  | \Phi_0^0\rangle &= - \delta_{m,q}\delta_{n,i}\delta_{p,a} \delta_{q \le {\mathcal N}} + 
\delta_{m,p}\delta_{n,i}\delta_{q,a} \delta_{p \le {\mathcal N}} + \delta_{m,i}\delta_{n,q}\delta_{p,a} \delta_{q \le {\mathcal N}}  - \delta_{m,i}\delta_{n,p}\delta_{q,a} \delta_{p \le {\mathcal N}}, 
\end{align*} 
this expression reduces to
\begin{align*}
[d_1]_{mn}=  -  \sum_{i=1}^N &  \sum_{a={\mathcal N}+1}^{N_b} 
\frac{\langle \Phi_{i}^{a} | \widehat W_1+\widehat W_2 | \Phi_0^0 \rangle}{\varepsilon_a^0-\varepsilon_i^0} 
\left( \delta_{n=i}\delta_{m=a} + \delta_{m=i} \delta_{n=a}  \right).
\end{align*}
We obtain that $d_1$ is of the form
$$
d_1 =  \left( \begin{array}{cc} 0 & d_1^{+-} \\ {d_1^{+-}}^T & 0 \end{array} \right) \quad \mbox{with} \quad 
\forall 1\le i \le {\mathcal N} <  {\mathcal N}+1 \le a \le N_b, \quad   [d_1]_{ai} = \frac{\langle \Phi_{i}^{a} | \widehat W_1+\widehat W_2 | \Phi_0^0 \rangle}{\varepsilon_a^0-\varepsilon_i^0} .
$$
Finally, we have
\begin{align*}
[d_1]_{ai} & = \sum_{m,n=1}^{N_b} [W_1]_{mn} \frac{\langle \Phi_{i}^{a} | c_m^\dagger c_n   | \Phi_0^0 \rangle}{\varepsilon_a^0-\varepsilon_i^0}  +  \sum_{m,n,p,q=1}^{N_b} [W_2]_{mnpq} \frac{\langle \Phi_{i}^{a} | c_m^\dagger  c_n^\dagger c_q c_p | \Phi_0^0 \rangle}{\varepsilon_a^0-\varepsilon_i^0} \\
& =  \frac{ [W_1+ J_{W_2}(d_0)-K_{W_2}(d_0)]_{ai}}{\varepsilon_a^0-\varepsilon_i^0} ,
\end{align*}
where the direct and exchange operators are respectively given by
$$
[J_{W_2}(d)]_{mn} := \sum_{p,q=1}^{N_b}  [W_2]_{npmq} d_{pq} \quad \mbox{and} \quad [K_{W_2}(d)]_{mn} := \sum_{p,q=1}^{N_b}  [W_2]_{npqm}d_{pq}.
$$
Introducing the linear response operator $\mathfrak L^+_{h_0}$ such that
$$
\1_{(-\infty,\epsilon_{\rm F}]}(h_0+W) = \underbrace{\1_{(-\infty,\epsilon_{\rm F}]}(h_0+W)}_{=d_0} - \mathfrak L^+_{h_0} W + O(\|W\|),
$$
we finally obtain 
\begin{equation}\label{eq:d1_perturb}
d_1 = - \mathfrak L^+_{h_0}\left(W_1+ J_{W_2}(d_0)-K_{W_2}(d_0) \right),
\end{equation}
this formula remaining valid in the general case when $h_0$ is not {\it a priori} diagonal and $\epsilon_{\rm F}$ not {\it a priori} equal to zero. 

\subsubsection{Perturbation expansion of the DMET ground-state}
\label{sec:PTE_2}

Under Assumption (A1), the Hartree-Fock problem
$$
\mathop{\rm argmin}_{D \in \Grass} \EMF_\alpha(D)
$$
has a unique minimizer $D^{\rm HF}(\alpha)$ for $\alpha$ small enough and the map $\alpha \mapsto D^{\rm HF}(\alpha)$ is real-analytic in the neighborhood of $0$. This results from a straightforward application of nonlinear perturbation theory, which we do not detail here for the sake of brevity. We set $ P^{\rm HF}(\alpha):={\rm Bd}(D^{\rm HF}(\alpha))$, and
\begin{align*}
&D_1^{\rm exact}:=\frac{dD^{\rm exact}}{d\alpha}(0), \qquad D_1^{\rm HF}:=\frac{dD^{\rm HF}}{d\alpha}(0), \qquad D_1^{\rm DMET}:=\frac{dD^{\rm DMET}}{d\alpha}(0),\\
&P_1^{\rm exact}:=\frac{dP^{\rm exact}}{d\alpha}(0), \qquad P_1^{\rm HF}:=\frac{dP^{\rm HF}}{d\alpha}(0), \qquad P_1^{\rm DMET}:=\frac{dP^{\rm DMET}}{d\alpha}(0).
\end{align*}
We are going to prove that the above first three matrices on the one hand, and the last three ones on the other hand are equal in $T_{\DGS}\Grass$ and $\cY$ respectively. 

\medskip

First, we deduce from \eqref{eq:d1_perturb} applied with $N_b=\DimH$, $\epsilon_{\rm F}=0$, $h_0=h$, $W_1=0$, $W_2=v$, that
$$
D_1^{\rm exact} = - \mathfrak L^+_{h}\left(J(\DGS )-K(\DGS ) \right),
$$
where $J$ and $K$ are the direct and exchange operators for the two-body interaction potential $\widehat V$ introduced in \eqref{eq:direct_exchange}.

\medskip

Next, by differentiating the self-consistent equation 
$$
D^{\rm HF}(\alpha) = \1_{(-\infty,0]}\left(h^{\rm MF}(\alpha,D^{\rm HF}(\alpha))\right),
$$
where 
$$
h^{\rm MF}(\alpha,D) = h + \alpha \left( J(D)-K(D) \right)
$$
is the Fock Hamiltonian for the interaction parameter $\alpha$, we get
$$
D^{\rm HF}_1 = - \mathfrak L^+_{h}\left(J(\DGS )-K(\DGS ) \right).
$$
Hence 
$$
D^{\rm HF}_1=D_1^{\rm exact} \quad \mbox{and} \quad P^{\rm HF}_1=\PartitionProj(D^{\rm HF}_1) = \PartitionProj(D_1^{\rm exact})=P_1^{\rm exact}.
$$

\medskip

Let us now show that $P^{\rm DMET}_1=P^{\rm HF}_1$. For convenience, we will use the following notation
\begin{align*}
& \LLMap(\alpha,P) := \LLMap_\alpha(P), \qquad  \HLMap(\alpha,D) = \HLMap_\alpha(D),  \\
& \HLMap_{\rm HF}(\alpha,D) := \sum_{x=1}^{\NFrag} \Projector[x] C^x(D) \1_{(-\infty,0]}\left( C^x(D)^T  \left( h^{\rm MF}(\alpha,D) - \mu^{\rm HF}(\alpha,D) \Projector[x] \right) C^x(D) \right) C^x(D)^T \Projector[x],
\end{align*}
where $\mu^{\rm HF}(\alpha,D) \in \R$ is the Lagrange parameter of the charge conservation constraint. The map $\HLMap_{\rm HF}(\alpha,D)$ is the high-level Hartree-Fock map for the interacting parameter $\alpha$, introduced in Remark~\ref{rem:hl_HF} for $\alpha=1$.

\medskip

We know from Theorem~\ref{thm:theory} that for all $\alpha$ small enough
$$
\HLMap\left(\alpha,\LLMap\left(\alpha,P^{\rm DMET}(\alpha)\right)\right) = P^{\rm DMET}(\alpha).
$$
Taking the derivative at $\alpha=0$, we get
\begin{align} \label{eq:deriv_DMET}
\partial_\alpha \HLMap(0,\DGS ) + \partial_D \HLMap(0,\DGS ) \left( \partial_\alpha  \LLMap(0,\PGS) + \partial_P  \LLMap(0,\PGS)  P^{\rm DMET}_1 \right) = P^{\rm DMET}_1.
\end{align}

\medskip

The same arguments as in the proof of Proposition~\ref{prop:DMET0} allow one to show that for all $\alpha$ small enough
$$
\HLMap_{\rm HF}\left(\alpha,\LLMap\left(\alpha,P^{\rm HF}(\alpha)\right)\right) = P^{\rm HF}(\alpha),
$$
yielding
\begin{align} \label{eq:deriv_HF}
\partial_\alpha \HLMap_{\rm HF}(0,\DGS ) + \partial_D \HLMap_{\rm HF}(0,\DGS ) \left( \partial_\alpha  \LLMap(0,\PGS) + \partial_P  \LLMap(0,\PGS)  P^{\rm HF}_1 \right) = P^{\rm HF}_1.
\end{align}
Since $\HLMap_{\rm HF}(0,D)=\HLMap(0,D)$ for all $D$ in the neighborhood of $\DGS $, we have 
$$
\partial_P  \HLMap_{\rm HF}(0,\DGS )=\partial_P  \HLMap(0,\DGS ).
$$
Using \eqref{eq:dPPhi} and the invertibility of $d_P\Phi(0,\PGS)$ established in Section~\ref{sec:EUA}, we obtain 
\begin{align}
P^{\rm DMET}_1  &= - \left(d_P\Phi(0,\PGS)\right)^{-1} \left(  \partial_\alpha \HLMap(0,\DGS ) + \partial_D \HLMap(0,\DGS ) \partial_\alpha  \LLMap(0,\PGS) \right), \label{eq:PDMET1} \\
P^{\rm HF}_1  &= - \left(d_P\Phi(0,\PGS)\right)^{-1} \left(  \partial_\alpha \HLMap_{\rm HF}(0,\DGS ) + \partial_D \HLMap(0,\DGS ) \partial_\alpha  \LLMap(0,\PGS) \right) \label{eq:PHF1}.
\end{align}
Let us show that $\partial_\alpha  \HLMap(0,\DGS ) =\partial_\alpha  \HLMap_{\rm HF}(0,\DGS )$. On the one hand, we have
\begin{align*}
\HLMap_{\rm HF}(\alpha,\DGS )&= \sum_{x=1}^{\NFrag} \Projector[x] C^x_0 \1_{(-\infty,0]} \left( {C^x_0}^T \left( h+\alpha \left( J(\DGS )-K(\DGS ) \right) - \mu_{\rm HF}(\alpha,\DGS )  \Projector[x] \right) C^x_0 \right) {C^x_0}^T  \Projector[x],
\end{align*}
and therefore
\begin{align} \label{eq:dalphaFhlHF}
\partial_\alpha  \HLMap_{\rm HF}(0,\DGS ) &= - \sum_{x=1}^{\NFrag} \Projector[x] C^x_0 {\mathfrak L}_x^+ \left( {C^x_0}^T \left( J(\DGS )-K(\DGS ) \right) {C^x_0} - \partial \mu_{\rm HF}(0,\DGS )  {\mathfrak p}^x  \right) {C^x_0}^T  \Projector[x].
\end{align}
On the other hand, we have
\begin{align*} 
\HLMap(\alpha,\DGS )&= \sum_{x=1}^{\NFrag} \Projector[x] C^x_0  D_{x,\DGS }^{\rm imp}(\alpha) {C^x_0}^T  \Projector[x],
\end{align*}
where $D_{x,\DGS }^{\rm imp}(\alpha)$ is the ground-state one-body reduced density matrix in the basis of $Y_{x,\DGS }$ defined by $C^x_0$ of the impurity Hamiltonian (see Proposition~\ref{prop:HimpDefinition}) 
\begin{align*}
\ImpurityHamiltonian{\DGS }(\alpha) =& \sum_{i,j=1}^{2\FragSize} \left[ {C^x_0}^T \left(h+\alpha(J({\mathfrak D}^x(\DGS ))-K({\mathfrak D}^x(\DGS )))\right) C^x_0 \right]_{ij} \Ann[i](\DGS )^\dagger \Ann[j](\DGS )  \nonumber \\
& +\frac \alpha 2 \sum_{i,j,k,\ell=1}^{2\FragSize} [V^x(\DGS )]_{ijkl} \Ann[i](\DGS )^\dagger \Ann[j](\DGS )^\dagger \Ann[\ell](\DGS )  \Ann[k](\DGS ) \\
& -\mu(\alpha)  \sum_{i,j=1}^{2\FragSize}  \left[ {C^x_0}^T \Projector[x] C^x_0 \right]_{ij} \Ann[i](\DGS )^\dagger \Ann[j](\DGS ), 
\end{align*}
where $\mu(\alpha)$ is the Lagrange multiplier of the charge neutrality constraint and where we have discarded the irrelevant constant $E^{\rm env}_x(\DGS )$. Using the notation introduced in \eqref{eq:matrix_hr}, this Hamiltonian can be rewritten as
\begin{align*}
\ImpurityHamiltonian{\DGS }(\alpha) = &  \sum_{i,j=1}^{2\FragSize} [{\mathfrak h}^x]_{ij} \Ann[i](\DGS )^\dagger \Ann[j](\DGS ) \\
&+ \alpha \bigg( \sum_{i,j=1}^{2\FragSize} \left[  {C^x_0}^T  \left(J({\mathfrak D}^x(\DGS ))-K({\mathfrak D}^x(\DGS )) \right) C^x_0 \right]_{ij} \Ann[i](\DGS )^\dagger \Ann[j](\DGS )  \nonumber \\
& \qquad +\frac 1 2 \sum_{i,j,k,\ell=1}^{2\FragSize} [V^x(\DGS )]_{ijkl}  \Ann[i](\DGS )^\dagger \Ann[j](\DGS )^\dagger \Ann[\ell](\DGS )  \Ann[k](\DGS ) \bigg) \\
&-\mu(\alpha)  \sum_{i,j=1}^{2\FragSize}  \left[ {C^x_0}^T \Projector[x] C^x_0 \right]_{ij} \Ann[i](\DGS )^\dagger \Ann[j](\DGS ).
\end{align*}
We have
$$
D_{x,\DGS }^{\rm imp}(0) = \left( \begin{array}{cc} I_{\FragSize} & 0 \\ 0 & 0 \end{array} \right).
$$
Since $\mu(0)=0$ and $\alpha \mapsto \mu(\alpha)$ is real-analytic, we can easily adapt the analysis done in the previous section to the case when 
$$
N_b=2\FragSize, \quad h_0={\mathfrak h}^x, \quad 
W_1 = {C^x_0}^T  \left(J({\mathfrak D}^x(\DGS ))-K({\mathfrak D}^x(\DGS ))-\mu'(0) \Projector[x]\right) C^x_0, \quad W_2=V^x(\DGS ),
$$
and infer that 
\begin{align*}
    D_{x,\DGS }^{\rm imp}(\alpha)= D_{x,\DGS }^{\rm imp}(0) 
    - {\mathfrak L}_x^+\bigg( & {C^x_0}^T  \left(J({\mathfrak D}^x(\DGS ))-K({\mathfrak D}^x(\DGS )) - \mu'(0) \Projector[x] \right) C^x_0 \\ &+ J_{V^x(\DGS )}(D_{x,\DGS }^{\rm imp}(0))- K_{V^x(\DGS )}(D_{x,\DGS }^{\rm imp}(0)) \bigg) + O(\alpha^2),
\end{align*}
where ${\mathfrak L}_x^+$ is the linear response operator introduced in \eqref{eq:linear_response}. Observing that
$$
{C^x_0}^T  \left(J({\mathfrak D}^x(\DGS ))-K({\mathfrak D}^x(\DGS )) \right) C^x_0+ J_{V^x(\DGS )}(D_{x,\DGS }^{\rm imp}(0))- K_{V^x(\DGS )}(D_{x,\DGS }^{\rm imp}(0)) = {C^x_0}^T \left( J(\DGS )-K(\DGS ) \right) C^x_0,
$$
we obtain that
\begin{align} \label{eq:dalphaFhl}
\partial_\alpha \HLMap(0,\DGS )&= - \sum_{x=1}^{\NFrag} \Projector[x] C^x_0  {\mathfrak L}_x^+\left( {C^x_0}^T \left( J(\DGS )-K(\DGS ) \right) C^x_0 - \mu'(0) {\mathfrak p}^x \right) {C^x_0}^T  \Projector[x].
\end{align}
Since the roles of the scalars $\partial_\alpha \mu(0,\DGS )$ in \eqref{eq:dalphaFhlHF} and $\mu'(0)$ in \eqref{eq:dalphaFhl} are simply to ensure charge neutrality, these two scalars are the same. It follows that $\partial_\alpha \HLMap_{\rm HF}(0,\DGS )=\partial_\alpha \HLMap(0,\DGS )$, which allows us to deduce from \eqref{eq:PDMET1}-\eqref{eq:PHF1} that $P_1^{\rm DMET}=P_1^{\rm HF}$. Finally, we obtain that $D_1^{\rm DMET}=D_1^{\rm HF}$ by differentiating the relations
$$
D^{\rm DMET}(\alpha) = \LLMap(\alpha,P^{\rm DMET}(\alpha)) \quad \mbox{and} \quad D^{\rm HF}(\alpha) = \LLMap(\alpha,P^{\rm HF}(\alpha)),
$$
and using the fact that $P_1^{\rm DMET}=P_1^{\rm HF}$.

\section*{Acknowledgements} This project has received funding from the European Research Council (ERC) under the European Union's Horizon 2020 research and innovation programm (grant agreement EMC2 No 810367) and from the Simons Targeted Grant Award No. 896630. 
Moreover, it was partially supported by the Air Force Office of Scientific Research under the award number FA9550-18-1-0095 and by the Simons Targeted Grants in Mathematics and Physical Sciences on Moir\'e Materials Magic (F.M.F.).
The authors thank Emmanuel Fromager, Lin Lin, and Solal Perrin-Roussel for useful discussions and comments. Part of this work was done during the IPAM program {\it Advancing quantum mechanics with mathematics and statistics}. 

\appendix 

\newpage

\section{Notation table}
\label{sec:appendix}
The following table collects the main notations in use in this article.
%%%%%%%%%%%%%%%Suggestion of notation to add
%-\Ann and \Ann^\dagger 
%
\begin{table}[ht] 
\begin{center}
\begin{tabular}{|c|c|c|}
\hline
\textbf{Symbol} & \textbf{Meaning} & \textbf{See Eq.} \\ \hline
$\Fock(E)$ & Fermionic Fock space associated with &  \\ & the one-particle state space $E \subset \HSpace$ & \\ \hline
$\HSpace=\R^\DimH$ & One-particle state space of the whole system, $\DimH$ its dimension &\eqref{eq:one_particle_state_space} \\
$\AtomicBasis=(\AtomicVector)_{1 \le \kappa \le \DimH} $ & Canonical basis of $\HSpace$  & \eqref{eq:one_particle_state_space}\\
$\Hamiltonian$ & Hamiltonian of the whole system (op. on ${\rm Fock}({\mathcal H})$) & \eqref{eq:Hamiltonian} \\
$\Hamiltonian_0$ & Non-interacting Hamiltonian of the whole system & \eqref{eq:def_H0} \\
$\Hamiltonian_\alpha$ & Hamiltonian of the whole system for coupling parameter $\alpha$ & \eqref{eq:Halpha}  \\
$\NElec$ & Number of electrons in the system &  \\
$\Grass$ & Set of 1-RDMs associated with $\NElec$-particles Slater states &  \eqref{eq:def_Grass} \\ & (Grassmann manifold $\mathrm{Gr}(\NElec,\DimH)$) &  \\
$\CHGrass$ & Convex hull of  $\Grass$ &  \eqref{eq:def_CHGrass} \\ & (set of mixed-state 1-RDMs with $\NElec$ particles)  &  \\
$\DGS$ & $\NElec$-particle round-state 1-RDM of $\widehat H_0$  & \eqref{eq:GS_DM} \\  
$D_\alpha^{\rm exact}$ & $\NElec$-particle ground-state 1-RDM of $\widehat H_\alpha$   & \\ 
$D_\alpha^{\rm HF}$ & Hartree-Fock $\NElec$-particle ground state 1-RDM of $\widehat H_\alpha$   & \\ 
$\EMF$ & Hartree-Fock energy functional   & \eqref{eq:HF_functional} \\ 
$J$ and $K$ & Coulomb and exchange energy functionals   & \eqref{eq:direct_exchange} \\ 
$h^{\rm HF}(D)$ & Mean-field (Fock) Hamiltonian (op. on $\cH$) & \eqref{eq:Fock_Hamiltonian} \\ \hline 
%$\Partition$ & A partition of $\llbracket 1, \DimH \rrbracket$ \\
%$\FragIndic$ & A set of the partition $\Partition$ \\
$\NFrag$ & Number of fragments  & \\
$\FragSize$ & Number of sites in fragment $x$ &  \\
$\Frag$ & $x$-th fragment subspace,  $\Frag=\Span(\AtomicVector, \kappa \in \FragIndic) \subset \cH$ & \eqref{eq:dec_H1} \\
$\Projector[x]$ & Orthogonal projector on $\Frag$ (op. on $\R^{\DimH \times \DimH}_{\rm sym}$) &  \\
$\MatFrag$ & Matrix of the $\FragSize$ orbitals of fragment $x$ ($\MatFrag \in \R^{\DimH \times \FragSize}$) & \eqref{eq:def_Ex}\\
$\PartitionProj$ & Projector defined by $\PartitionProj(M)= \sum_{x=1}^{\NFrag} \Projector[\Frag] M \Projector[\Frag]$  (op. on $\R^{\DimH \times \DimH}_{\rm sym}$) & \eqref{eq:PartitionProjDefinition}\\
$\mathcal P$ & Convex set of block-diagonal matrices with eigenvalues in $[0,1]$ & \eqref{eq:def_PP} \\
$\cY$ & Space of traceless block-diagonal matrices $\cY \subset \R^{\DimH \times \DimH}_{\rm sym}$ & \eqref{eq:defY} \\
$\Impurity{D}$ & $x$-th impurity space, subspace of $\HSpace$, $\Impurity{D} = \Frag + D\Frag \subset \cH$  & \eqref{eq:ImpuritySubspaceDefinition}  \\
$C^x(D)$, $\widetilde C^x(D)$ & Matrices in $\R^{\DimH \times 2\FragSize}$ defining orthonormal bases of $\Impurity{D}$ & \eqref{eq:matrix_C}, \eqref{eq:matrix_C_tilde} \\
$\ImpurityHamiltonian{D}$ & $x$-th impurity Hamiltonian (op. on ${\rm Fock}(\Impurity{D})$) & \eqref{eq:def_impurity_Ham_1}, \eqref{eq:impurity_Hamiltonian} \\
$R$ & 4-point DMET linear response function (op. on $\cY$) & \eqref{eq:def_R_formal}, \eqref{eq:def_R_math} \\
 \hline
$\LLMap$, resp. $\LLMap_\alpha$ & Low-level map for $\Hamiltonian$, resp. $\Hamiltonian_\alpha$  & \eqref{eq:LowLevelMapDefinition} \\
$\HLMap$, resp. $\HLMap_\alpha$ & High-level map for $\Hamiltonian$, resp. $\Hamiltonian_\alpha$  & \eqref{eq:HLmap}, \eqref{eq:def_Fhl_2_formal} \\
$\ChemicalPotential$ & DMET global chemical potential &  \\
\hline
\end{tabular}
\end{center}
\caption{\label{tab:notations}Collection of the main notations used in the paper.}
\end{table}

\pagebreak
\section{Analysis of the DMET bifurcation for H$_6^{4-}$}
\label{sec:appendixB}

We shall finally proceed with the analysis of the DMET solutions along the two bifurcation paths for H$_6^{4-}$ around $\Theta_3$ (see Section~\ref{sec:H6}). 
To begin with, we calculate the molecular orbitals at $\Theta_3$. 
The molecular orbital energies exhibit two-fold degeneracies resulting from the fact that the $E'$ and $E''$ are irreducible representations of the H$_6^{4-}$ symmetry point group (D$_{\rm 3h}$) are two-dimensional.
For a visual representation of the molecular orbital energies and their corresponding molecular orbitals, see Fig.~\ref{fig:MOsOcc} and Fig.~\ref{fig:MOsVirt}.

\begin{figure}[h!]
     \centering
     \begin{subfigure}[t]{0.45\textwidth}
         \centering
         \includegraphics[width=\textwidth]{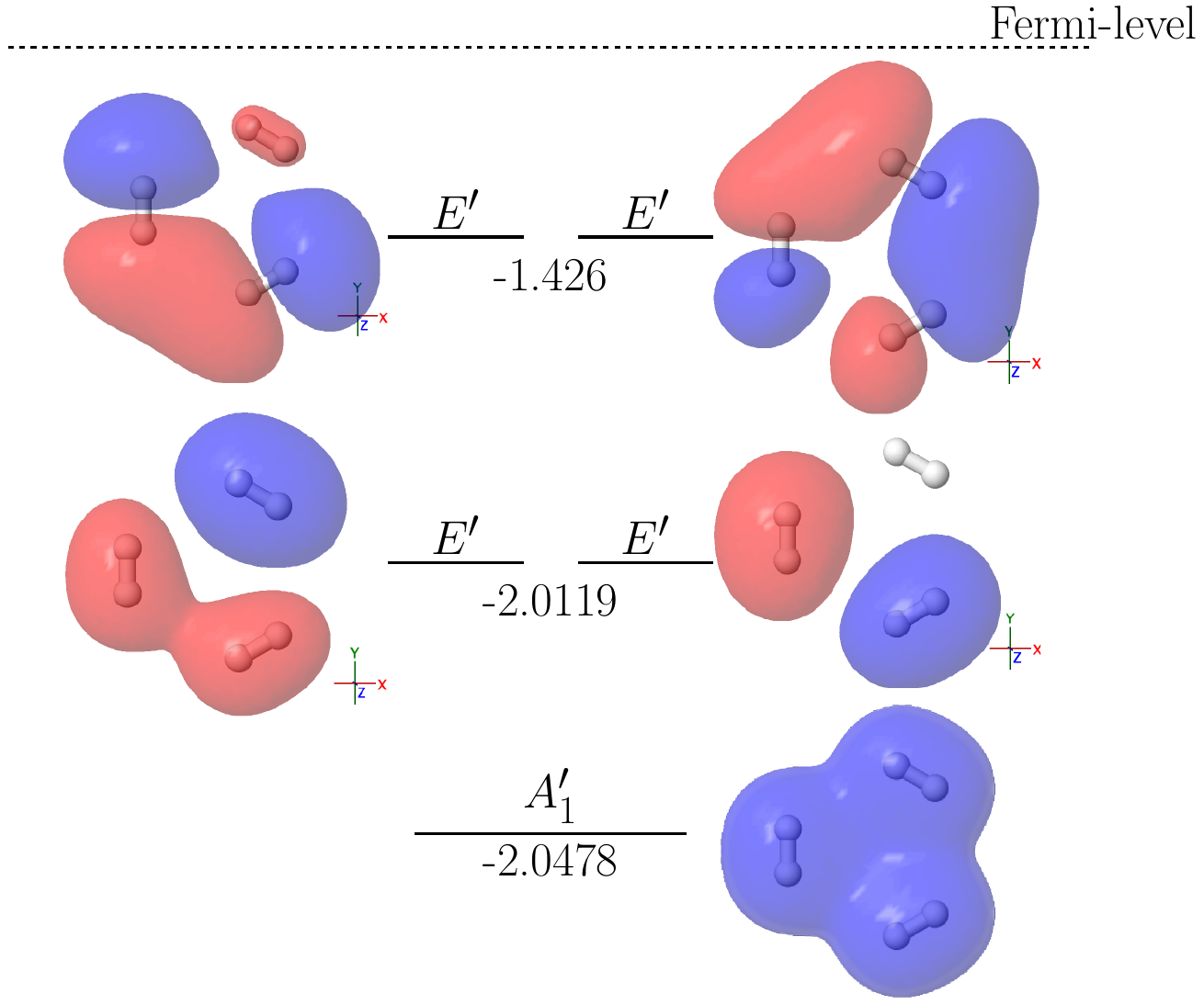}
         \caption{Occupied orbitals}
         \label{fig:MOsOcc}
     \end{subfigure}
     \hfill
     \begin{subfigure}[t]{0.45\textwidth}
         \centering
         \includegraphics[width=\textwidth]{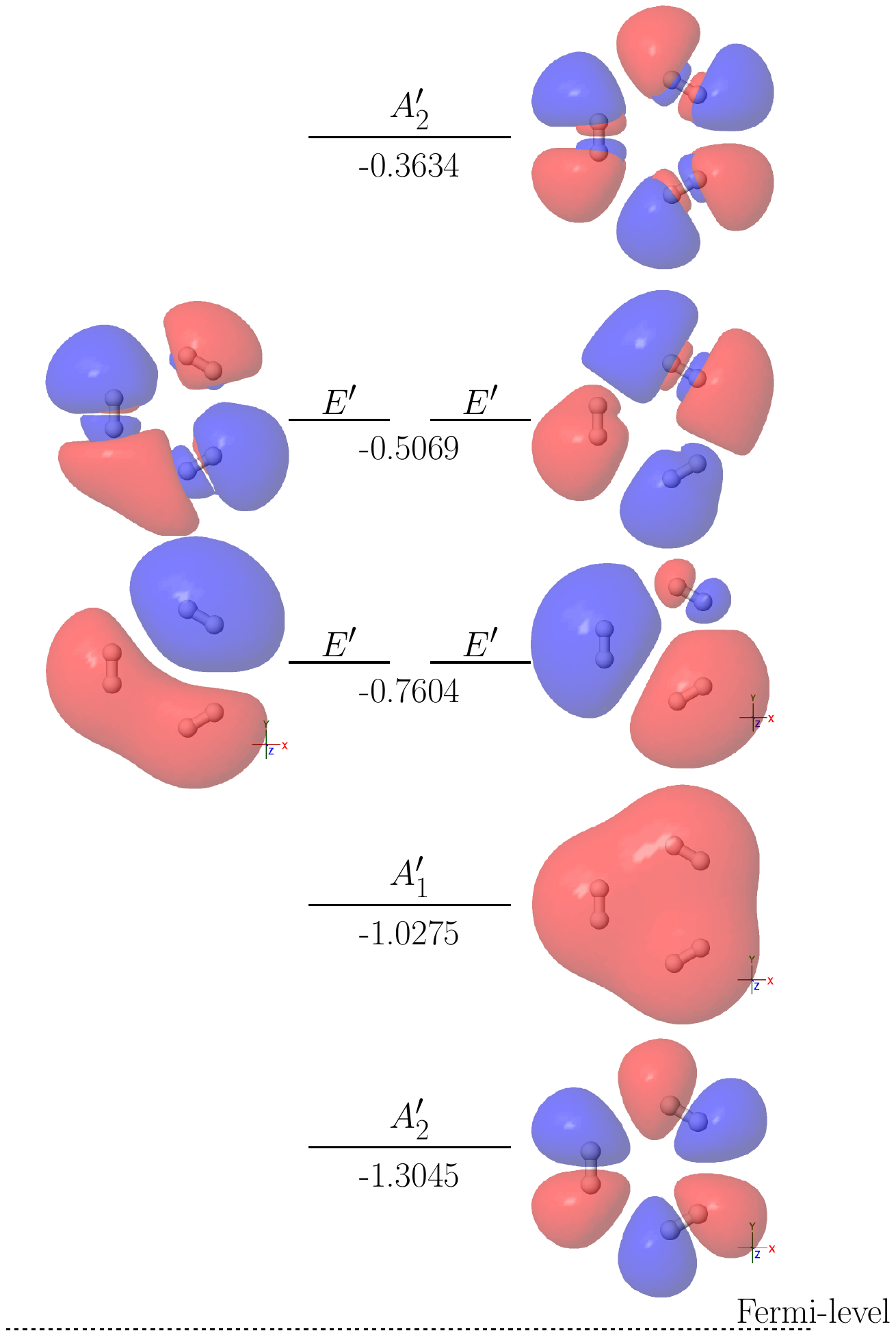}
         \caption{Virtual orbitals}
         \label{fig:MOsVirt}
     \end{subfigure}
     \caption{Depiction of the molecular orbitals, their irreducible representation with respect to the $D_{3h}$ point group symmetry and molecular energies. The left panel shows the occupied molecular orbitals and the right panel shows the virtual molecular orbitals. }
\end{figure}

For the two solutions on the respective bifurcation branches, $P_0$ and  $P_1$, we compute 
\begin{equation}
P_0(\Theta) - P_1(\Theta)
= (\Theta-\Theta_3)
\left[
\begin{array}{c|c}
0 & Q_{+-} \\
\hline
Q_{-+} & 0 \\
\end{array}
\right] + o(\Theta-\Theta_3),
\end{equation}
where $Q_{-+}= Q_{+-}^\top \in \mathbb{R}^{7 \times 5}$.
From the matrix $Q_{-+}$ we deduce ``excitation'' patterns that give physical insight into the different branches.
The numerical values of $Q_{-+}$ are given by 
\begin{equation}
Q_{-+}
=
\left[
\begin{array}{ccccc}
 0    &  0    &  0    &  0    &  0 \\    
-0.0004&  0    &  0    &  0    &  0   \\ 
 0    &  0.0001&  0    &  0.0002&  0    \\
 0    &  0    & -0.0001&  0    & -0.0002\\
 0    &  0.0001&  0    &  0.0001&  0    \\
 0    &  0    & -0.0001&  0    & -0.0001\\
 0    &  0    &  0    &  0    &  0    \\
\end{array}
\right]
\end{equation}
Upon inspecting $Q_{-+}$, we observe the following ``excitation'' pattern: 
The first molecular orbital (A$_1$' symmetry) is rotated in the direction of the seventh molecular orbital (A$_1$' symmetry), while the $4$-dimensional space generated by the second to fifth molecular orbitals ($E'$ symmetry) is tilted according to directions which are linear combinations of the eighth to eleventh molecular orbitals ($E'$ symmetry).
We summarize this ``excitation'' pattern in Fig.~\ref{fig:ExcitationPattern1}

\begin{figure}[h!]
    \centering
    \includegraphics[width = 0.5\textwidth]{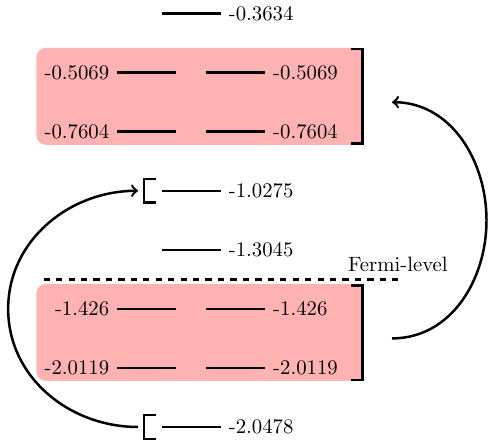}
    \caption{Molecular energies and ``excitation'' patterns concluded from $Q_{-+}$}
    \label{fig:ExcitationPattern1}
\end{figure}

We see that the pair of degenerate occupied orbitals are excited into the pair of degenerate virtual orbitals. 
This block of excitations is highlighted by the red shaded area in Fig.~\ref{fig:ExcitationPattern1} .
A more detailed depiction of the excitations between the red-shaded areas is given in Fig.~\ref{fig:ExcitationPattern2}.

\begin{figure}[h!]
    \centering
    \includegraphics[width = \textwidth]{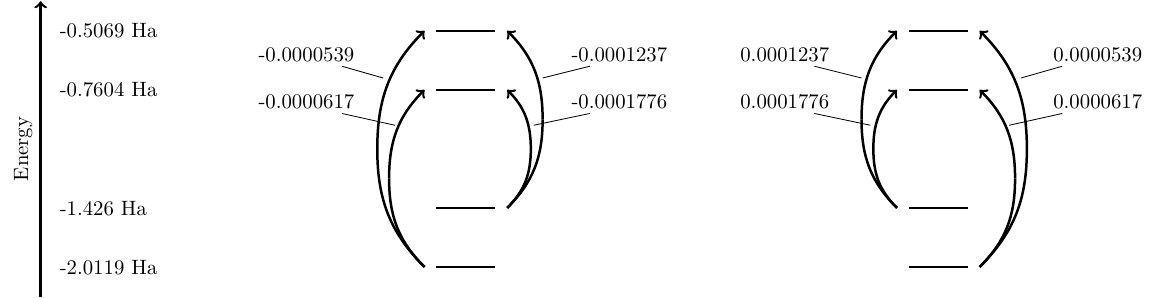}
    \caption{Excitation patterns concluded from $Q_{-+}$ for symmetric and anti-symmetric molecular orbitals respectively.}
    \label{fig:ExcitationPattern2}
\end{figure}

\newpage
% This uses numerical values for the references in the article while sorting them alphabetically, I think that makes it a bit easier to read
\bibliographystyle{IEEEtranS}
\bibliography{biblio}
\end{document}